\documentclass[a4paper,11pt]{article}

\usepackage[a4paper,hmargin=2.8cm,vmargin=3cm]{geometry}
\usepackage[utf8x]{inputenc} 
\usepackage[english]{babel}
\usepackage{amsmath,amsthm,amsfonts,amssymb,mathrsfs,dsfont}
\usepackage{authblk}
\usepackage[bookmarksnumbered=true]{hyperref}  
\usepackage[capitalise]{cleveref}
\crefname{equation}{}{}     

\usepackage{graphicx}
\usepackage{tikz-cd} 
\usepackage[nottoc]{tocbibind} 
\usetikzlibrary{positioning,calc,fadings,decorations.pathreplacing,shadings} 
 
\newcommand\pgfmathsinandcos[3]{%
  \pgfmathsetmacro#1{sin(#3)}%
  \pgfmathsetmacro#2{cos(#3)}%
}
\newcommand\LongitudePlane[3][current plane]{%
  \pgfmathsinandcos\sinEl\cosEl{#2} 
  \pgfmathsinandcos\sint\cost{#3} 
  \tikzset{#1/.style={cm={\cost,\sint*\sinEl,0,\cosEl,(0,0)}}}
}

\newcommand\LatitudePlane[3][current plane]{%
  \pgfmathsinandcos\sinEl\cosEl{#2} 
  \pgfmathsinandcos\sint\cost{#3} 
  \pgfmathsetmacro\yshift{\RadiusSphere*\cosEl*\sint}
  \tikzset{#1/.style={cm={\cost,0,0,\cost*\sinEl,(0,\yshift)}}} %
}

\newcommand\DrawLongitudeArc[4][black]{
  \LongitudePlane{\angEl}{#2}
  \tikzset{current plane/.prefix style={scale=1}}
  \pgfmathsetmacro\angVis{atan(sin(#2)*cos(\angEl)/sin(\angEl))} %
  \pgfmathsetmacro\angA{mod(max(\angVis,#3),360)} %
  \pgfmathsetmacro\angB{mod(min(\angVis+180,#4),360} %
  \draw[current plane,#1]  (\angA:\RadiusSphere) arc (\angA:\angB:\RadiusSphere);
}%
\newcommand\DrawLatitudeCircle[2][1]{
  \LatitudePlane{\angEl}{#2}
  \tikzset{current plane/.prefix style={scale=1}}
  \pgfmathsetmacro\sinVis{sin(#2)/cos(#2)*sin(\angEl)/cos(\angEl)}
  \pgfmathsetmacro\angVis{asin(min(1,max(\sinVis,-1)))}
  \draw[current plane] (\angVis:\RadiusSphere) arc (\angVis:-\angVis-180:\RadiusSphere);
}

\newcommand\DrawLatitudeArc[4][black]{
  \LatitudePlane{\angEl}{#2}
  \tikzset{current plane/.prefix style={scale=1}}
  \pgfmathsetmacro\sinVis{sin(#2)/cos(#2)*sin(\angEl)/cos(\angEl)}
  \pgfmathsetmacro\angVis{asin(min(1,max(\sinVis,-1)))}
  \pgfmathsetmacro\angA{max(min(\angVis,#3),-\angVis-180)} %
  \pgfmathsetmacro\angB{min(\angVis,#4)} %

  \draw[current plane,#1] (\angA:\RadiusSphere) arc (\angA:\angB:\RadiusSphere);
}

\newtheorem{theorem}{Theorem}[section]
\newtheorem{cor}[theorem]{Corollary}
\newtheorem{prop}[theorem]{Proposition}
\newtheorem{lemma}[theorem]{Lemma}

\theoremstyle{remark}
\newtheorem*{rem}{Remark}
\newtheorem*{rems}{Remarks}

\renewcommand{\theequation}{\thesection.\arabic{equation}}

\numberwithin{equation}{section}

\newcommand{\fock}{\mathcal{F}}		
\newcommand{\di}{{\textnormal{d}}}		

\newcommand{\Tbb}{\mathbb{T}}

\newcommand{\Ecal}{\mathcal{E}}
\newcommand{\Hbb}{\mathbb{H}}

\newcommand{\Ncal}{\mathcal{N}}		
\newcommand{\Hcal}{\mathcal{H}}		
\newcommand{\Ical}{\mathcal{I}}
\newcommand{\Ikp}{\Ical_{k}^{+}}

\newcommand{\Ik}{\Ical_{k}}

\newcommand{\F}{\textnormal{F}}

\newcommand{\Ocal}{\mathcal{O}}		

\newcommand{\Nbb}{\mathbb{N}}		
\newcommand{\Zbb}{\mathbb{Z}}

\renewcommand{\Im}{\operatorname{Im}} 	
\newcommand{\id}{\mathbb{I}}
\newcommand{\norm}[1]{\|#1\|}	
\newcommand{\tr}{\operatorname{tr}}
\newcommand{\HS}{{\textnormal{HS}}}

\newcommand{\tagg}[1]{ \stepcounter{equation} \tag{\theequation} \label{#1} } 

\newcommand{\Efrak}{\mathfrak{E}}

\newcommand{\north}{\Gamma^{\textnormal{nor}}}

\newcommand{\supp}{\operatorname{supp}}

\def\bR{\mathbb{R}}
\def\bD{\mathbb{D}}
\def\bT{\mathbb{T}}
\def\bN{\mathbb{N}}
\def\bH{\mathbb{H}}
\def\bX{\mathbb{X}}

\def\bZ{\mathbb{Z}}

\def\cA{\mathcal{A}}

\def\cO{\mathcal{O}}
\def\cF{\mathcal{F}}

\def\cI{\mathcal{I}}

\def\cN{\mathcal{N}}
\def\cE{\mathcal{E}}

\def\cH{\mathcal{H}}
\def\cS{\mathcal{S}}

\def\eps{\varepsilon}
\def\wt{\widetilde}

\def\norm#1{{\left|\hskip-.05em\left|#1\right|\hskip-.05em\right|}}

\def\Im{\textrm{Im}\,}

\definecolor{lightblue}{rgb}{0, 0.33, 0.71}

\newcommand{\Xbb}{\mathbb{X}}
\newcommand{\Dbb}{\mathbb{D}}
\newcommand{\BF}{B_\textnormal{F}}
\newcommand{\BFc}{B_\textnormal{F}^c}
\newcommand{\B}{\textnormal{B}}
\newcommand{\kF}{k_\textnormal{F}}
\newcommand{\op}{\textnormal{op}}

\title{Correlation Energy of a Weakly Interacting Fermi Gas with Large Interaction Potential}

\author[1]{Niels Benedikter}
\author[2]{Marcello Porta} 
\author[3]{Benjamin Schlein} 
\author[4]{Robert Seiringer}
\affil[1]{Universit\`a degli Studi di Milano, Dipartimento di Matematica, Via Cesare Saldini 50, 20133 Milano, Italy\\ORCID: \href{https://orcid.org/0000-0002-1071-6091}{0000-0002-1071-6091}, e--mail: \href{mailto:niels.benedikter@unimi.it}{niels.benedikter@unimi.it}}
\affil[2]{SISSA, Mathematics Area, Via Bonomea 265, 34136 Trieste, Italy}
\affil[3]{Institute of Mathematics, University of Zurich, Winterthurerstrasse 190, 8057 Zurich, Switzerland}
\affil[4]{IST Austria, Am Campus 1, 3400 Klosterneuburg, Austria}

\begin{document}
\maketitle
\begin{abstract}
Recently the leading order of the correlation energy of a Fermi gas in a coupled mean--field and semiclassical scaling regime has been derived, under the assumption of an interaction potential with a small norm and with compact support in Fourier space. We generalize this result to large interaction potentials, requiring only $|\cdot| \hat{V} \in \ell^1 (\bZ^3)$. Our proof is based on approximate, collective bosonization in three dimensions. Significant improvements compared to recent work include stronger bounds on non--bosonizable terms and more efficient control on the bosonization of the kinetic energy.   
\end{abstract} 

\tableofcontents

\section{Introduction}
The interacting high--density Fermi gas models a variety of important physical systems, in particular the behavior of electrons in alkali metals. The simplest approximation for the computation of its physical properties is mean--field theory, i.\,e., Hartree--Fock theory. Hartree--Fock theory only includes the minimal amount of quantum correlations unavoidable due to the antisymmetry requirement on the wave function of fermionic many--body systems. In the present paper we consider corrections to the Hartree--Fock energy due to non--trivial quantum correlations (i.\,e., entanglement in the ground state).

According to \cite{BP53}, the dominant effect of correlations on the ground state energy should be described by the \emph{random--phase approximation (RPA)}, which may also be formulated as a partial resummation of the perturbation series \cite{GB57} or as a theory of particle--hole pairs behaving as bosonic quasiparticles \cite{SBFB57}. The latter point of view was recently used by \cite{BNPSS,BNPSS0} (extending the second--order result of \cite{HPR20}) to rigorously prove the validity of the random--phase approximation for the ground state energy, assuming the interaction potential to be small and its Fourier transform to have compact support. \emph{In the present paper, that result is generalized to arbitrarily large interaction potentials without restriction on the support.} Our proof is a refinement of the method of \cite{BNPSS,BNPSS0}, a crucial point of which is to delocalize particle--hole pairs over patches on the Fermi surface, thus circumventing the Pauli principle and justifying the approximate bosonization of particle--hole pairs. This approach leads to a bosonic quasifree effective theory, from which the ground state energy can be computed. 

The further predictions of this bosonic effective theory have been discussed in \cite{Ben20} and it has also been proven to be a good approximation for the time evolution of the Fermi gas \cite{BNPSS}, refining the time--dependent Hartree--Fock approximation derived in \cite{BSS,BPS14,BPS14c,BJPSS16}. An alternative approach to the ground state energy, avoiding delocalization and thus closer in spirit to \cite{SBFB57} has been developed recently in \cite{CHN}: still, also there an averaging over different particle--hole pairs is needed to justify the bosonization. In another context, the low--density Fermi gas, bosonization ideas have been applied by \cite{FGHP21, Gia22a, Gia22b}.

\medskip

Let us turn to the mathematical description of our result. We consider a system of $N$ fermions on the torus $\bT^3 := \bR^3 / (2\pi \bZ^3)$ interacting through a potential $V$. The system is described on the Hilbert space $L^2_\textnormal{a} (\bT^{3N})$, consisting of all $\psi \in L^2 (\bT^{3N})$ that are antisymmetric under exchange of particles,  
\[  \psi (x_{\sigma(1)} , \dots , x_{\sigma (N)}) = \text{sgn} (\sigma) \psi (x_1, \dots , x_N) \]
for all permutations $\sigma \in \cS_N$. The Hamiltonian is the linear self--adjoint operator 
\begin{equation}\label{eq:HN0} 
H_N := \sum_{j=1}^N -\hbar^2 \Delta_{x_j} + \lambda \sum_{i<j}^N V (x_i - x_j) \;.
\end{equation} 
The interaction potential $V$ is assumed to have non--negative Fourier transform $\hat{V} \geq 0$. (For the interaction potential we use the convention that the Fourier transform is $V(x) = \sum_{k \in \Zbb^3} \hat{V}(k) e^{ik\cdot x}$, unlike for the Fourier transform of wave functions which we normalize to be unitary.)
Because of the antisymmetry of the wave functions, the sum of the Laplacians is typically of order $N^{5/3}$, as may be seen most easily from the the non--interacting case $V=0$, where the ground state is a Slater determinant of $N$ plane waves $f_k(x) = (2\pi)^{-3/2} e^{ik\cdot x}$, the momenta $k \in \Zbb^3$ being located in a ball of radius proprtional to $N^{1/3}$. To make both kinetic and potential energy scale extensively (i.\,e., proportionally to the number of particles $N$) we set
\[\hbar := N^{-1/3} \quad \textnormal{and} \quad \lambda := N^{-1}\;.\]
This is interpreted as a mean--field limit coupled to a semiclassical limit with effective Planck constant $\hbar = N^{-1/3} \to 0$ as $N \to \infty$; this scaling limit has been introduced by \cite{NS81,Spo81} to derive the Vlasov equation from many--body quantum mechanics.

We are interested in the ground state energy 
\[ E_N := \inf \text{spec} (H_N) = \inf \left\{ \langle \psi, H_N \psi \rangle : \psi \in L^2_\textnormal{a} (\bT^{3N}) ,\ \| \psi \| = 1 \right\}\;. \]
A first approximation for $E_N$ is the Hartree--Fock energy, defined by restricting the variational problem to Slater determinants, i.\,e.,
\[ E_N^\textnormal{HF} := \inf \Big\{ \langle \psi, H_N \psi \rangle : \psi = \bigwedge_{j=1}^N  u_j  \text{ where } \{ u_j \}_{j=1}^N \text{ is an orthonormal family in $L^2 (\bT^3)$} \Big\}\;. \]
As already mentioned, for the non--interacting case $V=0$, the Hartree--Fock and the many--body ground state energy are attained by the Fermi ball  
\begin{equation}
\psi_\F := \bigwedge_{k \in B_\F} f_k\;,   \label{eq:plane-waves}
\end{equation}
with the plane waves $f_k (x) := (2\pi)^{-3/2} e^{ik\cdot x}$, for $x \in \bT^3$ and $k \in \bZ^3$. Here, the Fermi ball $B_\F$ is a set of $N$ different momenta $p \in \bZ^3$ with $\sum_{p}\lvert p\rvert^2$ as small as possible. To simplify our analysis we assume that the Fermi ball is completely filled and thus uniquely defined, i.\,e., that $B_\F = \{ k \in \bZ^3 : |k| \leq k_\F \}$. This can be achieved by considering a sequence $k_\F \to \infty$ and fixing $N := |B_\F|$ as a function of $k_\F$. We find the relation $k_\F = \kappa N^{1/3}$ between the two parameters, with $\kappa = \kappa_0 + \cO (N^{-1/3})$ and $\kappa_0 := (3/4\pi)^{1/3}$.

Under the assumption of a complete Fermi ball and non--negative $\hat{V}$, it was proven in \cite[Theorem A.1]{BNPSS} that the Hartree--Fock energy $E_N^\textnormal{HF}$ is still attained by the Fermi ball \cref{eq:plane-waves}, even when $V \not = 0$. It follows that 
\begin{equation}\label{eq:HF-en}  E_N^\textnormal{HF} = \langle \psi_\F , H_N \psi_\F \rangle = \sum_{p \in B_\F} \hbar^2 p^2 + \frac{N}{2} \hat{V} (0) - \frac{1}{2N} \sum_{k,k' \in B_F} \hat{V} (k - k') \;. \end{equation}

In this paper we focus on the \emph{correlation energy}, defined as the difference $E_N - E_N^\textnormal{HF}$, due to  many--body interactions among particles. The following theorem, our main result, provides an explicit formula for the dominant order (order $\hbar$) of the correlation energy. 
\begin{theorem}[Main result: RPA correlation energy]\label{thm:main} 
Suppose $V \in L^1 (\Tbb^3)$ with $\hat{V} \geq 0$ and
\[ \sum_{k \in \bZ^3} \hat{V} (k) |k| < \infty \;.  \]
For $k_\F > 0$ let $N := | B_\F | = |\{ k \in \bZ^3 : |k| \leq k_\F \}|$. Then there exists $\alpha > 0$ such that 
\begin{equation}\label{eq:EN-asy} E_N = E_N^\textnormal{HF} + E_N^\textnormal{RPA} + \cO (N^{-1/3-\alpha}) \qquad \text{for $k_\F \to \infty$}\end{equation} 
where the RPA energy formula is
\begin{equation}\label{eq:RPA} E_N^\textnormal{RPA} := \hbar \kappa_0 \sum_{k \in \bZ^3} |k| \left( \frac{1}{\pi} \int_0^\infty \log \left( 1 + 2\pi \kappa_0 \hat{V} (k) \Big(1-\lambda \arctan \big(\frac{1}{\lambda}\big) \Big) \right) \di\lambda - \frac{\pi}{2} \kappa_0 \hat{V} (k) \right). \end{equation}
\end{theorem}
\begin{rems}
 \begin{enumerate}
  \item Unlike the result of \cite{BNPSS}, where $\lVert V\rVert_{\ell^\infty}$ was assumed to be small, here we do not assume smallness of the interaction potential.
  \item A further generalization is given in \cref{sec:appendix}: there, the upper bound of \cref{eq:EN-asy} is shown to hold assuming only $\hat{V} \geq 0$ and 
$\sum_{k \in \bZ^3} |k| \hat{V} (k)^2 < \infty$. Thanks to only the second power of the potential appearing, this almost covers the Coulomb potential. While our paper was under review, a new upper bound for the correlation energy has been established in \cite{CHN2} for square integrable potentials; this includes potentials with Coulomb singularity. In this case, an additional second order contribution to the exchange energy, which is part of the error in our setting, becomes relevant.
 \end{enumerate}
\end{rems}
In the next section we will introduce the correlation Hamiltonian which describes corrections to Hartree--Fock theory. In \cref{sec:strategy} we give a heuristic introduction to the bosonization method by which the correlation Hamiltonian can be approximately diagonalized. The remaining sections are dedicated to the steps of the rigorous implementation of this strategy, culminating in the proof of \Cref{thm:main} in \cref{sec:conclusion}.

\section{Correlation Hamiltonian}\label{sec:corrH}
As the first step to the proof of \cref{thm:main}, we apply a particle--hole transformation to the Hamiltonian, by which we obtain the \emph{correlation Hamiltonian} which describes only the corrections to mean--field (Hartree--Fock) theory. This is an exact computation not involving any approximation.

We use second quantization on the fermionic Fock space $\cF = \bigoplus_{n \geq 0} L^2 (\bT^3)^{\otimes_a n}$. On $\cF$, we use the well--known creation and annihilation operators satisfying canonical anticommutation relations, namely for all momenta $p,q \in \bZ^3$ we have 
\begin{equation}\label{eq:CAR}  \{ a_p , a_q^* \} = \delta_{p,q}\; , \qquad \{ a_p, a_q \} = \{ a_p^* , a_q^* \} = 0\; . \end{equation}As a simple consequence of \cref{eq:CAR}, we find the operator norms $\| a_p^* \|_\op \leq 1$ and $\| a_p \|_\op \leq 1$ for all $p \in \bZ^3$. We define the vacuum vector $\Omega = ( 1, 0, 0, \dots ) \in \cF$ and the number--of--fermions operator $\cN = \sum_{p \in \bZ^3} a_p^* a_p$. We extend the Hamiltonian \cref{eq:HN0} to the full Fock space $\cF$ setting 
\begin{equation}\label{eq:ham-fock} \cH_N := \sum_{p \in \bZ^3} \hbar^2 p^2 a_p^* a_p + \frac{1}{2N} \sum_{k,p,q \in \bZ^3} \hat{V} (k) a_{p+k}^* a_{q-k}^* a_q a_p \;.\end{equation} 
The restriction of $\cH_N$ to the $N$--particle sector $L^2_\textnormal{a} (\bT^{3N}) \subset \cF$ coincides with \cref{eq:HN0}. 

To analyse the correlation energy $E_N - E_N^\textnormal{HF}$, it is convenient to factor out the Fermi ball \cref{eq:plane-waves} and focus on its excitations. This is achieved through a \emph{particle--hole transformation} $R_\F : \cF \to \cF$ defined by
\begin{equation} \label{eq:RaR} R_\F^* a_p^* R_\F := \left\{ \begin{array}{ll} a_p^* \quad \text{if } p \in B_\F^c \\ a_p \quad \text{if } p \in B_\F \end{array} \right.
\qquad \textnormal{and} \qquad  R_\F \Omega := \bigwedge_{p \in B_\F} f_p = \psi_\F\;.
\end{equation} 
One has $R_\F = R_\F^* = R_\F^{-1}$. With \cref{eq:RaR} we find 
\[ R_\F^* \cN R_\F = \sum_{p \in B_\F} a_p a_p^* + \sum_{p \in B_\F^c} a_p^* a_p = N - \sum_{p\in B_\F} a_p^* a_p + \sum_{p \in B_\F^c} a_p^* a_p = N - \cN_\textnormal{h} + \cN_\textnormal{p} \]
where we defined the number--of--holes operator $\cN_\textnormal{h} := \sum_{h \in B_\F} a_h^* a_h$ and the number--of--particles operator $\cN_\textnormal{p} :=  \sum_{p \in B_\F^c} a_p^* a_p$. This shows that the $N$--particle sector $L^2_\textnormal{a} (\bT^{3N}) \subset \fock$ is the image under $R_\F$ of the eigenspace of $\cN_\textnormal{h} - \cN_\textnormal{p}$ associated with the eigenvalue $0$ (and thus $R_\F$ defines a unitary map from the eigenspace $\chi (\cN_\textnormal{h} - \cN_\textnormal{p} = 0) \cF$ to $L^2_\textnormal{a} (\bT^{3N})$). 

We introduce the correlation Hamiltonian $\Hcal_\textnormal{corr}$ by conjugating $\cH_N$ with $R_\F$ and subtracting the energy of the Fermi ball (which, as already noted in \cite[Theorem~A.1]{BNPSS}, in our scaling limit and with $\hat{V} \geq 0$ equals the Hartree--Fock ground state energy). With \cref{eq:RaR} and the canonical anticommutation relations \cref{eq:CAR}, a lengthy but straightforward computation leads to the \emph{correlation Hamiltonian} 
\begin{equation}\label{eq:corr}  \cH_\textnormal{corr} := R_\F^* \cH_N R_\F - E_N^\textnormal{HF} = \bH_0 + Q_\textnormal{B} + \cE_1 + \cE_2  + \bX \end{equation} 
with the main terms  
\begin{equation}\label{eq:RHR-main} \begin{split} 
\bH_0 & := \sum_{p \in \bZ^3} e(p) \, a_p^* a_p \; ,  \qquad \text{with  } e(p) := |\hbar^2 p^2 - \kappa^2| \; , \\
Q_\textnormal{B} & := \frac{1}{N} \sum_{k \in \bZ^3} \hat{V} (k) \left( b^* (k) b (k) + \frac{1}{2} \left( b^* (k) b^* (-k) + b(-k) b (k) \right) \right) \end{split} \end{equation}
and the error terms 
\begin{equation}\label{eq:errors} \begin{split}
\Xbb & := - \frac{1}{2N} \sum_{k \in \Zbb^3} \hat{V}(k) \bigg( \sum_{p \in \BFc \cap (\BF +k)} a^*_p a_p  + \sum_{h \in \BF \cap (\BFc-k)} a^*_{h} a_{h}\bigg)\;, \\
\cE_1 & := \frac{1}{2N} \sum_{k \in \bZ^3} \hat{V} (k) d^* (k) d(k) \;, \\ 
\cE_2 & := \frac{1}{2N} \sum_{k \in \bZ^3} \hat{V} (k)  \left[ d^* (k) b (-k)  + \text{h.c.} \right] \;.
\end{split} \end{equation} 
Here we defined the delocalized particle--hole pair creation and annihilation operators 
\begin{equation} \label{eq:bb} b^* (k) := \sum_{p \in B_\F^c \cap (B_\F + k)} a_p^* a_{p-k}^* , \qquad  b (k) := \sum_{p \in B_\F^c \cap (B_\F + k)} a_{p-k} a_{p} \end{equation} 
and the non--bosonizable operators 
\begin{equation}\label{eq:dd}  d^* (k) := \sum_{p \in B_\F^c \cap (B_\F^c + k)} a_p^* a_{p-k} - \sum_{h \in B_\F \cap (B_\F -k)} a_h^* a_{h+k} \;, \end{equation} 
satisfying $d^* (k) = d (-k)$ for all $k \in \Zbb^3$.

\medskip

To prove \cref{thm:main}, we improve the bosonization method introduced in \cite{BNPSS0} for the upper bound and show that
\[\inf_{\substack{\psi \in \fock: \norm{\psi} =1\\(\Ncal_\textnormal{p}-\Ncal_\textnormal{h})\psi = 0}} \langle \psi, \Hcal_\textnormal{corr} \psi\rangle = E^\textnormal{RPA}_N + \Ocal(N^{-1/3-\alpha})\;.\]

\section{Strategy of the Proof: Approximate Bosonization}\label{sec:strategy}

The key idea is to derive, from the fermionic correlation Hamiltonian \cref{eq:corr}, a quadratic, approximately\footnote{With \emph{approximate} bosonization we refer to the fact that we construct operators that only up to an error term satisfy canonical commutator relations; this is in contrast to certain one--dimensional fermionic systems \cite{ML65} and spin systems \cite{CG12,CGS15,Ben17,NS19} which can be expressed through operators that satisfy the canonical commutator relations exactly.} bosonic, Hamiltonian which can be approximately diagonalized by a Bogoliubov transformation to obtain the ground state energy.

The starting point is the observation that the particle--hole pair operators behave approximately as bosonic creation and annihilation operators, i.\,e., they approximately satisfy canonical commutator relations:
\[[b^*(k),b^*(l)] = 0 = [b(k),b(l)]\;, \qquad [b(k),b^*(l)] = \textnormal{const.}\times(\delta_{k,l} + \textnormal{lower order})\;.\]
Thus $Q_\B$ can be understood as an approximately bosonic quadratic Hamiltonian. The terms $\Xbb$, $\Ecal_1$, and $\Ecal_2$ do not have a bosonic interpretation and are going to be estimated as smaller errors. It remains to bosonize the kinetic energy $\Hbb_0$. Because this step requires us to linearize the dispersion relation, we need to localize of the pair operators to patches $B_\alpha$, i.\,e., to $M$ small regions covering a shell around the Fermi sphere in momentum space (see \cref{fig:blub} for an illustration of the patch decomposition we have in mind; eventually the number of patches $M$ will be chosen to tend to infinity as $N\to \infty$):
\begin{equation}\label{eq:pairop}
b^*(k) \simeq \sum_{\alpha=1}^M n_\alpha(k) b^*_\alpha(k)\;, \qquad b^*_\alpha(k) := \frac{1}{n_\alpha(k)} \sum_{\substack{p\colon p \in B_\F^c \cap B_\alpha \\ p-k\in B_\F \cap B_\alpha}}  a^*_p a^*_{p-k}\;,
\end{equation}
with a normalization constant $n_\alpha(k)$ so that the one--pair states $b^*_\alpha(k)\Omega$ have norm one. 
\begin{figure}\centering
\begin{tikzpicture}[scale=0.7]
\def\RadiusSphere{4}
\def\angEl{20}
\def\angAz{-20}

\filldraw[ball color = white] (0,0) circle (\RadiusSphere);

\DrawLatitudeCircle[\RadiusSphere]{75+2}
\foreach \t in {0,-50,...,-250} {
  \DrawLatitudeArc{75}{(\t+50-4)*sin(62)}{\t*sin(62)}
 \DrawLongitudeArc{\t*sin(62)}{50+2}{75}
 \DrawLongitudeArc{(\t-4)*sin(62)}{50+2}{75}
  \DrawLatitudeArc{50+2}{(\t+50-4)*sin(62)}{\t*sin(62)}
 }
 \foreach \t in {0,-50,...,-300} {
   \DrawLatitudeArc{50}{(\t+50-4)*sin(37)}{\t*sin(37)}
 \DrawLongitudeArc{\t*sin(37)}{25+2}{50}
  \DrawLongitudeArc{(\t-4)*sin(37)}{25+2}{50}
   \DrawLatitudeArc{25+2}{(\t+50-4)*sin(37)}{\t*sin(37)}
 }
 \DrawLatitudeArc{50}{(-300-4)*sin(37)}{-330*sin(37)}
 \foreach \t in {0,-50,...,-450} {
    \DrawLatitudeArc{25}{(\t+50-4)*sin(23)}{\t*sin(23)}
 \DrawLongitudeArc{\t*sin(23)}{00+2}{25}
 \DrawLongitudeArc{(\t-4)*sin(23)}{00+2}{25}
 \DrawLatitudeArc{00+2}{(\t+50-4)*sin(23)}{\t*sin(23)}
 }
     \DrawLatitudeArc{25}{(-450-4)*sin(23)}{-500*sin(23)}

\fill[black] (0,3.75) circle (.07cm);

\fill[black] (1.72,3.08) circle (.07cm);
\fill[black] (.76,2.73) circle (.07cm);
\fill[black] (-.66,2.73) circle (.07cm);
\fill[black] (-1.73,3.04) circle (.07cm);

\fill[black] (2.25,1.5) circle (.07cm);
\fill[black] (.8,1.2) circle (.07cm);
\fill[black] (-.85,1.22) circle (.07cm);
\fill[black] (-2.27,1.5) circle (.07cm);
\fill[black] (-3.09,1.97) circle (.07cm);
\fill[black] (3.09,1.97) circle (.07cm);

\fill[black] (2.57,-.15) circle (.07cm);
\fill[black] (1.43,-.37) circle (.07cm);
\fill[black] (.155,-.48) circle (.07cm);
\fill[black] (-1.17,-.41) circle (.07cm);
\fill[black] (-2.35,-.2) circle (.07cm);
\fill[black] (-3.26,0.1) circle (.07cm);
\fill[black] (-3.79,.55) circle (.07cm);
\fill[black] (3.37,.18) circle (.07cm);
\fill[black] (3.85,.57) circle (.07cm);

\end{tikzpicture}
 \caption{Decomposition of (a shell around) the Fermi surface into patches. The vectors $\hat{\omega}_\alpha$ (marked with dots) are the patch centers. The decomposition of the southern half sphere is obtained through reflection by the origin. See \cite{BNPSS0} for the details of the construction.}\label{fig:blub}
\end{figure}
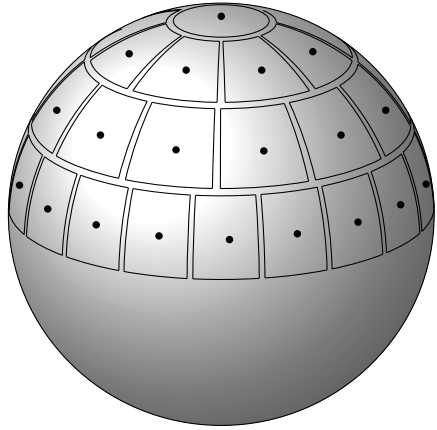
There is a catch here: the sum over pairs in \cref{eq:pairop} is only non--empty if the relative momentum $k$ is pointing outward from the Fermi ball, so for about half of the possible values of $\alpha$ the operators $b^*_\alpha(k)$ vanish. To be sure that many particle--hole pairs contribute to the sum defining $b^*_\alpha(k)$, we introduce a cutoff by defining the index set
\[\Ical_k^+ := \left\{\alpha \in \{1,2,\ldots,M\}: k \cdot \hat{\omega}_\alpha \geq N^{-\delta} \right\}\]
(with $\delta > 0$ to be optimized at the end) and combine the retained $b^*_\alpha(k)$--operators into
\[c^*_\alpha(k) := \left\{ \begin{array}{ll}
                            b^*_\alpha(k) & \textnormal{for }\alpha \in \Ical_k^+ \\
                            b^*_\alpha(-k) & \textnormal{for }\alpha \in \Ical_{-k}^+ \;.
                           \end{array} \right.
\]
These operators again behave approximately bosonic in the sense that
\begin{equation}\label{eq:appr_ccr} [c^*_\alpha(k), c^*_\beta(l)] = 0 = [c_\alpha(k),c_\beta(l)]\;, \quad [c_\alpha(k),c^*_\beta(l)] = \delta_{\alpha,\beta} \left( \delta_{k,l} + \mathcal{O}\left( \frac{\Ncal}{n_\alpha(k)^2} \right) \right)\;. \end{equation}
This provides important intuition on how to make the approximate bosonization rigorous: because $n_\alpha(k)^2$ counts the number of particle--hole pairs of relative momentum $k$ in patch $B_\alpha$, we need the size of the patches to be sufficiently big and we need to bound the number of excitations counted by $\Ncal$ in states close to the ground state.

By virtue of the localization to patches we can linearize the dispersion relation $e(p)$ locally in every patch, and thus find (the computation here shown for the case $\alpha \in \Ikp$)
\begin{equation}\label{eq:HDcomm}
\begin{split}[\Hbb_0, c^*_\alpha(k)] & = \frac{1}{n_\alpha(k)} \sum_{\substack{p\colon p \in B_\F^c \cap B_\alpha \\ p-k\in B_\F \cap B_\alpha}}  \left( e(p) - e(p-k) \right) a^*_p a^*_{p-k}\\
& = \frac{1}{n_\alpha(k)} \sum_{\substack{p\colon p \in B_\F^c \cap B_\alpha \\ p-k\in B_\F \cap B_\alpha}}  \hbar^2 \left( 2 p\cdot k - \lvert k\rvert^2 \right) a^*_p a^*_{p-k} \\
& \simeq \frac{1}{n_\alpha(k)} \sum_{\substack{p\colon p \in B_\F^c \cap B_\alpha \\ p-k\in B_\F \cap B_\alpha}}  2 \hbar^2 \, \omega_\alpha \cdot k\, a^*_p a^*_{p-k}\\
& \simeq [\Dbb_\B, c^*_\alpha(k)]
\end{split}\end{equation} 
if we introduce the quadratic approximately bosonic operator
\[\Dbb_\B = 2 \kappa \hbar \sum_{k \in \north} \sum_{\alpha =1}^M |k \cdot \hat{\omega}_\alpha| \, c_\alpha^* (k) c_\alpha (k)\;.\]
While the substitution of $\Hbb_0$ by $\Dbb_\textnormal{B}$ has here been motivated only in commutators with almost bosonic operators, a key step of our analysis is to justify this step also on general states close to the ground state. This step is explained in \cref{eq:Tinv} to \cref{eq:Tinv_end}.

Our further goal is to approximately (to order $\hbar$, the dominant contribution of the correlation energy) diagonalize the bosonic quadratic Hamiltonian $\Dbb_\B + Q_\B$ by an approximately bosonic Bogoliubov transformation $T$, allowing us to read off the correlation energy. Given a state $\psi \in \mathcal{F}$ such that $(\mathcal{N}_{\text{p}} - \mathcal{N}_{\text{h}}) \psi = 0$ (think of the ground state of $\Hcal_\text{corr}$), and setting $\xi := T^{*} \psi$, we write
\begin{equation}\label{eq:Ttrasfo}
\begin{split}
\langle \psi, \mathcal{H}_{\text{corr}} \psi \rangle &= \langle T \xi, \mathcal{H}_{\text{corr}} T \xi  \rangle \\
&= \langle T\xi, ( \mathbb{D}_{\text{B}} + Q_{\text{B}} ) T \xi \rangle + \langle T\xi, ( \mathbb{H}_{0} -  \mathbb{D}_{\text{B}} ) T\xi \rangle + \langle T\xi, ( \mathbb{X} + \mathcal{E}_{1} + \mathcal{E}_{2} ) T\xi \rangle\;.
\end{split}
\end{equation}
Through a suitable choice of the Bogoliubov kernel $K (k)$ (a matrix indexed by the patch labels), the approximate Bogoliubov transformation 
\begin{equation}
T = \exp \Bigg( \frac{1}{2} \sum_{k \in \north} \sum_{\alpha,\beta \in \cI^{+}_k \cup \cI^{+}_{-k}} K (k)_{\alpha,\beta} \, c_\alpha^* (k) c_\beta^* (k) - \text{h.c.} \Bigg)
\end{equation}
diagonalizes approximately the quadratic Hamiltonian $\mathbb{D}_{\text{B}} + Q_{\text{B}}$. On states with few particles (ie. with few excitations of the Fermi sea), we find as suggested by exact bosonic Bogoliubov theory that
\begin{equation} \label{eq:diag-EH}
\langle T\xi, ( \mathbb{D}_{\text{B}} + Q_{\text{B}} ) T \xi \rangle \simeq E^{\text{RPA}}_{N} + \langle \xi, \mathcal{H}^{\text{exc}}_{\text{B}} \xi \rangle\;,
\end{equation}
with the intended $E^{\text{RPA}}_{N}$ as in \cref{eq:RPA}, and for the description of the possible bosonic excitation one obtains an effective Hamiltonian of the form
\begin{equation}\label{eq:frakK}
\mathcal{H}^{\text{exc}}_{\text{B}} = \sum_{k \in \north} \sum_{\alpha,\beta \in \cI^{+}_k \cup \cI^{+}_{-k}} 2 \hbar \kappa \lvert k\rvert \mathfrak{K}(k)_{\alpha,\beta} c^*_\alpha(k) c_\beta(k) \geq 0  \;.
\end{equation} 

To make these heuristics rigorous, apart from controlling the bosonic approximation (arising from the neglect of the error term in \cref{eq:appr_ccr}) in the bosonic Bogoliubov diagonalization, we need to estimate the second and the third terms in \cref{eq:Ttrasfo}. There are two obstacles. One is to give a meaning to the heuristics $\Hbb_0 \simeq \Dbb_\B$, which, a priori, holds only as in \cref{eq:HDcomm}, at the level of commutators with the approximately bosonic operators. The other is to control the non--bosonizable term $\Ecal_1$ and the term $\Ecal_2$ which couples almost bosonic $c$--operators to non--bosonizable $d$--operators. (The exchange term $\Xbb$ instead can be controlled by more elementary estimates.)

Both problems were solved in \cite{BNPSS} under the assumption that the interaction potential $V$ is small and compactly supported in Fourier space. In the present work we overcome these limitations and prove the validity of the random--phase approximation for a much larger class of interaction potentials. The main achievements of the present paper, compared to \cite{BNPSS0, BNPSS}, are the following:

\medskip

\begin{itemize}
\item The combination $\mathbb{H}_{0} - \mathbb{D}_{\text{B}}$ is approximately invariant under conjugation with the approximately bosonic Bogoliubov transformation because its action can be expanded in commutators:
\begin{equation}\label{eq:Tinv}
\langle T\xi, ( \mathbb{H}_{0} -  \mathbb{D}_{\text{B}} ) T\xi \rangle \simeq \langle \xi, ( \mathbb{H}_{0} -  \mathbb{D}_{\text{B}} ) \xi \rangle\;.
\end{equation}
In the proof of the upper bound for the correlation energy, the vector $\xi$ coincides with the vacuum, and the right--hand side is zero. For the lower bound this is not true, and we are left with controlling the negative term $-\mathbb{D}_{\text{B}}$. In \cite{BNPSS}, this was achieved by exploiting the positivity of $\mathcal{H}^{\text{exc}}_{\text{B}}$ in \cref{eq:diag-EH}. More precisely, we proved that 
\[ \langle \xi, \mathcal{H}^{\text{exc}}_{\text{B}} \xi \rangle \geq \langle \xi, \mathbb{D}_\text{B} \xi \rangle - C \| \hat V \|_{1} \langle \xi, \mathbb{H}_{0} \xi \rangle, \]
which, for small potential, is enough to control the r.\,h.\,s.\ of \cref{eq:Tinv}. In the present paper, we need a more refined analysis. In order to compare $\cH_B^\text{exc}$ with $\mathbb{D}_\text{B}$, we need to diagonalize the matrix $\mathfrak{K}(k)_{\alpha,\beta}$ appearing on the r.\,h.\,s.\ of \cref{eq:frakK} (because $\mathbb{D}_\text{B}$ is already expressed through a diagonal matrix). This can be achieved through a second approximately bosonic Bogoliubov transformation having the form 
\begin{equation}
Z = \exp \Bigg( \sum_{k \in \north} \sum_{\alpha,\beta \in \cI^{+}_k \cup \cI^{+}_{-k}} L (k)_{\alpha,\beta} \, c_\alpha^* (k) c_\beta (k) \Bigg)
\end{equation}
for an antisymmetric matrix $L (k)_{\alpha,\beta}$. If $c^*$ and $c$ were bosonic operators, we could write $Z = \exp\big({\sum_{k \in \north} \di\Gamma (L (k))}\big) = \prod_{k \in \north} \Gamma (e^{L (k)})$ (where $\di\Gamma$ and $\Gamma$ are the operators of bosonic second quantization) and its action on \cref{eq:frakK} would be simply
\[ Z^* \cH_B^\text{exc} Z = \sum_{k \in \north} \di\Gamma (e^{-L (k)} \mathfrak{K} (k) e^{L (k)}) \;, \]
i.\,e., conjugation of $\mathfrak{K}(k)$ by the one--boson unitary $e^{L(k)}$.
This would allow us to diagonalize the matrix $\mathfrak{K} (k)$ by an appropriate choice of $L(k)$. Even though $c$ and $c^*$ are not exactly bosonic operators, this remains approximately true on states with few excitations. After this diagonalization, it is simple to compare with $\mathbb{D}_B$ and conclude that (up to subleading error terms)
\begin{equation}\label{eq:HexcD}
Z^* \mathcal{H}^{\text{exc}}_{\text{B}} Z \gtrsim \mathbb{D}_{\text{B}}\;.
\end{equation}
Since, similarly to \cref{eq:Tinv}, also $Z$ leaves the difference $\mathbb{H}_0 - \mathbb{D}_\text{B}$ almost invariant (the fact that $Z$ can be expressed in terms of almost bosonic operators by \cref{eq:HDcomm} implies $[Z, \mathbb{H}_0 - \mathbb{D}_\text{B}] \simeq 0$), we obtain, with \cref{eq:HexcD}, the desired lower bound  
\begin{equation}\label{eq:Tinv_end} \begin{split} 
\langle T Z \xi, (\mathbb{D}_\text{B} + Q_\text{B}) T Z \xi \rangle &+ \langle TZ \xi , (\mathbb{H}_0 - \mathbb{D}_\text{B}) TZ \xi \rangle \\ &\simeq  E_N^\text{RPA} + \langle Z \xi , \cH_\text{B}^\text{exc} Z \xi \rangle + \langle \xi, (\mathbb{H}_0 - \mathbb{D}_B) \xi \rangle \gtrsim E_N^\text{RPA} \, . \end{split} \end{equation}

\item In \cite{BNPSS}, we controlled the non--bosonizable error terms as, informally stated, $T^*(\Ecal_1 + \Ecal_2)T \gtrsim - C \lVert \hat{V}\rVert_{\ell^1} \Hbb_0$, explaining the necessity of the interaction potential being small to control this term by a positive $\Hbb_0$. In the present paper instead we control $\mathcal{E}_{1}$ more precisely. In particular, we prove that on states $\xi$ close to the ground state of the correlation Hamiltonian, the following  improved bound holds true (see \cref{lm:dd}):
\begin{equation}\label{eq:E1est}
\langle T\xi, \mathcal{E}_{1} T\xi \rangle \ll C\hbar\;.
\end{equation}
This means that the contribution of the non--bosonizable term $\mathcal{E}_{1}$ to the energy is subleading with respect to $E^{\text{RPA}}_{N}$, which is of order $\hbar$. Concerning $\mathcal{E}_{2}$, by the Cauchy--Schwarz inequality we get (see \cref{cor:cE})
\[\pm  \mathcal{E}_{2}  \leq C N^{\alpha}  \mathcal{E}_{1}  + C\| \hat V \|_{1}N^{-\alpha}\mathbb{H}_{0} \;.\] 
The first term in the bound is controlled by the improved bound \cref{eq:E1est}, while the second term is controlled by positivity of $\langle T\xi, \mathbb{H}_{0} T\xi \rangle$ in \cref{eq:Ttrasfo}, for $N$ large enough without any smallness assumption on $V$. 

\item Furthermore, to implement this strategy, we improve the a--priori bounds on the number and the energy of excitations: our \cref{lm:H0-apri} and \cref{cor:apri-NN} generalize estimates of \cite{BNPSS} to interaction potentials with $\hat{V} \geq 0$ and $|\cdot| \hat{V} \in \ell^1 (\bZ^3)$. Moreover, \cref{lm:kinetic} now holds uniformly in $k$.  
 \end{itemize}

 The rigorous implementation is the subject of all remaining sections.
\section{A--Priori Estimates on Excitations of the Fermi Ball} 

The following lemma shows that vectors with total energy close to the ground state energy contain also only a small amount of kinetic energy.

\begin{lemma}[A--priori bound on kinetic energy] \label{lm:H0-apri}  
Assume $\sum_{k \in \bZ^3}  \lvert \hat{V} (k)\rvert |k| < \infty$ and $\hat{V}\geq 0$. Then there exists a $C >0$ such that we have
\[ \cH_\textnormal{corr} = R_\F^* H_N R_\F - E_N^\textnormal{HF} \geq \bH_0 - C \hbar \;. \]
Hence, for every $\psi \in L^2_\textnormal{a} (\bT^{3N})$ with $\| \psi \| = 1$ and 
$\langle \psi , H_N \psi \rangle \leq E_N^\textnormal{HF} + C \hbar$
the excitation vector $\xi = R_\F^* \psi \in \cF$ satisfies
\[ \langle \xi , \bH_0 \xi \rangle \leq C \hbar \;.\]
\end{lemma}
\begin{rem}
In the present paper we will apply \cref{lm:H0-apri} to the ground state $\psi_\textnormal{gs}$, which by the variational principle even satisfies $\langle \psi_\textnormal{gs} , H_N \psi_\textnormal{gs} \rangle \leq E_N^\textnormal{HF}$.
\end{rem}

\begin{proof}[Proof of \cref{lm:H0-apri}] 
From $\hat{V} \geq 0$ we get 
\[ \begin{split} 
0 &\leq \int_{\bT^3 \times \bT^3} V (x-y) \Bigg( \sum_{j=1}^N \delta (x_j - x) - N \Bigg) \Bigg( \sum_{i=1}^N \delta (x_i - y) - N \Bigg) \di x \di y \\
&= 2 \sum_{i<j}^N V(x_i - x_j) + N V (0) - N^2 \hat{V} (0) \;. \end{split} \]
Thus 
\[ H_N \geq \sum_{j=1}^N -\hbar^2 \Delta_{x_j} + \frac{N}{2} \hat{V} (0) - \frac{V(0)}{2}\;. \]
Switching to Fock space $\cF$ and conjugating with $R_\F$, we conclude that 
\begin{equation}\label{eq:low-RHR} R_\F^* \cH_N R_\F \geq \sum_{p \in \bZ^3} \hbar^2 p^2 R_\F^* a_p^* a_p R_\F + \frac{N}{2} \hat{V} (0) - \frac{V(0)}{2} = \bH_0 + \sum_{p \in B_\F} \hbar^2 p^2 + \frac{N}{2} \hat{V} (0) - \frac{V(0)}{2} \;. \end{equation} 
We compare the r.\,h.\,s.\ of \cref{eq:low-RHR} with the Hartree--Fock energy \cref{eq:HF-en}. We have 
\[ \frac{1}{2N} \sum_{k,k' \in B_\F} \hat{V} (k-k')  = \frac{V(0)}{2} - \frac{1}{2N} \sum_{k \in B_\F} \sum_{k' \in B_\F^c} \hat{V} (k-k') \;. \]
Setting $q = k-k'$ and noting that $| B_\F \cap (B_\F^c + q)| \leq C |q| N^{2/3}$, we estimate   
\[ \begin{split}  \frac{1}{2N} \sum_{k \in B_\F} \sum_{k' \in B_\F^c} \hat{V} (k-k') &= \frac{1}{2N} \sum_{k \in B_\F} \sum_{q \in B_\F^c +k} \hat{V} (q) \\ &=  \frac{1}{2N} \sum_{q \in \bZ^3} \hat{V} (q) \sum_{k \in B_\F \cap (B_\F^c +q)} 1 \leq C \hbar \sum_{q \in \bZ^3} \hat{V} (q) |q| \;. \end{split}  \]
By assumption on $V$, this implies  
\[ \frac{1}{2N} \sum_{k,k' \in B_\F} \hat{V} (k-k') \geq \frac{V(0)}{2} - C \hbar \;.  \]
With \cref{eq:HF-en} and \cref{eq:low-RHR} we conclude that 
$R_\F^* \cH_N R_\F \geq E_N^\textnormal{HF} + \bH_0 - C \hbar$.
\end{proof} 

The a--priori bound from \cref{lm:H0-apri} for the kinetic energy $\bH_0$ has several consequences. First of all, it gives control on the number of excitations of the Slater determinant. Here, it is useful to introduce gapped number--of--fermions operators which are easier to control than $\Ncal$. For 
$\eps > 0$, we define the gapped number operator
\begin{equation}\label{eq:gappedN} \cN_\eps := \sum_{p \in \bZ^3 :\, ||p| - k_\F| > N^{-\eps}} a_p^* a_p \end{equation} 
measuring the number of excitations with momenta further than a distance $N^{-\eps}$ from the Fermi sphere. (The definition \cref{eq:gappedN} differs slightly from the definition used in \cite{BNPSS} but that is merely a matter of convenience.)
\begin{cor}[A--priori bounds on particle number] \label{cor:apri-NN} 
There exists a constant $C > 0$ such that, on  $\chi (\cN_\textnormal{p} - \cN_\textnormal{h} = 0) \cF$, we have
\begin{equation}\label{eq:NNe-bd}
\begin{split}  
\cN \leq C N^{2/3} \bH_0  \quad \text{and} \quad
\cN_\eps \leq C N^{1/3+\eps} \bH_0 \quad \text{for every $\eps > 0$.}
\end{split} 
\end{equation}
Assume furthermore that $\sum_{k \in \bZ^3} \lvert \hat{V} (k)\rvert |k| < \infty$ and $\hat{V}\geq 0$. Then, for $\psi \in L^2_\textnormal{a} (\bT^{3N})$ with $\| \psi \| = 1$ and $\langle \psi , H_N \psi \rangle \leq E_N^\textnormal{HF} + C \hbar$, 
the excitation vector $\xi = R_\F^* \psi \in \cF$ satisfies 
\begin{equation}\label{eq:N-apri} \langle \xi, \cN  \xi \rangle \leq C N^{1/3}  \quad \text{and} \quad \langle \xi , \cN_\eps \xi \rangle \leq C N^\eps  \quad \text{for every $\eps > 0$.}\end{equation} 
\end{cor} 
\begin{proof} To prove \cref{eq:NNe-bd} for $\Ncal_\varepsilon$, observe that $||p| - k_\F| > N^{-\eps}$ implies $|\hbar |p| - \kappa | > \hbar N^{-\eps}$ and thus 
\[ |\hbar^2 p^2 - \kappa^2 | \geq \kappa \hbar N^{-\eps} \;.\]
Thus
\[ \bH_0 \geq \sum_{p \in \bZ^3 :||p| - k_\F| > N^{-\eps}}  |\hbar^2 p^2 - \kappa^2 |  a_p^* a_p \geq \kappa \hbar N^{-\eps} \cN_\eps \;.\]
The bound for $\cN$ is proven in \cite[Lemma 2.4]{BNPSS}; \cref{eq:N-apri} follows using \cref{lm:H0-apri}.
\end{proof}

Furthermore, the estimate for $\bH_0$ from \cref{lm:H0-apri} allows us to bound the particle--hole pair operators $b(k)$ and $b^* (k)$ introduced in \cref{eq:bb}.
\begin{lemma}[Kinetic bound on particle--hole pairs] \label{lm:kinetic}
There exists a constant $C > 0$ such that, for all $k \in \bZ^3$, 
\begin{equation}\label{eq:num1} \sum_{p \in B_\F^c \cap (B_\F + k)} \| a_p a_{p-k} \psi \| \leq C N^{1/2}  \| \bH_0^{1/2} \psi \|\; \end{equation} 
and moreover 
\begin{equation}\label{eq:num2} \sum_{\substack{p \in B_\F^c \cap (B_\F + k) : \\  e(p ) + e(p-k) \leq C N^{-1/3-\delta}}}  \| a_p a_{p-k} \psi \| \leq C N^{1/2-\delta/2}  \| \bH_0^{1/2} \psi \| \;.\end{equation} 
\end{lemma} 

The bounds \cref{eq:num1} and \cref{eq:num2} have been established in \cite[Appendix B]{BNPSS} (and previously in \cite[Lemma 4.7]{HPR20}) for fixed $k$ (which was sufficient since there only $k$ in the compact support of $\hat{V}$ was relevant). Here, we improve the proof given in \cite{BNPSS} to obtain uniformity in $k$. We use the following number theoretic result.
\begin{prop}[Lattice points in convex bodies, \cite{Hux}]\label{thm:number} Let $K\subset \mathbb{R}^{2}$ be a smooth convex body and let $R K$ be its dilation by a factor $R>0$, $R K := \{ x\in \mathbb{R}^{2} \mid x / R \in K \}$. Consider the number of points of $\mathbb{Z}^{2}$ belonging to $R K$,
\begin{equation}
\mathfrak{N}_{K}(R) := \big| \{ n \in \mathbb{Z}^{2} \mid \frac{n}{R} \in K \} \big|\;.
\end{equation}
Let
\begin{equation}
\mathcal{E}_{K}(R) := \mathfrak{N}_{K}(R) - R^{2} |K|\;.
\end{equation}
Then, for any $\gamma > 131/208$, there exists $C_{K,\gamma}>0$ independent of $R$ such that
\begin{equation}\label{eq:numbth}
| \mathcal{E}_{K}(R) | \leq C_{K,\gamma} R^{\gamma}\;.
\end{equation}
\end{prop}
\begin{rem}The constant $C_{K,\gamma}$ in the estimate \cref{eq:numbth} depends on the curvature of the boundary of $K$. In particular, $C_{K,\gamma}$ is finite as long as the curvature is strictly positive. For us it is sufficient that \cref{eq:numbth} holds for some $\gamma  < 1$. A simple proof for $2/3< \gamma < 1$ is given in \cite[Theorem~7.7.16]{Hor} (the condition $0 \in K$ given there can always be achieved by a translation).
\end{rem}

\begin{proof}[Proof of \cref{lm:kinetic}]
We first prove \cref{eq:num1}.
Proceeding as in \cite[Lemma~4.7]{HPR20} by the Cauchy--Schwarz inequality we get
\begin{align*}
 & \sum_{p \in \BFc \cap (\BF +k)} \norm{a_p a_{p-k} \psi} \\
 & \leq \Bigg(\sum_{p \in \BFc \cap (\BF +k)} \frac{1}{e(p)+e(p-k)} \Bigg)^{1/2} \Bigg(\sum_{p \in \BFc \cap (\BF +k)}\left( e(p)+e(p-k) \right) \norm{a_p a_{p-k} \psi}^2 \Bigg)^{1/2}\;.
\end{align*}
The second factor is bounded by the kinetic energy as claimed,
\begin{align*}
 & \sum_{p \in \BFc \cap (\BF +k)}\left( e(p)+e(p-k) \right) \norm{a_p a_{p-k} \psi}^2 \\
 & \leq \sum_{p \in \BFc \cap (\BF +k)} e(p) \norm{a_p \psi}^2 + \sum_{p \in \BFc \cap (\BF +k)} e(p-k) \norm{a_{p-k} \psi}^2  \leq \langle \psi, \Hbb_0 \psi\rangle\;. 
\end{align*}
Therefore it is enough to show 
\begin{equation}\label{eq:I-1}
 \sum_{p \in B_\F^c \cap (B_\F+k)} \frac{1}{p^2 - (p-k)^2} \leq C N^{1/3} \;.
\end{equation}
If $|k| > C_0 N^{1/3}$ (for a $C_0 > 0$ large enough), we have $p^2 - (p-k)^2 > C_1 N^{2/3}$ for all $p \in B_\F^c \cap (B_\F + k)$ (with a different constant $C_1 > 0$) and \cref{eq:I-1} is clear. Thus we can assume that from now on
\[ \lvert k\rvert  \leq C_0 N^{1/3} \;. \]
We need to further distinguish the cases $p^2 - (p-k)^2 \geq 4 N^{1/3}$ and $p^2 - (p-k)^2 < 4 N^{1/3}$.
\paragraph{The case $p^2 - (p-k)^2 \geq 4 N^{1/3}$.}
We apply the argument used in \cite[Eq.~(5.13)]{FLLS}.
If $\eta \in (0, \frac{3}{2C_0})$ then for $q \in B_\eta(p)$ we have
\[ \begin{split} \lvert q^2 - (q-k)^2 \rvert & \geq \big\lvert \lvert p^2 - (p-k)^2 \rvert - \lvert 2(p-q) \cdot k\rvert \big\rvert \geq 4 N^{1/3} - 2 \eta C_0 N^{1/3}\geq N^{1/3} \;.
   \end{split} 
\]
With
\[ \nabla_q \frac{1}{q^2 - (q-k)^2} = \frac{2k}{q^2 - (q-k)^2} \frac{1}{q^2 - (q-k)^2} \]
we conclude that  
\[ \Big| \frac{1}{p^2 - (p-k)^2}  -  \frac{1}{\wt{p}^2 - (\wt{p}-k)^2} \Big| \leq  \eta 2 C_0 \sup_{q \in B_\eta (p)}   \frac{1}{q^2 - (q-k)^2} \]
for all $\wt{p} \in B_\eta (p)$. Hence, if $\eta > 0$ is small enough, we get
\[  \sup_{q \in B_\eta (p)}   \frac{1}{q^2 - (q-k)^2} \leq  \frac{2}{p^2 - (p-k)^2}  \]
and 
\[ \frac{1}{p^2 - (p-k)^2} \leq 2 \inf_{q \in B_\eta (p)}   \frac{1}{q^2 - (q-k)^2} \;.\]
Possibly choosing $\eta >0$ still smaller, the balls $B_\eta (p)$ are disjoint for different $p$, and we obtain
\[
\begin{split}
\sum_{p \in B_\F^c \cap (B_\F+k)} \frac{\chi (p^2 - (p-k)^2 \geq 4N^{1/3})}{p^2 - (p-k)^2} &\leq C \int_{p \in B_\F^c \cap (B_\F+k)} \frac{1}{p^2 - (p-k)^2}  \di p \\ &\leq C N^{1/3} \int_{|p| > 1, |p-k'| < 1} \frac{1}{p^2 - (p-k')^2} \di p
\end{split}
\]
where we defined $k' := k/ k_\F$. With
\[
p^2 - (p-k)^2 = (p^2 - 1) + (1- (p-k)^2) \geq 2 \big( p^2 - 1 \big)^{1/2} \big(1- (p-k)^2 \big)^{1/2}
\]
we conclude that 
\[
\begin{split}
\sum_{p \in B_\F^c \cap (B_\F+k)} \frac{\chi (p^2 - (p-k)^2 \geq 4N^{1/3})}{p^2 - (p-k)^2} &\leq C N^{1/3} \int_{\substack{\lvert p\rvert > 1,\\\lvert p-k'\rvert < 1}} \frac{1}{(p^2-1)^{1/2} (1 - (p-k')^2)^{1/2}} \di p\\
&\leq C N^{1/3}
\end{split}
\]
uniformly in $k$, as shown in \cite[Lemma 3.4]{FLLS}.

\paragraph{The case $p^2 - (p-k)^2 < 4 N^{1/3}$.} We observe that $p \in B_\F^c$ and $p-k \in B_\F$ together imply the lower bound (recall that all momenta are elements of $\Zbb^3$)
\[
1 \leq p^2 - (p-k)^2 = 2p \cdot k - k^2 =: m \in \bN \;.
\]
Since moreover $p^2 > k_\F^2$ and $(p-k)^2 = p^2 - m \leq k_\F^2$, we find 
\[ k_\F^2 < p^2 \leq k_\F^2 + m\;.\]
We obtain   
\begin{equation}\label{eq:I-bd} 
 \sum_{p \in B_\F^c \cap (B_\F+k)} \frac{\chi (p^2 - (p-k)^2 \leq 4N^{1/3})}{p^2 - (p-k)^2} \leq \sum_{m=1}^{4N^{1/3}} \frac{1}{m} |B_m (k)| 
\end{equation}
with 
\[
B_m (k) := \Big\{ p \in \bZ^3 : k_\F^2 < \lvert p\rvert^2 \leq k_\F^2 +m \text{ and }  2 p\cdot k - \lvert k\rvert^2 = m \Big\} \;.
\]  
Without loss of generality $|k_1| \geq |k_2|$ and $|k_1| \geq |k_3|$ (in particular, since $k \not = 0$, we have $k_1 \not = 0$). Then, for $p = (p_1 , p_2, p_3) \in B_m (k)$, the condition $2p\cdot k - \lvert k\rvert^2 =m$ is solved by 
\begin{equation}\label{eq:p1-sol}  p_1 = \frac{m+k^2}{2k_1} - p_2 \frac{k_2}{k_1} - p_3 \frac{k_3}{k_1} \;. \end{equation} 
Thus $|B_m (k)|$ is bounded by the number of points $(p_2, p_3) \in \bZ^2$ with 
\begin{equation}\label{eq:ellip} k_\F^2 \leq \left( \frac{m+k^2}{2k_1} - p_2 \frac{k_2}{k_1} - p_3 \frac{k_3}{k_1} \right)^2 + p_2^2 + p_3^2 \leq k_\F^2 + m \;.\end{equation}
(This is only an upper bound because $(p_2, p_3) \in \bZ^2$ for which the r.\,h.\,s.\ of \cref{eq:p1-sol} is not integer do not contribute to $B_m (k)$). On the $(p_2, p_3)$--plane, we define new variables $(q_2, q_3)$ by
\begin{equation}\label{eq:ellip2} \begin{split} 
p_2 & := \frac{k_2}{\sqrt{k_2^2+k_3^2}}  q_2 - \frac{k_3}{\sqrt{k_2^2+ k_3^2}} q_3 + \frac{k^2 +m}{2|k|} \frac{\sqrt{k_2^2+ k_3^2}}{|k|} \;, \\
p_3 & := \frac{k_3}{\sqrt{k_2^2+ k_3^2}} q_2 +  \frac{k_2}{\sqrt{k_2^2+k_3^2}}  q_3 \;.\end{split} \end{equation} 
In terms of these new variables, we can rewrite \cref{eq:ellip} as 
\begin{equation}\label{eq:ellip3} k_\F^2 - \left( \frac{k^2+m}{2|k|} \right)^2  \leq \frac{k^2}{k_1^2} q_2^2 + q_3^2 \leq k_\F^2 +m - \left( \frac{k^2+m}{2|k|} \right)^2  \;.\end{equation}
We can therefore apply \cref{thm:number} to estimate the number of points $(p_2, p_3) \in \bZ^2$ contained between the two ellipses described by \cref{eq:ellip3}. (From the assumptions $|k_1| \geq |k_2|$ and $|k_1| \geq |k_3|$ we have $1 \leq |k| / |k_1| \leq 3$, which implies that the error term in \cref{eq:numbth} is uniform in $k$.) We conclude that 
\[ |B_m (k)| \leq \pi \frac{k_1}{|k|}  m + C k_\F^\gamma \leq C (m+N^{\gamma/3})   \qquad \textnormal{for a } \gamma > \frac{131}{208}.\]
Inserting this bound in \cref{eq:I-bd} and choosing $\gamma < 1$ we arrive at 
\[ \sum_{p \in B_\F^c \cap (B_\F+k)} \frac{\chi (p^2 - (p-k)^2 \leq 4N^{1/3})}{p^2 - (p-k)^2} \leq C \sum_{m=1}^{4N^{1/3}} \frac{1}{m} (m+N^{\gamma/3}) \leq C N^{1/3} \;. \]

To show \cref{eq:num2}, we proceed analogously. The only difference is that now the sum in \cref{eq:I-bd} can be restricted to $m \leq C N^{1/3 - \delta}$ (here, the case $p^2 - (p-k)^2 \geq 4N^{1/3}$ is not relevant).  
\end{proof}

From \cref{lm:kinetic}, we immediately obtain  a bound on the operators $b (k)$ and $b^* (k)$. For details, see \cite[Lemma 2.3]{BNPSS}.
\begin{cor}[Kinetic bound on pair operators] \label{lm:bb}
There exists a $C > 0$ such that for all $k \in \bZ^3$ we have
\[ b^* (k) b (k) \leq C N  \bH_0\;, \qquad b (k) b^* (k) \leq C N (\bH_0 + \hbar) \;. \] 
\end{cor}

Using the last corollary, we obtain an a--priori bound for the bosonizable interaction $Q_\textnormal{B}$.  
\begin{cor}[Bosonizable interaction] \label{cor:QB}
Assume $\lVert \hat{V} \rVert_1 < \infty$. Then there exists $C > 0$ such that 
\[ -C (\bH_0 + \hbar) \leq Q_\textnormal{B} \leq C (\bH_0 + \hbar) \;.\]
\end{cor}
\begin{proof} 
We observe that, for any $k \in \bZ^3$, by \cref{lm:bb},  
\[ \begin{split} 0 &\leq (b^* (k) \pm b (-k)) (b (k) \pm b (-k)) \\ &= b^* (k) b (k) + b (-k) b^* (-k) \pm \left[ b^* (k) b^* (-k) + b(-k) b(k) \right] \\ &\leq C  N (\bH_0 + \hbar) \pm \left[ b^* (k) b^* (-k) + b(-k) b(k) \right] \;.\end{split} \]
Hence 
\[ -C N (\bH_0 + \hbar) \leq  b^* (k) b^* (-k) + b(-k) b(k) \leq C N (\bH_0 + \hbar) \, . \]
After summing over $k$, this implies the desired estimate for $Q_\textnormal{B}$.
\end{proof} 

Finally, the a--priori bound for $\bH_0$ (and the resulting estimates on $\cN$ and $\cN_\eps$ from \cref{cor:apri-NN}) imply that the error terms in \cref{eq:errors} are negligible. 
First of all, the exchange operator $\bX$ can be bounded with the following lemma, taken from \cite[Lemma 2.5]{BNPSS}.
\begin{lemma}[Exchange term] \label{lm:X}
Assume $\| \hat{V} \|_1 < C$. Then there exists a $C >0$ such that for all $\xi \in \chi (\cN^p - \cN^h = 0) \cF$ we have
\[ |\langle \xi, \bX \xi \rangle | \leq C N^{-1/3} \langle \xi, \bH_0 \xi \rangle \;.\]
\end{lemma} 

The next lemma provides control on the error term $\cE_1$ in \cref{eq:errors}. It is one of the key achievements of the present paper.
\begin{lemma}[Non--bosonizable interaction]\label{lm:dd} 
Assume $\| \hat{V} \|_1 < \infty$. Fix $0 < \eps < 1/3$ and $131/208 < \gamma < 1$. Then there exists $C > 0$ such that for all $\xi \in \chi (\cN_\textnormal{h} - \cN_\textnormal{p} = 0) \cF$ we have
\begin{equation}\label{eq:eps1} 
\begin{split} 
\langle \xi , \cE_1 \xi \rangle \leq \; &C N^{-1} \| (\cN+1)^{3/2} \xi \| \| \cN_{1/3-\eps}^{1/2} \xi \| +C N^{\eps-1} (N^\eps + N^{\gamma/3})  \| \cN^{1/2} \xi \|^2 \;.
\end{split} \end{equation}  
\end{lemma} 

\begin{rem}
With a localization argument, we will be able to restrict our attention to states for which $\cN \leq C N^{1/3}$ and $\cN_\delta \leq C N^\delta$  (for the expectation value as stated in \cref{cor:apri-NN}, but also for higher moments). Applying \cref{eq:eps1} for such states, choosing $\gamma < 1$ and $\eps > 0$ small enough, we conclude that $\cE_1 \ll N^{-1/3}$ and therefore that $\cE_1$ does not contribute to the correlation energy, to leading order.
\end{rem}

\begin{proof}[Proof of \cref{lm:dd}]
Recall the definition \cref{eq:dd} of the operators $d^* (k)$ and $d(k)$. 
Since $d(0) = d^* (0) = 0$ on $\chi (\cN_\textnormal{h} - \cN_\textnormal{p} = 0) \cF$, we find
\[ \begin{split} 
\langle \xi , \cE_1 \xi \rangle = \frac{1}{2N} \sum_{k \in \bZ^3 \backslash \{ 0 \}} \hat{V} (k) \sum_{q_1,q_2 \in [B_\F^c \cap (B_\F^c + k)] \cup [B_\F \cap (B_\F -k)]} \hspace{-.5cm}  \sigma_{q_1} \sigma_{q_2}  \langle \xi , a_{q_1}^* a_{q_1 - \sigma_{q_1} k} a^*_{q_2 -\sigma_2 k} a_{q_2} \xi \rangle \end{split} \]
where we introduced the notation $\sigma_q = 1$, if $q \in B_\F^c \cap (B_\F^c + k)$, and $\sigma_q = -1$, if $q \in B_\F \cap (B_\F +k)$. With the canonical anticommutation relations \cref{eq:CAR}, we obtain
\begin{align}
 \langle \xi , \cE_1 \xi \rangle = &- \frac{1}{2N} \sum_{k \in \bZ^3 \backslash \{ 0 \}} \hat{V} (k) \sum_{q_1,q_2 \in [B_\F^c \cap (B_\F^c + k)] \cup [B_\F \cap (B_\F -k)]} \hspace{-.5cm}  \sigma_{q_1} \sigma_{q_2}  \langle \xi , a_{q_1}^* a^*_{q_2 -\sigma_2 k} a_{q_1 - \sigma_{q_1} k}  a_{q_2} \xi \rangle \nonumber\\
 &+  \frac{1}{2N} \sum_{k \in \bZ^3 \backslash \{ 0 \}} \hat{V} (k) \sum_{q_1 \in [B_\F^c \cap (B_\F^c + k)] \cup [B_\F \cap (B_\F -k)]} \hspace{-.3cm}  \langle \xi , a_{q_1}^*  a_{q_1} \xi \rangle \;. \label{eq:commu} 
\end{align} 
The second term can be estimated by 
\[ \frac{1}{2N} \sum_{k \in \bZ^3 \backslash \{ 0 \}} \hat{V} (k) \sum_{q_1 \in [B_\F^c \cap (B_\F^c + k)] \cup [B_\F \cap (B_\F -k)]}  \| a_{q_1} \xi \|^2 \leq C  N^{-1} \| \cN^{1/2} \xi \|^2 \;.\]
Let us focus on the first term on the r.\,h.\,s.\ of \cref{eq:commu}. The first observation is that contributions with at least one of the four momenta $q_1$, $q_1 - \sigma_1 k$, $q_2$, $q_2 - \sigma_2 k$ at distances larger than $N^{-1/3 + \eps}$ from the Fermi sphere, for an $0 < \eps < 1/3$ to be chosen later, can be bounded using a combination of $\cN$ and of the gapped number operator $\cN_{1/3-\eps}$ defined in \cref{eq:gappedN}. In fact, considering for example the case $||q_1| - k_\F| > N^{-1/3 + \eps}$ (and dropping, for an upper bound, all other restrictions on $q_1$ and $q_2$), we have  
\[
\begin{split}
\frac{1}{N} \sum_{k \in \bZ^3 \backslash \{ 0 \}} &\hat{V} (k) \sum_{q_1, q_2 \in \bZ^3 : ||q_1| - k_\F| > N^{-1/3 + \eps}} | \langle \xi , a_{q_1}^* a^*_{q_2 - \sigma_2 k} 
a_{q_1 - \sigma_{q_1} k}  a_{q_2} \xi \rangle | \\ 
\leq \; & \frac{1}{N} \sum_{k \in \bZ^3 \backslash \{ 0 \}} \hat{V} (k) \left( \sum_{q_1, q_2 \in \bZ^3 : ||q_1| - k_\F| > N^{-1/3 + \eps}} \| a_{q_1} a_{q_2 - \sigma_2 k} (\cN + 1)^{-1/2} \xi \|^2 \right)^{1/2} \\ & \hspace{2.3cm} \times \left( \sum_{q_1, q_2 \in \bZ^3}  \| a_{q_1 - \sigma_{q_1} k}  a_{q_2} (\cN+1)^{1/2} \xi \|^2 \right)^{1/2} \\
\leq \; &C N^{-1}  \| \cN_{1/3 - \eps}^{1/2} \xi \| \| (\cN+1)^{3/2} \xi \| \end{split}
\]
where we used $a^*_p \cN = (\cN-1) a^*_p$ for all $p\in \Zbb^3$. Thus
\begin{equation}\label{eq:cE1-3}
\begin{split}
\langle \xi , \cE_1 \xi \rangle \leq \; &C N^{-1} \| \cN^{1/2} \xi \|^2 + C N^{-1} \| \cN_{1/3 - \eps}^{1/2} \xi \| \| (\cN+1)^{3/2} \xi \| \\
&+ \frac{1}{N} \sum_{k \in \bZ^3 \backslash \{ 0 \}} \hat{V} (k) \sum_{q_1,q_2  \in A^\textnormal{p}_k \cup A^\textnormal{h}_k} | \langle \xi , a_{q_1}^* a_{q_2 - \sigma_2 k}^* a_{q_1 -\sigma_1 k} a_{q_2} \xi \rangle |
\end{split}
\end{equation} 
where we defined the momentum sets
\[
\begin{split} 
A^\textnormal{p}_k & := \left\{ q \in \bZ^3 : k_\F < |q| < k_\F + N^{-1/3 + \eps} \quad \textnormal{and} \quad k_\F < |q-k| < k_\F + N^{-1/3 + \eps} \right\}\;, \\
A^\textnormal{h}_k & := \left\{ q \in \bZ^3 : k_\F  - N^{-1/3 + \eps} < |q| \leq k_\F  \quad \textnormal{and} \quad k_\F - N^{-1/3 + \eps} < |q + k| \leq k_\F \right\}\;.
\end{split}
\]
Note that for $q_1 \in A^\textnormal{p}_k$ we have  
\[
\begin{split} 
k_\F^2 &\leq (q_1 - k)^2 = q_1^2 + k^2 -2 q_1 \cdot k \\ &\leq (k_\F + N^{-1/3+ \eps})^2 + k^2 -2 q_1 \cdot k \leq k_\F^2 + C N^\eps + k^2 -2 q_1 \cdot k
\end{split} \]
and thus $2 q_1 \cdot k - k^2 \leq C N^{\eps}$. Inverting the roles of $q_1$ and $q_1 - k$, we also obtain $2 q_1 \cdot k -k^2 \geq - C N^\eps$. Arguing similarly for $q_1 \in A^\textnormal{h}_k$, we conclude that 
\begin{equation}\label{eq:qdot0} 
-C N^\eps \leq 2 q_1 \cdot k - k^2  \leq C N^\eps 
\end{equation} 
for all $q_1 \in A^\textnormal{p}_k \cup A^\textnormal{h}_k$ (which means that the set $A^\textnormal{p}_k \cup A^\textnormal{h}_k$ is localized close to the equator of the Fermi sphere, thinking of the direction of $k$ as defining the north pole). 

Using the Cauchy--Schwarz inequality and $\| a_{q_1} \|_\op \leq 1$, $\| a_{q_1 - \sigma_1 k} \|_\op \leq 1$, we conclude that the last term on the r.\,h.\,s.\ of \cref{eq:cE1-3} can be bounded by
\begin{equation}\label{eq:simpler}
\begin{split}
\frac{1}{N} &\sum_{k \in \bZ^3 \backslash \{ 0 \}} \hat{V} (k) \sum_{q_1,q_2  \in A^\textnormal{p}_k \cup A^\textnormal{h}_k} | \langle \xi , a_{q_1}^* a_{q_2 - \sigma_2 k}^* a_{q_1 -\sigma_1 k} a_{q_2} \xi \rangle |  \\
&\leq \frac{1}{N} \sum_{k \in \bZ^3 \backslash \{ 0 \}} \hat{V} (k)  |A^\textnormal{p}_k \cup A^\textnormal{h}_k| \| \cN^{1/2} \xi \|^2 \leq \frac{\| \cN^{1/2} \xi \|^2}{N}  \sum_{k \in \bZ^3 \backslash \{ 0 \}} \hat{V} (k)  \sum_{m=-CN^\eps}^{CN^\eps} |B_{m,k}| 
\end{split}
\end{equation} 
where we defined 
\begin{equation}
\label{eq:Bsk}
\widetilde{B}_{m}(k) := \{ q \in \bZ^3 : k_\F - N^{-1/3 + \eps} \leq |q| \leq k_\F + N^{-1/3+ \eps} \text{ and } 2 q \cdot k - k^2 = m \} \; .
\end{equation} 
Proceeding as in the proof of \cref{lm:kinetic} following \cref{eq:I-bd}, we find, for $131/208 < \gamma < 1$, 
\[  |\widetilde{B}_{m}(k)| \leq C  (N^{\varepsilon} + N^{\gamma/3}) \;. \]
Inserting in \cref{eq:simpler} and using $\| \hat{V} \|_1 < \infty$, we obtain 
\[ \begin{split} \frac{1}{N} &\sum_{k \in \bZ^3 \backslash \{ 0 \}} \hat{V} (k) \sum_{q_1,q_2  \in A^\textnormal{p}_k \cup A^\textnormal{h}_k} | \langle \xi , a_{q_1}^* a_{q_2 - \sigma_2 k}^* a_{q_1 -\sigma_1 k} a_{q_2} \xi \rangle |  \leq \frac{C}{N} N^{\eps} (N^\eps + N^{\gamma/3})  \| \cN^{1/2} \xi \|^2 \;.
\end{split} \]
With \cref{eq:cE1-3} this concludes the proof of \cref{lm:dd}.
 \end{proof}

\cref{lm:dd} proves that the error term $\cE_1$ is negligible (in the ground state and, more generally, on low--energy states with correlation energy of order $\hbar$). Together with \cref{lm:bb}, it also allows us to neglect the term $\cE_2$ in \cref{eq:errors}. The following corollary improves \cite[Lemma 9.1]{BNPSS} in not requiring smallness of $V$, and is also simpler to prove.
\begin{cor}[Coupling of bosonizable and non--bosonizable terms] \label{cor:cE}
Assume $\| \hat{V} \|_1 < \infty$ and $\hat{V} \geq 0$. With the error terms $\cE_1$, $\cE_2$ defined as in \cref{eq:errors}, we have
\begin{equation}\label{eq:CS-cE} \pm \cE_2 \leq N^\alpha \cE_1 + C N^{-\alpha} \bH_0  \qquad \text{for every $\alpha \geq 0$.}\end{equation} 
With \cref{lm:dd}, we conclude that for $131/208 < \gamma < 1$ and $\eps > 0$ small enough (choosing $\alpha = \eps/4$ in \cref{eq:CS-cE}), there exists a constant $C > 0$ such that 
\begin{equation}
\label{eq:cE1cE2} 
\begin{split} 
\langle \xi , (\cE_1 + \cE_2) \xi \rangle \geq \; &- C N^{-1+ \eps/4} \| (\cN+1)^{3/2} \xi \| \| \cN^{1/2}_{1/3-\eps} \xi \| - C N^{5\eps/4  + \gamma/3 - 1} \| \cN^{1/2} \xi \|^2 \\ &- C N^{-\eps/4} \| \bH_0^{1/2} \xi \|_2^2 
\end{split} \end{equation} 
for all $\xi \in \chi (\cN_\textnormal{h} - \cN_\textnormal{p} = 0) \cF$. 
\end{cor} 

\begin{rem}The choice $\alpha = \eps/4$ optimizes the sum of the first and the last term on the r.\,h.\,s.\ of \cref{eq:cE1cE2}, counting (following the argument in the remark after \cref{lm:dd}) $\| (\cN+1)^{3/2} \xi \| \lesssim N^{1/2}$, $\| \cN^{1/2}_{1/3-\eps} \xi \| \lesssim N^{1/6-\eps/2}$, and $\| \bH^{1/2}_0 \xi \|^2 \lesssim N^{-1/3}$. The second term on the r.\,h.\,s.\ of \cref{eq:cE1cE2} is of lower order if $\gamma$ is chosen small enough.
\end{rem}

\begin{proof}[Proof of \cref{cor:cE}] 
By Cauchy--Schwarz, \cref{lm:bb}, and $\| \hat{V} \|_1 < \infty$, we find 
\[ \pm \cE_2 \leq N^\alpha \cE_1 + N^{-\alpha - 1} \sum_{k \in \bZ^3} \hat{V} (k) b^* (k) b(k) \leq N^\alpha \cE_1 + C N^{-\alpha} \bH_0 \; .  \qedhere\]
\end{proof}

\section{Patch Decomposition and Almost Bosonic Operators} 
\label{sec:patch} 

The bounds in last section allow us to approximate the correlation Hamiltonian \cref{eq:corr} by $\bH_0 + Q_\textnormal{B}$, with $\bH_0$ and $Q_\textnormal{B}$ defined in \cref{eq:RHR-main}. The term $Q_\textnormal{B}$, arising from the interaction, is quadratic in the particle--hole pair creation and annihilation operators $b^* (k)$, $b(k)$. It turns out that, on states with few excitations of the Fermi ball, the operators $b^* (k)$ and $b(k)$ satisfy approximately bosonic commutation relations. 

In order to express also the kinetic energy $\bH_0$ in terms of almost bosonic creation and annihilation operators, we have to decompose a layer around the Fermi sphere $\partial B_\F$ into $M$ patches $\{ B_\alpha \}_{\alpha=1}^M$, for the number of patches $M \in \bN$ to be chosen as a function of $N$ at the end of the paper. Such a decomposition has been constructed in \cite{BNPSS0}. One starts by decomposing a half sphere in $M/2$ patches. The sidelengths of the patches are comparable (they are both of order $N^{1/3} / M^{1/2}$). The patches have thickness
\[1 \ll 2R \ll N^{1/3}\]
in the radial direction (later we will impose stronger conditions). Furthermore, the patches are disjoint and separated by corridors, larger than $R$. We denote by $\omega_\alpha$ the center of the patch $B_\alpha$. Finally, the patch decomposition of the first half sphere is mirrored by the map $k \mapsto -k$ onto the other half sphere. The construction is so that the area of the radial projection $p_\alpha$ of the patch $B_\alpha$ on the unit sphere $\mathbb{S}_2$ has area $4\pi / M$, up to corrections of order $N^{-1/3} M^{-1/2}$, and diameter bounded by $C/\sqrt{M}$, for all $\alpha = 1, \dots , M$; see \cite[Section 3.2]{BNPSS0} for the details.

For fixed $k \in \bZ^3$ with $|k| < R$, we are going to exclude patches in a small strip around the equator (thinking of the direction of $k$ as defining the north direction) of the Fermi sphere. More precisely, for $0 < \delta < 1/6$, we define $\cI_k := \cI_k^+ \cup \cI_k^-$, with 
\begin{equation}\label{eq:Ik-def}
\begin{split} 
\cI_k^+ & := \{ \alpha \in \{ 1, \dots , M \} : k \cdot \hat{\omega}_\alpha \geq N^{-\delta} \} \;,\\
\cI_k^- & := \{ \alpha \in \{ 1, \dots , M \} : k \cdot \hat{\omega}_\alpha \leq - N^{-\delta} \} \;.
\end{split}
\end{equation}
Given $k \in \bZ^3$, $|k| < R$ and $\alpha \in \cI_k^+$, we introduce the  
particle--hole pair creation operator  
\begin{equation}
\label{eq:normali}
b^*_\alpha(k) := \frac{1}{n_\alpha(k)} \sum_{\substack{p\colon p \in B_\F^c \cap B_\alpha \\ p-k\in B_\F \cap B_\alpha}}  a^*_p a^*_{p-k}
\end{equation}
with the normalization constant  
\[
n_\alpha(k)^2 := \sum_{\substack{p\colon p \in B_\F^c \cap B_\alpha\\p-k\in B_\F \cap B_\alpha}} 1
\]
counting the number of particle--hole pairs of relative momentum $k$ in $B_\alpha$. The normalization constant $n_\alpha(k)$ should be large (the more summands contribute to \cref{eq:normali}, the less the $b^*$-operators are affected by the Pauli principle, and the more bosonic they behave). The following lemma is a variation of \cite[Prop.~3.1]{BNPSS0} and \cite[Lemma 5.1]{BNPSS}.
\begin{lemma}[Number of pairs per patch]\label{lem:counting} 
Assume that $N^{2\delta} R^2 \ll M \ll N^{\frac{2}{3}-2\delta} R^{-4}$. Then for all $k \in \bZ^3$ with $|k| < R$ and $\alpha \in \cI_k$, we have 
\[
n_\alpha(k)^2  =  \frac{4\pi k_\F^2}{M} \lvert k \cdot \hat\omega_\alpha \rvert  \left( 1 + o(1) \right)\;.
\]
\end{lemma}
\begin{proof}
The proof follows the argument given in \cite[Section~6]{BNPSS0}; only the control of the error terms needs to be refined in two respects.

First, in order for the vector $k$ to point from inside the Fermi ball to outside the Fermi ball even at the boundaries of the patch, we need $N^{2\delta} R^2 \ll M$, as can be verified by elementary geometry. This condition is illustrated in \cref{fig:smallangle}.

Second, the error term arising from the loss of particle--hole pairs near the boundary of the patch (thus proportional to the number of pairs in the patch of thickness $\lvert k\rvert \leq R$ not more than a distance $\lvert k\rvert \leq R$ from the patch boundary on the Fermi sphere) implies
\begin{align} \label{eq:nalphak}
 n_\alpha(k)^2 = \frac{4\pi k_{\textnormal{F}}^2}{M} \lvert k \cdot \hat{\omega}_\alpha \rvert + \mathcal{O}\left(\frac{N^{1/3}}{\sqrt{M}} \lvert k\rvert^2 \right)= \frac{4\pi k_{\textnormal{F}}^2}{M} \lvert k \cdot \hat{\omega}_\alpha \rvert \left(1 + \Ocal\left(\frac{\sqrt{M} \lvert k\rvert^2}{N^{1/3} \lvert k \cdot \hat{\omega}_\alpha\rvert}\right)\right)\;.
\end{align}
The error term becomes $o(1)$ since by assumption $\sqrt{M} R^2 N^{-1/3} N^\delta \ll 1$.
\begin{figure}
\vspace{-3cm}
\centering
\begin{tikzpicture}[scale=.5]
\clip (-6,-6) rectangle (6,16.5);
\draw[very thick] (0,-6) circle (11);
\node [above] at (-2.8,5) {$\scriptscriptstyle{k}$};
\draw[thick,->] (0,-6)--(0,5);
\node [right] at (0,4.2) {$\scriptscriptstyle{\omega_\alpha}$};
\draw(-5,5)--(5,5);
\node [above] at (-.3,-4.3) {{$\scriptscriptstyle{\theta_1}$}};
\draw (0,-6) +(90:1.8) arc (90:110:1.8);
\draw (-2.03,5) +(0:1.8) arc (0:20:1.8);
\node [right] at (-.3,5.4) {$\scriptscriptstyle{\theta_2}$};
\draw (0,-6)--(-4.5,6.5);
\draw (0,-6)--(4.5,6.5);
\draw [fill]  (-3.73,4.37) circle [radius=0.07];
\draw[very thick,->](-3.73,4.37) -- (-2,5.5);
\draw(-6,3.57) -- (1.8,6.37);
 \end{tikzpicture}
   \caption{Illustration for the condition $N^{2\delta}R^2 \ll M$ of \cref{lem:counting}. The angle between patch center and patch boundary is $\theta_1 \sim 1/\sqrt{M}$. The angle between the tangent at the center and at the boundary is $\theta_2 = \theta_1$ by elementary geometry. We know $k\cdot\hat{\omega}_\alpha \geq N^{-\delta}$ by definition of $\Ical_k$. This means that the angle between $k$ and the tangent at the center (being perpendicular to $\omega_\alpha$) is at least of order $\sim N^{-\delta}/R$. To have $k$ pointing from the inside to the outside of the Fermi ball even at the boundary we need $N^{-\delta}/R \gg 1/\sqrt{M}$.
   }\label{fig:smallangle}
\end{figure}
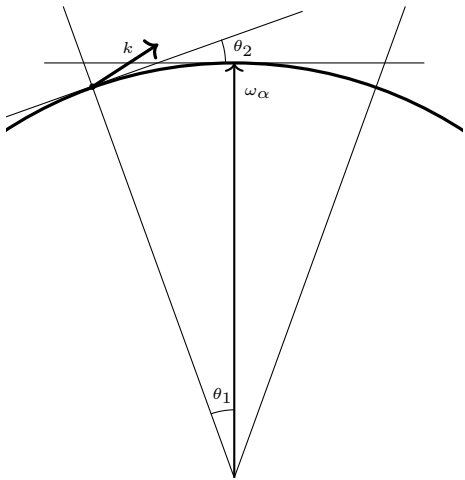
\end{proof}
It will be convenient to combine modes associated with $k$ and $-k$. To this end, we set 
\begin{equation} \label{eq:calpha}
c_\alpha^* (k) := \left\{ \begin{array}{ll} b^*_\alpha (k) &\quad \text{for $\alpha \in \cI_k^+$} \\
b_\alpha^* (-k) &\quad \text{for $\alpha \in \cI_k^-$} \end{array} \right. \end{equation} 
for every $k \in \north$. Here, we introduce the notation 
\begin{equation}\label{eq:north}
\begin{split}
\north := &\Big\{ k = (k_1, k_2, k_3) \in \bZ^3  \text{ with } |k| < R: \; k_3 > 0 \text{ or } (k_3 = 0  \text{ and } k_2 > 0) \\ &\hspace{7cm} \text{ or } (k_3 = k_2= 0 \text{ and } k_1 > 0) \Big\}
\end{split}
\end{equation} 
so that $\north \cap (-\north) = \emptyset$ and $\north \cup (-\north) = B_R (0) \backslash \{ 0 \}$. Note that compared to \cite{BNPSS}, in the definition of $\north$ we replaced the restriction $k \in \supp \hat{V}$ by $\lvert k\rvert < R$, with the parameter $R$ to be optimized at the end.  
    
Our analysis is based on the observation that the pair operators $c_\alpha^* (k)$ and $c_\alpha (k)$ behave approximately as bosonic creation and annihilation operators, on states with few excitations. This is established by the following lemma, taken from \cite[Lemma 4.1]{BNPSS0} and \cite[Lemma 5.2]{BNPSS}.
\begin{lemma}[Approximate bosonic CCR]
\label{lem:ccr}
Let $k, \ell \in \north$. Let $\alpha \in \cI_k$ and $\beta \in \cI_\ell$. Then
\begin{equation}
[c_\alpha(k),c_\beta( \ell )] = 0 = [c^*_\alpha(k),c^*_\beta (\ell)]\;, \quad [c_\alpha(k),c^*_\beta (\ell)]  = \delta_{\alpha,\beta}\big( \delta_{k,\ell} + \cE_\alpha (k,\ell) \big)\;, 
\label{eq:approximateCCR}
\end{equation}
where the error operator $\cE_\alpha(k,\ell)$ is controlled by the bounds
\begin{align} \label{eq:ccr-0}
\sum_{\alpha\in \cI_k \cap \cI_\ell} \lvert \cE_\alpha(k,\ell) \rvert^2  \leq C (M N^{-\frac{2}{3}+\delta} \cN)^2
\end{align}
and
\begin{align} \label{eq:ccr-1}
\sum_{\alpha\in \cI_k \cap \cI_\ell} \norm{ \cE_\alpha (k,\ell) \psi }  \leq C M^{\frac{3}{2}}N^{-\frac{2}{3}+\delta} \norm{\cN \psi} \qquad \text{for all $\psi \in \cF$.}
\end{align}
\end{lemma}

Another important property of the operators $c_\alpha^* (k)$ and $c_\alpha (k)$ is that they can be controlled in terms of the gapped number of particles operator $\cN_\delta$ introduced in \cref{eq:gappedN}, with $\delta > 0$ the parameter introduced in \cref{eq:Ik-def} to exclude a strip around the equator of the Fermi sphere in the definition of the sets $\cI_k$. The point is that, since we are away from the equator, $k$ has a component orthogonal to the Fermi sphere, which makes sure that the momentum of either the particle or of the hole annihilated by $c_\alpha (k)$ is at least at distance $N^{-\delta}$ from the Fermi sphere. More precisely, we have the following lemma, whose proof can be found in \cite[Lemmas 5.3 and 5.4]{BNPSS} (the first estimate in \cref{eq:c-N} and in \cref{eq:cg-N} are not stated explicitly in \cite[Lemmas 5.3 and 5.4]{BNPSS} but can be proven like the second bounds).
\begin{lemma}[Bounds on pair operators]\label{lm:c-N} 
Assume $M \gg R^2 N^{2\delta}$ and $R \ll N^{1/6-\delta/2}$. For all $k \in \north$ we have 
\begin{equation}\label{eq:c-N0}
\sum_{\alpha \in \cI_k} c_\alpha^* (k) c_\alpha (k) \leq \cN_\delta \;.
\end{equation} 
Moreover, for any $f \in \ell^2 (\cI_k)$, 
\begin{equation}\label{eq:c-N}
\Big\| \sum_{\alpha \in \cI_k} f_\alpha c_\alpha (k) \psi \Big\| \leq \| f \|_2 \| \cN_\delta^{1/2} \psi \|\; , \qquad  \Big\| \sum_{\alpha \in \cI_k} f_\alpha c^*_\alpha (k) \psi \Big\| \leq \| f \|_2 \| (\cN_\delta+1)^{1/2} \psi \| \;.
\end{equation} 
For $k \in \north$, $\alpha \in \cI_k$ and $g : \bZ^3 \times \bZ^3 \to \bR$, we define the weighted pair operator 
\[ c_\alpha^g (k) := \frac{1}{n_\alpha (k)} \sum_{\substack{p : p \in B_\F^c \cap B_\alpha \\ p- \sigma_\alpha k \in B_\F \cap B_\alpha}}  g(p,k) a_{p- \sigma_\alpha k} a_p \]  
with $\sigma_\alpha = 1$ if $\alpha \in \cI_k^+$, and $\sigma_\alpha = -1$ if $\alpha \in \cI_k^-$. 
Similarly to \cref{eq:c-N0} and \cref{eq:c-N}, we find
\[ \sum_{\alpha \in \cI_k} c_\alpha^{g*} (k) c^g_\alpha (k) \leq \| g \|_\infty^2 \cN_\delta \;.\]
Furthermore
\begin{equation}\label{eq:cg-Nsum}
\begin{split} 
\sum_{\alpha \in \cI_k} \Big\| c^g_\alpha (k) \psi \Big\| &\leq CM^{1/2} \| g \|_\infty \| \cN_\delta^{1/2} \psi \|^2 \;, \\ 
 \sum_{\alpha \in \cI_k} \Big\|  c^{g*}_\alpha (k) \psi \Big\| &\leq C M^{1/2} \| g \|_\infty \| (\cN_\delta+M)^{1/2} \psi \|^2
\end{split}
\end{equation} 
and, for $f \in \ell^2(\Ik)$,
\begin{equation}\label{eq:cg-N}
\begin{split} 
\Big\| \sum_{\alpha \in \cI_k} f_\alpha c^g_\alpha (k) \psi \Big\| &\leq \| f \|_2 \| g \|_\infty \| \cN_\delta^{1/2} \psi \|^2 \;, \\ 
 \Big\| \sum_{\alpha \in \cI_k} f_\alpha c^{g*}_\alpha (k) \psi \Big\| &\leq \| f \|_2 \| g \|_\infty \| (\cN_\delta+1)^{1/2} \psi \|^2 \;.
\end{split}
\end{equation} 
\end{lemma} 

\section{Reduction to an Almost Bosonic Quadratic Hamiltonian} 
Comparing \cref{eq:bb} with \cref{eq:calpha}, we find 
\[
b^* (k)  \simeq \sum_{\alpha \in \cI_k^+} n_\alpha (k) c^*_\alpha (k)\; , \qquad b^*(-k) \simeq  \sum_{\alpha \in \cI_k^-} n_\alpha (k) c^*_\alpha (k)
\]
for all $k \in \north$ (these are only approximate decompositions since, on the r.\,h.\,s., pairs in corridors and close to the equator are missing). Inserting this decomposition in \cref{eq:RHR-main} we find the following approximation for $Q_\textnormal{B}$, quadratic in $c$-- and $c^*$--operators:
\begin{equation}\label{eq:QBR}
\begin{split} 
Q_\textnormal{B}^R = \; & \frac{1}{N} \sum_{k \in \north} \hat{V} (k) \Bigg( \sum_{\alpha, \beta \in \cI_k^+} n_\alpha (k) n_\beta (k) c_\alpha^* (k) c_\beta (k) + \sum_{\alpha, \beta \in \cI_k^-} n_\alpha (k) n_\beta (k) c_\alpha^* (k) c_\beta (k) \\ &\hspace{2.5em} + \sum_{\alpha \in \cI_k^+ , \beta \in \cI_k^-} n_\alpha (k) n_\beta (k) c^*_\alpha (k) c^*_\beta (k) +  \sum_{\alpha \in \cI_k^-, \beta \in \cI_k^+} n_\alpha (k) n_\beta (k) c_\alpha (k) c_\beta (k) \Bigg)\,.
\end{split}
\end{equation}  
The difference between $Q_\textnormal{B}$ and $Q_\textnormal{B}^R$ is estimated in the following lemma, which we take from \cite[Lemma 4.1]{BNPSS}. Compared to \cite{BNPSS}, here we only need to compare $Q_\textnormal{B}$ with $Q_\textnormal{B}^R$ since we already controlled $\cE_2$ in \cref{cor:cE}; therefore the bound also does not use $\Ecal_1$.
\begin{lemma}[Removing corridors and removing patches near the equator]\label{lm:QB}
Assume that $\sum_{k \in \bZ^3} \lvert \hat{V} (k)\rvert |k| < \infty$. Then there exists $C > 0$ such that for all $\psi \in \fock$ we have
\[ \begin{split} 
\lvert \langle \psi, \left(Q_\textnormal{B} - Q_\textnormal{B}^R \right) \psi\rangle \rvert \leq C (N^{-\delta/2} + R^{1/2} M^{1/4} N^{-1/6+ \delta /2} + R^{-1 /2}) \langle \psi,(\bH_0 + \hbar) \psi\rangle \;. \end{split} \]
\end{lemma}
\begin{proof}
We consider the difference 
\[ b(k) - \sum_{\alpha \in \cI_k^+} n_\alpha (k) c_\alpha (k) = \sum_{p \in U_k} a_{p-k} a_p \]
where $U_k$ consists of all momenta $p \in B_\F^c$ with $p-k \in B_\F$ that do not belong to any patch. 
For $|k| < R$, we bound 
\[ \Big\| \Big( b(k) - \sum_{\alpha \in \cI_k^+} n_\alpha (k) c_\alpha (k) \Big) \psi  \Big\| \leq
 \sum_{p \in Y_k} \| a_{p-k} a_p \psi \| +  \sum_{p \in U_k \backslash Y_k}  \| a_{p-k} a_p \psi \| \]
 with 
\[ Y_k := \{ p \in U_k  :  e(p ) + e(p-k) \leq 4 N^{-1/3-\delta} \} \]
containing pairs close to the equator. Proceeding as in the proof of \cite[Lemma 4.1]{BNPSS} and using \cref{eq:num2}, we obtain 
\[ \sum_{p \in Y_k} \| a_{p-k} a_p \psi \| \leq C N^{1/2-\delta/2} \| \bH_0^{1/2} \psi \| \]
and (again under the assumption that $|k| < R$) 
\[ \sum_{p \in U_k \backslash Y_k}  \| a_{p-k} a_p \psi \| \leq C |k|^{1/2}  R^{1/2} M^{1/4} N^{1/3+\delta/2}  \| \bH_0^{1/2} \psi \|\;. \]
Here we estimated $|U_k \backslash Y_k| \leq C R |k| N^{1/3} M^{1/2}$ (for $|k| < R$, the set $U_k \backslash Y_k$ contains momenta $p \in \bZ^3$ localized in a shell of thickness $|k|$ around the Fermi sphere, so that either the projection of $p$ or the projection of $p-k$ onto the Fermi sphere falls in  corridors of size $R$ between patches). For $|k| > R$, on the other hand, we use \cref{lm:bb}. We conclude that  
\[ \begin{split} \Big\| \Big( b(k) - &\sum_{\alpha \in \cI_k^+} n_\alpha (k) c_\alpha (k) \Big) \psi  \Big\| \\ &\leq C \left( N^{1/2-\delta/2} + |k|^{1/2} R^{1/2} M^{1/4}  N^{1/3+\delta/2} + \chi (|k| > R) N^{1/2} \right) \| \bH_0^{1/2} \psi \|\;. \end{split} \]
Proceeding as in the last part of the proof of \cite[Lemma 4.1]{BNPSS}, using \cref{lm:bb} and  the assumption $\sum_{k \in \bZ^3} \hat{V} (k) |k| < \infty$, we arrive at the intended bound.
\end{proof}
 
To understand how the kinetic energy $\bH_0$, defined in \cref{eq:RHR-main}, can be expressed through the patch--wise particle--hole creation and annihilation operators, we compute the commutator 
\begin{equation*}  \begin{split} [ \bH_0 , c_\alpha^* (k)] &= \left[ \sum_{q \in \bZ^3} e(q) a_q^* a_q , \frac{1}{n_\alpha (k)} \sum_{p \in B_\F^c \cap (B_\F+k) \cap B_\alpha} a_p^* a_{p-k}^* \right] \\ &= \frac{1}{n_\alpha (k)} \sum_{ p \in B_\F^c \cap (B_\F+k) \cap B_\alpha} (e(p) + e(p-k)) a_p^* a_{p-k}^* \;. \end{split} \end{equation*} 
With $e(p) + e(p-k) = \hbar^2 p^2 - \hbar^2 (p-k)^2 \simeq 2 \hbar \kappa |k \cdot \hat{\omega}_\alpha|$ (with $\hat{\omega}_\alpha = \omega_\alpha / |\omega_\alpha|$ the normalized vector pointing to the center of the $\alpha$-th patch), we obtain
\begin{equation} \label{eq:H0-comm1}  [ \bH_0 , c_\alpha^* (k)] \simeq 2\hbar \kappa \, |k \cdot \hat{\omega}_\alpha| c_\alpha^* (k) \end{equation} 
which suggests that, in a sense to be made precise, 
\begin{equation}\label{eq:H0-comm2}  \bH_0 \simeq  2 \kappa \hbar \sum_{k \in \north} \sum_{\alpha =1}^M |k \cdot \hat{\omega}_\alpha| \, c_\alpha^* (k) c_\alpha (k) = : \bD_\textnormal{B} \;.\end{equation} 
Based on this heuristic observation, we expect that the correlation Hamiltonian \cref{eq:corr} 
can be approximated by
\begin{equation}\label{eq:quad-app} \cH_\textnormal{corr} \simeq \bD_\textnormal{B} + Q^R_B = \sum_{k \in \north} 2\hbar \kappa |k| h_\text{eff} (k) \end{equation} 
with the quadratic (in $c$-- and $c^*$--operators) expression 
\begin{equation}\label{eq:heff-def} h_\text{eff} (k)  = \sum_{\alpha,\beta \in \cI_k} \left( (D(k) + W(k))_{\alpha,\beta} c_\alpha^* (k) c_\beta (k) + \frac{1}{2} \wt{W} (k)_{\alpha,\beta} \big( c_\alpha^* (k) c_\beta^* (k) + c_\beta (k) c_\alpha (k) \big) \right) \end{equation}
where $D(k)$, $W(k)$, and $\wt{W} (k)$ are $|\cI_k| \times |\cI_k|$ real symmetric matrices with entries 
\begin{equation}\label{eq:DWwtW} \begin{split} 
D(k)_{\alpha,\beta} &= \delta_{\alpha,\beta} |\hat{k} \cdot \hat{\omega}_\alpha | \;, \quad \text{for all } \alpha, \beta \in \cI_k \\
W(k)_{\alpha, \beta} &= \frac{\hat{V} (k)}{2\hbar \kappa N |k|} \times \left\{ \begin{array}{ll} n_\alpha (k) n_\beta (k) \quad &\text{ if $\alpha, \beta \in \cI_k^+$ or $\alpha,\beta \in \cI_k^-$} \\
0 \quad &\text{ otherwise}\,, \end{array} \right. \\
\wt{W} (k)_{\alpha,\beta} &= \frac{\hat{V} (k)}{2\hbar \kappa N |k|} \times \left\{ \begin{array}{ll}  0 \quad &\text{ if $\alpha, \beta \in \cI_k^+$ or $\alpha,\beta \in \cI_k^-$} \\
n_\alpha (k) n_\beta (k) \quad &\text{ otherwise}\,. \end{array} \right. \end{split} \end{equation}

\section{Approximate Bogoliubov Transformations} 
\label{sec:BT} 
If the $c$-- and $c^*$--operators were exactly bosonic, we could write 
\[ h_\text{eff} (k) = \bH - \frac{1}{2} \tr \, (D (k) +W (k)) \]
with the quadratic Hamiltonian (in the following discussion we omit the fixed argument $k$)
\begin{equation}\label{eq:H-quad}
\bH := \frac{1}{2} ((c^*)^T , c^T)  \left( \begin{array}{ll} D+W & \wt{W} \\ \wt{W} & D+W \end{array} \right) \left( \begin{array}{c} c \\ c^* \end{array} \right) \;.
\end{equation} 
Introducing the $|\cI_k | \times |\cI_k |$ matrix 
\[ E := \left[ (D + W - \wt{W})^{1/2} (D + W + \wt{W}) (D + W -\wt{W})^{1/2} \right]^{1/2} \,  \]
and setting $S_1 := (D+W-\wt{W})^{1/2} E^{-1/2}$, $S_2 := (D+W-\wt{W})^{-1/2} E^{1/2}$ (so that $S_1 S_2^T = S_2 S_1^T = 1$) and 
\begin{equation}\label{eq:Sdef} S := \left( \begin{array}{ll} S_1 & 0 \\ 0 & S_2 \end{array} \right) \end{equation} 
we can decompose 
\begin{equation} \label{eq:dec-quad} \left( \begin{array}{ll} D+W & \wt{W} \\ \wt{W}  & D+W \end{array} \right) = \left( \begin{array}{ll} \frac{S_1 + S_2}{2} & \frac{S_1 - S_2}{2} \\ \frac{S_ 1 - S_2}{2} & \frac{S_1 + S_2}{2} \end{array} \right)^T  \left( \begin{array}{ll} E & 0 \\ 0 & E \end{array} \right) \left( \begin{array}{ll} \frac{S_1 + S_2}{2} & \frac{S_1 - S_2}{2} \\ \frac{S_ 1 - S_2}{2} & \frac{S_1 + S_2}{2} \end{array} \right) \;.\end{equation} 

Using the polar decomposition $S_1 = O |S_1|$ with an orthogonal matrix $O$ and the positive matrix $|S_1| = (S_1^T S_1)^{1/2}$ we obtain $S_2 = O |S_1|^{-1}$ from $S_2 S_1^T = 1$. Moreover, $|S_1^T| = O |S_1| O^T$ and thus $S_1 = |S_1^T| O$, $S_2 = |S_1^T| O$ and, from \cref{eq:dec-quad}, 
\[ \begin{split}  \left( \begin{array}{ll} D+W & \wt{W} \\ \wt{W}  & D+W \end{array} \right) &= \left( \begin{array}{ll} \frac{|S_1^T| + |S_1^T|^{-1}}{2} & \frac{|S^T_1| - |S_1^T|^{-1}}{2} \\ \frac{|S^T_ 1| - |S^T_1|^{-1}}{2} & \frac{|S_1^T| + |S_1^T|^{-1}}{2} \end{array} \right)  \left( \begin{array}{ll} O & 0 \\ 0 & O \end{array} \right)  \left( \begin{array}{ll} E & 0 \\ 0 & E \end{array} \right) \\ &\hspace{3cm} \times  \left( \begin{array}{ll} O & 0 \\ 0 & O \end{array} \right)^T \left( \begin{array}{ll} \frac{|S_1^T| + |S_1^T|^{-1}}{2} & \frac{|S^T_1| - |S_1^T|^{-1}}{2} \\ \frac{|S^T_ 1| - |S^T_1|^{-1}}{2} & \frac{|S_1^T| + |S_1^T|^{-1}}{2} \end{array} \right)\;. \end{split} \]
Defining
\[K := \log |S_1^T|\]
we obtain 
\begin{equation}\label{eq:deco-quad2} \begin{split} \left( \begin{array}{ll} D+W & \wt{W} \\ \wt{W} & D+W \end{array} \right) &= 
 \left( \begin{array}{ll} \cosh (K)  & \sinh (K)  \\ \sinh (K)  & \cosh (K) \end{array} \right) \left( \begin{array}{ll} O & 0 \\ 0 & O \end{array} \right) \left( \begin{array}{ll} E & 0 \\ 0 & E \end{array} \right) \\ &\hspace{3cm} \times \left( \begin{array}{ll} O & 0 \\ 0 & O \end{array} \right)^T    \left( \begin{array}{ll} \cosh (K)  & \sinh (K)  \\ \sinh (K)  & \cosh (K) \end{array} \right)\;.  \end{split} \end{equation} 
Hence, a symplectic conjugation of the $2|\cI_k| \times 2 |\cI_k|$ matrix defining the quadratic Hamiltonian \cref{eq:H-quad} is sufficient to obtain a block--diagonal matrix (with $|\cI_k| \times |\cI_k|$ blocks $O E O^T$) corresponding to a ``diagonal'' quadratic Hamiltonian in the sense of containing only terms of the form $c^* c$ and none of the form $c^* c^*$ or $c c$.

However, it will be important to further transform the block--diagonal matrix as to make the resulting quadratic Hamiltonian comparable with the bosonic kinetic energy $\bD_\textnormal{B}$, defined in \cref{eq:H0-comm2}.
To reach this goal we have to look more closely at $E$, decomposing it further into blocks associated to the index sets $\cI_k^+$ and $\cI_k^-$ (associated with patches in the north and south hemisphere, respectively). Note that $I = |\cI_k^+| = |\cI_k^-| = |\cI_k| /2$. With \cref{eq:DWwtW} we write
\begin{equation}\label{eq:db-intro} D = \left( \begin{array}{ll} d & 0 \\ 0 &  d \end{array} \right) , \quad W = \left( \begin{array}{ll} b & 0 \\ 0 & b \end{array} \right) , \quad \wt{W} = \left( \begin{array}{ll} 0 & b \\ b & 0 \end{array} \right) \end{equation} 
where $d = \text{diag} \{ u_\alpha^2 ,  \alpha = 1, \dots , I \}$ and $b = g |v \rangle \langle v |$. Here we introduced 
\[ g = \frac{\kappa}{2} \hat{V} (k)\; , \quad u_\alpha = |\hat{k} \cdot \hat{\omega}_\alpha|^{1/2}\;, \quad v_\alpha = \frac{\hbar}{\kappa \sqrt{|k|}}  n_\alpha (k) \qquad \textnormal{for } \alpha = 1, \dots , I\;. \]
It will play an important role in the proof of \cref{lm:LwtL} that, as a consequence of \cref{eq:Ik-def} and \cref{lem:counting}, we have  
\begin{equation}\label{eq:prop-uv} 
N^{-\delta} \leq u_\alpha^2 \leq 1\;, \qquad  |v_\alpha | \leq C \frac{u_\alpha}{M^{1/2}}  
\end{equation} 
which implies $\norm{v} \leq C$ and $\norm{ d^{-1/2} v } \leq C$.

To block--diagonalize $E$ (w.\,r.\,t.\ the decomposition $\cI_k = \cI_k^+ \cup \cI_k^-$), we introduce
\begin{equation} \label{eq:U-def} U := \frac{1}{\sqrt{2}} \left( \begin{array}{ll} \id & \id \\ \id & -\id \end{array} \right) \end{equation} 
(where $\id$ is the $I\times I$ identity matrix)
and observe
\[ U^T (D+W+\wt{W}) U = \left( \begin{array}{ll} d+2b & 0 \\ 0 & d \end{array} \right)\; , \qquad U^T (D+W - \wt{W} ) U = \left( \begin{array}{ll} d & 0 \\ 0 & d+ 2b \end{array} \right)\;. \]
This implies that 
\begin{equation} \label{eq:UEU} U^T E U = \left( \begin{array}{ll} [d^{1/2} (d+2b) d^{1/2} ]^{1/2} & 0 \\ 0 & [ (d+2b)^{1/2} d (d+2b)^{1/2} ]^{1/2} \end{array} \right)\;.  \end{equation}
The upper--left entry is clearly larger than the operator $d$. It seems more difficult to compare the lower--right entry with $d$ (thus, it seems difficult to compare $U^T E U$ with $D$). To solve this problem, we define the $I \times I$ matrix $X := (d+2b)^{1/2} d^{1/2}$ and consider its polar decomposition $X = A P$, with $A$ orthogonal and $P := (X^* X)^{1/2}$. Then, from \cref{eq:UEU}, we have 
\[ \begin{split}  U^T E U  &=   \left(  \begin{array}{ll} (X^* X)^{1/2}  & 0 \\ 0 & (X X^*)^{1/2} \end{array} \right) \\ &= \left(  \begin{array}{ll} P  & 0 \\ 0 & A P A^T \end{array} \right) = \left(  \begin{array}{ll} 1  & 0 \\ 0 & A  \end{array} \right) \left(  \begin{array}{ll} P  & 0 \\ 0 & P  \end{array} \right)\left(  \begin{array}{ll} 1  & 0 \\ 0 & A^T \end{array} \right)\;. \end{split} \]
Using the easily--checked invariance of the matrix with blocks $P$ on the diagonal w.\,r.\,t.\ conjugation with $U$ we conclude that 
\[ E = \wt{O} \wt{P} \wt{O}^T\;, \]
where we defined 
\begin{equation} \label{eq:wtO-def} \wt{O} := U  \left(  \begin{array}{ll} 1  & 0 \\ 0 & A  \end{array} \right) U^T \;, \qquad \wt{P} := \left(  \begin{array}{ll} P  & 0 \\ 0 & P  \end{array} \right) \;.\end{equation} 
Inserting in \cref{eq:deco-quad2}, we arrive at 
\begin{align*}
\left( \begin{array}{ll} D+W & \wt{W} \\ \wt{W} & D+W \end{array} \right) &= 
 \left( \begin{array}{ll} \cosh (K)  & \sinh (K)  \\ \sinh (K)  & \cosh (K) \end{array} \right) \left( \begin{array}{ll} O & 0 \\ 0 & O \end{array} \right)  \left( \begin{array}{ll} \wt{O} & 0 \\ 0 & \wt{O} \end{array} \right)  \left( \begin{array}{ll} \wt{P} & 0  \\ 0 &  \wt{P} \end{array} \right) \\
 & \quad \times 
 \left( \begin{array}{ll} \wt{O} & 0 \\ 0 & \wt{O} \end{array} \right)^T 
 \left( \begin{array}{ll} O & 0 \\ 0 & O \end{array} \right)^T    \left( \begin{array}{ll} \cosh (K)  & \sinh (K)  \\ \sinh (K)  & \cosh (K) \end{array} \right)\;.   \tagg{eq:deco-quad3}
 \end{align*}

If the $c$-- and $c^*$--operators were exactly bosonic we could therefore bring the quadratic operator \cref{eq:H-quad} into a diagonal form comparable to the bosonic kinetic energy $\bD_\textnormal{B}$ by means of the two Bogoliubov transformations\footnote{The transformation $Z$ is a trivial Bogoliubov transformation, corresponding to only a change of basis in the one-boson Hilbert space. In the language of bosonic second-quantized operators, it corresponds to a transformation of the form $e^{\di\Gamma(L)} = \Gamma(e^L)$, where $e^L$ is an orthogonal matrix acting on the one-boson space.} 
\begin{equation}\label{eq:3BT}
\begin{split} 
T &= \exp \left( \frac{1}{2} \sum_{k \in \north} \sum_{\alpha,\beta \in \cI_k} K (k)_{\alpha,\beta} \, c_\alpha^* (k) c_\beta^* (k) - \text{h.c.} \right) \;, \\
Z &= \exp \left( \sum_{k \in \north} \sum_{\alpha,\beta \in \cI_k} L_{\alpha,\beta} (k) \, c_\alpha^* (k) c_\beta (k) \right) \;,
\end{split}
\end{equation} 
where (re-inserting now the dependence on $k$ in the notation) we introduced the matrix
\begin{equation}
L(k) := \log \left( O(k) \wt{O}(k) \right) \;.
\end{equation}
Recall that $O (k)$ and $ \wt{O} (k)$ are orthogonal matrices, i.\,e., all their eigenvalues are on the unit circle. The function $\log$ denotes an arbitrary branch of the complex logarithm with $\Im \log 1 = 0$. The matrix $L (k)$ is by definition antisymmetric, so that $Z$ is a unitary operator on Fock space. If the $c$-- and $c^*$--operators were exactly bosonic, we would find 
\begin{equation}\label{eq:diagonal-bos}
Z^* T^* \bH T Z  = \frac{1}{2} \sum_{\alpha, \beta \in \cI_k} \wt{P}_{\alpha, \beta}  \left( c_\alpha^* (k) c_\beta (k) + \delta_{\alpha,\beta} \right) \;.
\end{equation}
Recall that $\tr \tilde{P} = \tr E$. Since $P = (X^* X)^{1/2} = [ d^{1/2} (d+2b) d^{1/2} ]^{1/2} \geq d$, we could use $\wt{P} \geq D$ to conclude that 
\begin{equation}\label{eq:diagonal-bos2}
Z^* T^* \bH T Z \geq \sum_{\alpha \in \cI_k} u_\alpha^2 (k) c_\alpha^* (k) c_\alpha (k) + \frac{1}{2}\tr E = \bD_\textnormal{B} + \frac{1}{2}\tr E\;.
\end{equation}
This comparison is not surprising in view of the discussion of the spectrum of $E(k)$ in \cite{Ben20}. There the problem is reduced to a rank--one perturbation of the matrix $D(k)$; the perturbed eigenvalues are all larger than the corresponding unperturbed eigenvalues. However, $E(k)$ and $D(k)$ cannot be simultaneously diagonalized, so we do not have an operator inequality between $E(k)$ and $D(k)$. This problem is overcome here noting that $E(k)$ can be diagonalized by a Bogoliubov transformation which leaves $\Hbb_0 - \Dbb_\B$ (though not $\Dbb_\B$ alone) invariant. 

Since the $c$-- and $c^*$--operators are not exactly bosonic, we can expect \cref{eq:diagonal-bos} to hold only approximatively, on states with few excitations of the Fermi ball. To prove that this is indeed the case, we need some estimates on the kernels $K(k)$ and $L(k)$. The following bound for $K(k)$ has already been shown in \cite[Lemma~2.5]{BNPSS2}.
\begin{lemma}[Bogoliubov kernel] \label{lm:K}
There exists a $C >0$ such that for all $k \in \north$ we have 
\[ |K(k)_{\alpha,\beta}| \leq C \frac{\hat{V} (k)}{M} \qquad \text{for all $\alpha,\beta \in \cI_k$.}\]
In particular $\| K (k) \|_\textnormal{HS} \leq C \hat{V} (k)$. 
\end{lemma}

The following bounds for the antisymmetric matrix $L(k)$ are new. 
\begin{lemma}[Kernel of one--particle transformation]\label{lm:LwtL} 
Suppose that the parameters $\delta, M, R$ used to define the patch decomposition in Section \ref{sec:patch} are such that $M \gg R^2 N^{2\delta}$. Then there exists a $C >0$ such that   for all $k \in \north$ we have
\begin{equation} \label{eq:L-bd}
\begin{split} 
\| L (k) \|_\HS &\leq C \hat{V} (k) \, . 
\end{split}
\end{equation}  
\end{lemma}

\begin{rem} Since $L (k)$ is the logarithm of an orthogonal matrix, we always have $\| L(k) \|_\op \leq 2\pi$. From Lemma \ref{lm:LwtL}, we also have $\| L(k) \|_\op \leq C \hat{V} (k)$, which improves the bound if $\hat{V} (k)$ is small.
\end{rem}

\begin{proof}
All matrices depend on $k$ but in this proof we do not indicate this dependence explicitly. We split the bound in two parts by
\[ \begin{split} 
 \norm{L}_\textnormal{HS} & = \| \log (O\wt{O}) \|_\HS \leq C \| O\wt{O}-1 \|_\HS \leq C \| O \|_\op \| \wt{O}-1 \|_\HS + C \| O - 1 \|_\HS\;.
\end{split} \]
Since $O$ is orthogonal we have $\| O \|_\op =1$ and we only need to estimate $\| \wt{O}-1 \|_\HS$ and $\| O - 1 \|_\HS$. The same applies for the operator norm.

\paragraph{Bound for $\| \wt{O}-1 \|_\HS$.}
From the definition \cref{eq:wtO-def}, we get 
\begin{equation}\label{eq:OthrouA} \| \wt{O} - 1 \|_\HS = \| A - 1 \|_\HS \end{equation} 
with $A$ the orthogonal matrix arising from the polar decomposition of $X = (d+2b)^{1/2} d^{1/2}$, i.\,e., $A = X (X^* X)^{-1/2}$. We have 
\begin{equation}\label{eq:wtO1} \begin{split} 
\| A - 1 \|_\HS &= \left\| X \frac{1}{\sqrt{X^* X}} - 1 \right\|_\HS \leq \left\| X \left( \frac{1}{\sqrt{X^* X}} - \frac{1}{d} \right) \right\|_\HS + \left\| X \frac{1}{d} - 1 \right\|_\HS \;.\end{split} \end{equation}
To bound the second term on the r.\,h.\,s.\ of the last equation, we use the representation
\begin{equation}\label{eq:repre} \sqrt{z} = \frac{1}{\pi} \int_0^\infty \frac{\di s}{\sqrt{s}} \left( 1 -  \frac{s}{s+z} \right) \end{equation} 
to write by means of a resolvent identity
\begin{equation}\label{eq:Xd-1}  \begin{split} 
X \frac{1}{d} - 1 =  \left( (d+2b)^{1/2}  - d^{1/2} \right) \frac{1}{d^{1/2}}&= -\frac{1}{\pi}  \int_0^\infty \di s \, \sqrt{s} \, \left( \frac{1}{s+d+2b} - \frac{1}{s+d} \right) \frac{1}{d^{1/2}} \\ &= \frac{2}{\pi} \int_0^\infty \di s \sqrt{s}\, \frac{1}{s+d+2b} \, b \, \frac{1}{s+d} \, \frac{1}{d^{1/2}} \;.\end{split} \end{equation} 
Recalling that $b = g |v \rangle \langle v |$ with $g = \kappa \hat{V} (k) /2$ we find 
\begin{equation}
\label{eq:Xd1}
\big\| X \frac{1}{d} - 1 \big\|_\text{HS} \leq C \hat{V} (k) \int_0^\infty \di s \, \sqrt{s} \, \Big\| \frac{1}{s+d+2b} v \Big\| \,  \Big\| \frac{1}{s+d} \frac{1}{d^{1/2}} v \Big\|\;.
\end{equation}
To control the norms in this integral (and similar norms that will arise in the rest of the proof), we use \cref{eq:prop-uv} so that, for $j=1,2$ and $-1/2 \leq k \leq j-1$, we have
\begin{equation}
\label{eq:riem0}
\Big\| \frac{1}{s+d^j} d^k v \Big\|^2 = \Big\langle v , \frac{d^{2k}}{(s+d^j)^2} v \Big\rangle = \sum_{\alpha \in \Ikp}  \frac{v_\alpha^2 u_\alpha^{4k}}{(s+u_\alpha^{2j})^2} \leq \frac{C}{M} \sum_\alpha  \frac{u_\alpha^{4k+2}}{(s+u_\alpha^{2j})^2}  \;.
\end{equation}
Recall that $u_\alpha^2 = |\hat{k} \cdot \hat{\omega}_\alpha| = \cos \theta_\alpha$ where $\theta_\alpha \in (0;\pi/2)$ is the inclination angle of the center $\omega_\alpha$ of the patch $B_\alpha$, measured w.\,r.\,t.\ the vector $k$. We consider then the sum on the r.\,h.\,s.\ of \cref{eq:riem0} as a Riemann sum for a surface integral on the northern hemisphere of the unit sphere, parametrized by the angles $\theta \in (0,\pi/2)$ and $\varphi \in (0,2\pi)$. To estimate the error in going from the Riemann sum to the integral, we set
\[
f(\theta) = \frac{\cos^{2k+1} \theta}{(s+\cos^j \theta)^2}
\]
and compute its derivative, finding
\[
f' (\theta) = f(\theta) \left( (2k+1) \frac{\sin \theta}{\cos \theta} - 2j \frac{\cos^{j-1} \theta \sin \theta}{(s+\cos^j \theta)} \right)\;.
\]
Let $p_\alpha$ denote the surface area on the unit sphere $\mathbb{S}_2$ covered by the patch $B_\alpha$. With slight abuse of notation, let us also write $p_\alpha$ for the set of inclination angles $\theta \in (0,\pi/2)$ corresponding to points in $p_\alpha$. For all $\theta, \tilde\theta \in p_\alpha$ we have $|\theta - \tilde\theta| \leq C M^{-1/2}$ (this being the order of the diameter of the patch). According to the definition \cref{eq:Ik-def} of the index set, for $\alpha \in \Ikp$ we have $\cos \theta_\alpha \geq R^{-1} N^{-\delta}$. Thus for all $\theta \in p_\alpha$ we have
\[
\cos \theta \geq \cos \theta_\alpha - \lvert \cos \theta - \cos \theta_\alpha \rvert \geq R^{-1} N^{-\delta} - CM^{-1/2} \geq \frac{1}{2} R^{-1} N^{-\delta}\,,
\]
where we recall the assumption $M \gg R^2 N^{2\delta}$. Moreover, by the mean value theorem (if necessary enlarging the set of angles $p_\alpha$ to its convex hull in all the following supremuma to make sure that $\theta_0$ is contained)
\[
| f(\theta) - f(\tilde\theta)| \leq \sup_{\theta_0 \in p_\alpha} \lvert f'(\theta_0)\rvert \lvert \theta - \tilde\theta \rvert \leq C \frac{R N^\delta}{\sqrt{M}} \sup_{\theta_0 \in p_\alpha} f(\theta_0) \;.
\]
This implies $| f(\theta) - f(\tilde\theta_\alpha)| \leq 2^{-1} \sup_{\theta_0 \in p_\alpha} f(\theta_0)$. Thus for all $\theta \in p_\alpha$ we have
\[\sup_{\tilde\theta \in p_\alpha} f(\tilde\theta) \leq \sup_{\tilde\theta \in p_\alpha} |f(\tilde\theta) - f(\theta) | + f(\theta) \leq \frac{1}{2} \sup_{\tilde\theta \in p_\alpha} f(\tilde\theta) + f(\theta)\;;\]
in particular $f(\theta_\alpha) \leq 2 f(\theta)$ for all $\theta \in p_\alpha$. Therefore
\[
\Big\| \frac{1}{s+d^j} d^k v \Big\|^2  \leq C \sum_{\alpha \in \Ikp} \int_{p_\alpha} \frac{ \cos^{2k+1} \theta}{(s+\cos^j \theta)^2} \sin \theta \di\theta \di\varphi \leq C  \int_0^1 \frac{t^{2k+1}}{(s+t^j)^2} \di t \;.
\]
We conclude that 
\begin{equation}
\label{eq:riem}
\Big\| \frac{1}{s+d^j} d^k v \Big\| \leq C \left\{ \begin{array}{ll}  \min \{ s^{-1} , s^{- 1+ (1+k)/j} \}  \; &\text{if } \; 1+k < j \\
\min \{ s^{-1} , |\log s|^{1/2} \} \; &\text{if } \; 1+k = j \;. \end{array} \right.
\end{equation}
In particular, with $j=1$, $k=-1/2$, we find 
\[
\begin{split}
\Big\| \frac{1}{s+d} \frac{1}{d^{1/2}} v \Big\| \leq C \min \{ s^{-1} ; s^{-1/2} \} \;. \end{split}
\]
To bound the other norm in the integral in \cref{eq:Xd1}, we write
\[  \frac{1}{s+d+2b} v = \frac{1}{s+d} v - 2 \frac{1}{s+d+2b} b \frac{1}{s+d} v = \frac{1}{s+d} v - 2\Big\langle v, \frac{1}{s+d} v \Big\rangle \, \frac{1}{s+d+2b}  v \]
which implies, applying \cref{eq:riem} with $j=1$ and $k=0$,
\[
\Big\| \frac{1}{s+d+2b} v  \Big\| \leq \Big\| \frac{1}{s+d} v \Big\| \leq  C \min  \{ s^{-1}, |\log s|^{1/2}  \}  \;.
\]
Inserting this bound in \cref{eq:Xd1} and integrating the variable $s$ separately over the intervals $[0,1]$ and $[1,\infty)$, we conclude that
\[
\big\| X \frac{1}{d} - 1 \big\|_\text{HS} \leq C \hat{V} (k) \;.
\]

As for the first term on the r.\,h.\,s.\ of \cref{eq:wtO1}, we proceed analogously, writing 
\[ \begin{split} X \left( \frac{1}{\sqrt{X^* X}} - \frac{1}{d} \right) &= \frac{1}{\pi} \int_0^\infty \frac{\di s}{\sqrt{s}} \, X \left( \frac{1}{s+d^{1/2} (d+2b) d^{1/2}} - \frac{1}{s+d^2} \right) \\
& = -\frac{2}{\pi} \int_0^\infty \frac{\di s}{\sqrt{s}} \, (d+2b)^{1/2} d^{1/2} \, \frac{1}{s+d^{1/2} (d+2b) d^{1/2}} d^{1/2} \, b \, d^{1/2} \frac{1}{s+d^2} \;.
\end{split} \]
We write $b = g |v \rangle \langle v |$. We can bound $\| d^{-1/2} v \| \leq C$, as well as
 \[
 \begin{split}
 \left\| (d+2b)^{1/2} d^{1/2} \, \frac{1}{s+d^{1/2} (d+2b) d^{1/2}} \, d^{1/2} (d+2b)^{1/2} \right\|_\op &\leq 1 \;,
\\   \| (d+2b)^{-1/2} d^{1/2} \|_\op &\leq 1 \;,
\end{split}
\]
and, using \cref{eq:riem} with $j=2$ and $k=1/2$,
\[
\Big\| \frac{1}{s+d^2} d^{1/2} v \Big\| \leq C \min \{ s^{-1} , s^{-1/4} \} \;.
\]
We conclude that 
\[ \left\| X \left( \frac{1}{\sqrt{X^* X}} - \frac{1}{d} \right) \right\|_\HS \leq C \hat{V} (k) \, .  \]
Combined with \cref{eq:wtO1} and \cref{eq:Xd1}, this implies 
\[
\| A - 1 \|_\HS \leq C \hat{V} (k) \;.
\]

\paragraph{Bound for $\norm{O-1}_\HS$.} Recall that $O$ arises from the polar decomposition \cref{eq:Sdef} of $S_1$, i.\,e., 
\[ O = S_1 |S_1|^{-1} =  (D+W-\wt{W})^{1/2} E^{-1/2} \frac{1}{\sqrt{E^{-1/2} (D+W-\wt{W}) E^{-1/2}}}\;. \]
Using the orthogonal matrix $U$ defined in \cref{eq:U-def} and the fact that $O-1$ and $U^T (O-1) U$ have the same spectrum we obtain 
\begin{equation}\label{eq:LHS-deco}
\begin{split}
\| O-1 \|_\HS  \leq \; & \Big\| d^{1/2} (X^* X)^{-1/4} \frac{1}{\sqrt{(X^* X)^{-1/4} d (X^* X)^{-1/4}}} -1 \Big\|_\HS  \\ &+ \Big\| 
(d+2b)^{1/2} (X X^*)^{-1/4} \frac{1}{\sqrt{(X X^*)^{-1/4} (d+2b) (X X^*)^{-1/4}}} -1 \Big\|_\HS\;.
\end{split}
\end{equation}  
To estimate the first norm on the r.\,h.\,s.\ of \cref{eq:LHS-deco} we decompose 
\begin{equation} \label{eq:deco-2} 
\begin{split} 
& d^{1/2} (X^* X)^{-1/4} \frac{1}{\sqrt{(X^* X)^{-1/4} d (X^* X)^{-1/4}}} -1 \\
& = d^{1/2} \left( (X^* X)^{-1/4} - d^{-1/2} \right)\frac{1}{\sqrt{(X^* X)^{-1/4} d (X^* X)^{-1/4}}}  \\
& \quad + \frac{1}{\sqrt{(X^* X)^{-1/4} d (X^* X)^{-1/4}}} -1 \;.
\end{split} \end{equation} 
We start with the first summand on the r.\,h.\,s.\ of \cref{eq:deco-2}. 
With an integral representation similar to \cref{eq:repre} and using 
$X^* X - d^2 = 2d^{1/2} b d^{1/2}$, we write it as 
\begin{equation}\label{eq:term1a} \begin{split} 
d^{1/2} &\left( (X^* X)^{-1/4} - d^{-1/2} \right) \frac{1}{\sqrt{(X^* X)^{-1/4} d (X^* X)^{-1/4}}} \\
&= C \int_0^\infty \frac{\di s}{s^{1/4}} \, d^{1/2} \frac{1}{s+d^2} \, d^{1/2} \, b \, d^{1/2} \frac{1}{s+ X^* X} \frac{1}{\sqrt{(X^* X)^{-1/4} d (X^* X)^{-1/4}}}\;. \end{split} \end{equation} 
We estimate $\| d^{-1/2} v \| \leq C$ and  
\[
\begin{split} 
\Big\| d &\frac{1}{s+X^* X}  \frac{1}{\sqrt{(X^* X)^{-1/4} d (X^* X)^{-1/4}}}  \Big\|_\op^2  \\ &\leq \Big\| d \frac{1}{s+X^* X}  \frac{1}{(X^*X)^{-1/4} d (X^*X)^{-1/4}} \frac{1}{s+X^* X}  \, d \,  \Big\|_\op  \\ 
&\leq  \| d (X^* X)^{-1/2} \|_\op \Big\|  \frac{(X^* X)^{1/4}}{s+X^* X} \Big\|_\op \| (X^* X)^{1/2} d^{-1} \|_\op \Big\|  \frac{(X^* X)^{3/4}}{s+X^* X} \Big\|_\op \| (X^* X)^{-1/2} d \|_\op \\ &\leq C \min \{ s^{-2} ,s^{-1} \} \;.
\end{split}
\]
Here we used (recalling $X^* X = d^{1/2} (d+2b) d^{1/2}$) that $\| d (X^* X)^{-1/2} \|_\op \leq 1$ and also 
\begin{equation}
\label{eq:XXd}
\| (X^* X)^{1/2} d^{-1} \|_\op^2 = \| 1 + d^{-1/2} b d^{-1/2} \|_\op \leq C\;.
\end{equation}
Using \cref{eq:riem} with $j=2$, $k=1$, we obtain
\[
\Big\| \frac{1}{s+d^2} d v \Big\| \leq C \min  \{ s^{-1} , |\log s |^{1/2} \} \;.
\]
We conclude therefore that  
\begin{equation}
\label{eq:term1a-fin}
\Big\| d^{1/2} \left( (X^* X)^{-1/4} - d^{-1/2} \right) \frac{1}{\sqrt{(X^* X)^{-1/4} d (X^* X)^{-1/4}}} \Big\|_\HS \leq C\hat{V} (k)\;.
\end{equation}
 
 Let us now consider the second summand on the r.\,h.\,s.\ of \cref{eq:deco-2}. Since $X^* X = d^{1/2} (d+2b) d^{1/2} \geq d^2$, we observe that  
 \[
 d^{1/2} (X^*X)^{-1/2} d^{1/2} \leq 1\;.
 \]
From $d^{-1/2} b d^{-1/2} \leq C$ (uniformly in $N$ and in $k$, since $\hat{V}$ is bounded), we also have $X^* X \leq C d^2$ and thus 
 \[   d^{1/2} (X^*X)^{-1/2} d^{1/2} \geq c \]
 for a constant $c > 0$, independent of $N$ and $k$. The last two bounds imply that $c \leq  (X^* X)^{-1/4} d (X^* X)^{-1/4}  \leq 1$
 and therefore that with 
 \[
J := 1 -  (X^* X)^{-1/4} d (X^* X)^{-1/4}  
 \]
we have
\[0 \leq J \leq 1-c < 1\;.\] 
We write 
\[ \frac{1}{\sqrt{(X^* X)^{-1/4} d (X^* X)^{-1/4}}} -1  = \frac{1}{\sqrt{1 - J}} - 1 = \frac{1}{\pi} \int_0^\infty \frac{\di s}{\sqrt{s}} \frac{1}{s + 1 - J} J \frac{1}{s+1} \;. \]
With $1-J \geq c > 0$, we conclude that 
\begin{equation}\label{eq:redu-J}
\Big\|  \frac{1}{\sqrt{(X^* X)^{-1/4} d (X^* X)^{-1/4}}} -1  \Big\|_\HS \leq C \| J \|_\HS\;.
\end{equation} 
To estimate the Hilbert-Schmidt norm of $J$, we expand, similarly as we did in \cref{eq:Xd-1}, 
\[ \begin{split} 
J &= (X^* X)^{-1/4} ((X^*X)^{1/2} - d) (X^* X)^{-1/4} \\ &= \frac{1}{\pi} \int_0^\infty \di s \, \sqrt{s} \, (X^* X)^{-1/4} \frac{1}{s+X^*X} \, d^{1/2} \, b \, d^{1/2} \, \frac{1}{s+d^2}\,  (X^*X)^{-1/4} \;.
\end{split} \]
Writing again $b = g |v \rangle \langle v |$ and using the bounds $\| d^{-1/2} v \| \leq C$, $\| (X^*X)^{-1/4} d^{1/2} \|_\op \leq C$, and $\| d (X^*X)^{-1/2} \|_\op \leq C$ (the latter two bounds are simple consequences of $X^* X  \geq d^2$),
\[ \begin{split} 
\| (X^*X)^{1/4} (s+ X^*X)^{-1} \|_\op &\leq \min \big\{ s^{-3/4} , s^{-1} \big\} \end{split} \]
and also (\ref{eq:riem}) with $j=2$, $k=0$ to bound 
\[ \Big\| \frac{1}{s+d^2}  v \Big\| \leq \min \{ s^{-1}, s^{-1/2} \} \, ,  \]
we arrive at $\| J \|_\HS \leq C \hat{V} (k)$. 
Inserting in \cref{eq:redu-J} and combining the resulting bound with \cref{eq:term1a-fin}, we conclude that
\begin{equation}\label{eq:fin1}  \Big\| d^{1/2} (X^* X)^{-1/4} \frac{1}{\sqrt{(X^* X)^{-1/4} d (X^* X)^{-1/4}}} -1 \Big\|_\HS \leq C \hat{V} (k) \;. \end{equation} 

We turn to the second term on the r.\,h.\,s.\ of \cref{eq:LHS-deco}. Similarly as for the first term
\begin{align*}
& (d+2b)^{1/2} (X X^*)^{-1/4} \frac{1}{\sqrt{(X X^*)^{-1/4} (d+2b) (X X^*)^{-1/4}}} -1 \\
& = (d+2b)^{1/2} \left( (X X^*)^{-1/4} - (d+2b)^{-1/2} \right) \frac{1}{\sqrt{(X X^*)^{-1/4} (d+2b) (X X^*)^{-1/4}}} \\ & \quad + \frac{1}{\sqrt{(X X^*)^{-1/4} (d+2b) (X X^*)^{-1/4}}}  -1\;.   \tagg{eq:deco-3}
\end{align*}
The term on the first line can be bounded analogously as we did with the first term on the r.\,h.\,s.\ of \cref{eq:deco-2}. With $XX^* - (d+2b)^2 = - 2 (d+2b)^{1/2} \, b \, (d+2b)^{1/2}$ we find
\begin{equation}\label{eq:term2a}
\begin{split}
& (d+2b)^{1/2} \left( (X X^*)^{-1/4} - (d+2b)^{-1/2} \right) \frac{1}{\sqrt{(X X^*)^{-1/4} (d+2b) (X X^*)^{-1/4}}} \\ & = C \int_0^\infty \frac{\di s}{s^{1/4}}  \, (d+2b)^{1/2} \, \frac{1}{s+ (d+2b)^2} (d+2b)^{1/2}  \, b \, (d+2b)^{1/2} \, \frac{1}{s+XX^*}  \\ &\hspace{6cm} \times \frac{1}{\sqrt{(X X^*)^{-1/4} (d+2b) (X X^*)^{-1/4}}} \;.
\end{split}
\end{equation}
From $\| d^{1/2} (d+2b)^{-1/2} \|_\op \leq C$ and $\| d^{-1/2} v \| \leq C$, we obtain $\| (d+2b)^{-1/2} v \| \leq C$. Moreover, we find 
\[
\begin{split}
& \Big\| (d+2b) \frac{1}{s+XX^*} \frac{1}{\sqrt{(X X^*)^{-1/4} (d+2b) (X X^*)^{-1/4}}} \Big\|_\op^2 \\ &= \Big\| (d+2b) \frac{1}{s+XX^*} \frac{1}{(X X^*)^{-1/4} (d+2b) (X X^*)^{-1/4}}  \frac{1}{s+XX^*}  (d+2b) \Big\|_\op \\  &\leq \| 
(d+2b) (XX^*)^{-1/2} \|_\op  \Big\| \frac{(XX^*)^{1/4}}{s + XX^*} \Big\|_\op \| (XX^*)^{1/2} (d+2b)^{-1} \|_\op   \\ &\quad \times  \Big\|   \frac{(XX)^{3/4}}{s + XX^*}   \Big\|_\op \| (XX^*)^{-1/2} (d+2b) \|_\op \leq C \min \big\{ s^{-2}, s^{-1} \big\} \;.
\end{split}
\]
Here we used, analogously to \cref{eq:XXd}, the bounds $\| (XX^*)^{1/2} (d+2b)^{-1} \|_\op \leq 1$ and 
\begin{equation} \label{eq:dXX2}
\begin{split}
\| (d+2b) (XX^*)^{-1/2} \|^2_\op  &= \| (d+2b) (XX^*)^{-1} (d+2b) \|_\op  \\
&= \| (d+2b)^{1/2} d^{-1}  (d+2b)^{1/2} \|_\op \\ &= \| d^{-1/2} (d+2b) d^{-1/2} \|_\op = \| 1 + 2 d^{-1/2} b d^{-1/2} \|_\op \leq C \; .
\end{split}
\end{equation} 

On the other hand, we can bound
\[
\Big\| \frac{1}{s+(d+2b)^2} (d+2b) v \Big\|^2 \leq \left\langle v, \frac{1}{s+(d+2b)^2 } v \right\rangle \;.
\]
 With 
\[
\frac{1}{s+(d+2b)^2 } = \frac{1}{s+d^2} - \frac{1}{s+(d+2b)^2} \left[ (d+ 2b) 2b + 2b d \right] \frac{1}{s+d^2}
\]
and using again $b = g |v \rangle \langle v |$, we get
\[
\begin{split}
\Big\langle v, \frac{1}{s+(d+2b)^2 } v \Big\rangle =\; & \big\langle v, \frac{1}{s+d^2} v \big\rangle - 2g \big\langle v ,  \frac{(d+2b)}{s+(d+2b)^2} v \big\rangle \big\langle v,   \frac{1}{s+d^2} v \big\rangle \\ &- 2g \big\langle v ,  \frac{1}{s+(d+2b)^2} v \big\rangle \big\langle v,   \frac{d}{s+d^2} v \big\rangle
\end{split}
\]
and therefore (proceeding as in the proof of \cref{eq:riem}) arrive at
\begin{equation}\label{eq:expe-riem}  \Big\langle v, \frac{1}{s+(d+2b)^2 } v \Big\rangle \leq  \big\langle v, \frac{1}{s+d^2} v \big\rangle \leq \frac{C}{M} \sum_\alpha \frac{u_\alpha^2}{s+u_\alpha^4} \leq C \min \{ s^{-1} , |\log s| \} \, . \end{equation} 
This implies that 
\begin{equation} \label{eq:riem-d+2b} \Big\| \frac{1}{s+(d+2b)^2} (d+2b) v \Big\| \leq  C \min \{ s^{-1/2} , |\log s|^{1/2} \} \;.\end{equation}
From \cref{eq:term2a}, we conclude that 
\begin{equation}
\label{eq:term2a-fin}
\begin{split}
&\Big\| (d+2b)^{1/2} \left( (X X^*)^{-1/4} - (d+2b)^{-1/2} \right) \frac{1}{\sqrt{(X X^*)^{-1/4} (d+2b) (X X^*)^{-1/4}}}  \Big\|_\HS \\
& \hspace{30em}\leq C \hat{V} (k) \, .
\end{split}
\end{equation} 

Finally, let us consider the term on the second line of the r.\,h.\,s.\ of \cref{eq:deco-3}. 
Since $XX^* \leq (d+2b)^2$ (recall that $XX^* = (d+2b)^{1/2} d (d+2b)^{1/2}$), we have 
\begin{equation} \label{eq:lower-J2} (d+2b)^{1/2} (XX^*)^{-1/2} (d+2b)^{1/2} \geq 1 \end{equation}
which also implies that $(XX^*)^{-1/4} (d+2b) (XX^*)^{-1/4} \geq 1$. We define therefore  
\[ W := (XX^*)^{-1/4} (d+2b) (XX^*)^{-1/4} - 1 \geq 0 \;.\]
Then we have 
\[ \frac{1}{\sqrt{(XX^*)^{-1/4} (d+2b) (XX^*)^{-1/4}}}- 1 = \frac{1}{\sqrt{1+W}} - 1 = - \frac{1}{\pi}  \int_0^\infty \frac{\di s}{\sqrt{s}} \frac{1}{s+1+W} W \frac{1}{s+1} \]
and thus  
\begin{equation}\label{eq:redu-W} \Big\|  \frac{1}{\sqrt{(XX^*)^{-1/4} (d+2b) (XX^*)^{-1/4}}} - 1 \Big\|_\HS  \leq C \| W \|_\HS \;.\end{equation} 
To estimate the Hilbert-Schmidt norm of $W$ we write 
\begin{align*} 
W &= (XX^*)^{-1/4} [(d+2b) - (XX^*)^{1/2}] (XX^*)^{-1/4} \\ 
&= \frac{1}{\pi}  \int_0^\infty \di s \sqrt{s} \, (XX^*)^{-1/4} \frac{1}{s+(d+2b)^2} \, (d+2b)^{1/2} \, b \, (d+2b)^{1/2} \, \frac{1}{s+XX^*} (XX^*)^{-1/4} \\
&=  \frac{1}{\pi}   \int_0^\infty \di s \, \sqrt{s} \, (XX^*)^{-1/4} (d+2b)^{1/2}  \, \frac{(d+2b)^{1/2}}{s+(d+2b)^2} \, (d+2b)^{-1/2} \, b \, (d+2b)^{-1/2} \\ &\hspace{7.5cm} \times (d+2b) \, (XX^*)^{-1/2} \, \frac{(XX^*)^{1/4}}{s+XX^*} \;.
\end{align*}
With the resolvent identity, we obtain 
\[
\left( 1+  \big\langle v, \frac{d}{s+d^2} v \big\rangle \right)  \frac{1}{s + (d+2b)^2} v = \frac{1}{s+d^2} v   - \big\langle v, \frac{1}{s+d^2} v  \big\rangle  \frac{(d+2b)}{s+(d+2b)^2} v
\]
and thus 
\[
\Big\| \frac{1}{s + (d+2b)^2} v \Big\| \leq \Big\|  \frac{1}{s+d^2} v \Big\| + \big\langle v, \frac{1}{s+d^2} v  \big\rangle  \Big\| \frac{(d+2b)}{s+(d+2b)^2} v \Big\| \;.
\]
Using \cref{eq:riem} with $j=2$, $k=0$, \cref{eq:expe-riem}, and \cref{eq:riem-d+2b}
we arrive at
\[
\Big\| \frac{1}{s + (d+2b)^2} v  \Big\| \leq C \min \{ s^{-1/2} , s^{-1}  \} \;.
\]
Applying also \cref{eq:dXX2}, $\| (d+2b)^{-1/2} v \| \leq C$ and  
\begin{align*}  \Big\| (XX^*)^{1/4} \frac{1}{s+XX^*} \Big\|_\op \leq C \min \{ s^{-1} , s^{-3/4} \}    \end{align*}
we conclude that 
\begin{equation}
\label{eq:WHS}
\| W \|_\HS \leq C \| (XX^*)^{-1/4} (d+2b)^{1/2} \|_\op \hat{V} (k) \;.
\end{equation}
Since
\[
\begin{split}
\| (XX^*)^{-1/4} (d+2b)^{1/2} \|_\op^2 & = \| (XX^*)^{-1/4} (d+2b) (XX^*)^{-1/4} \|_\op \\
& = \| 1 + W \|_\op  \leq 1 + \| W \|_\HS
\end{split}
\]
we arrive at 
\[ \| W \|_\HS \leq C \hat{V} (k) \;. \]
Inserting this bound in \cref{eq:redu-W} and combining it with \cref{eq:term2a-fin}, we can bound \cref{eq:deco-3} by
\[ \Big\| (d+2b)^{1/2} (X X^*)^{-1/4} \frac{1}{\sqrt{(X X^*)^{-1/4} (d+2b) (X X^*)^{-1/4}}} -1 \Big\|_\HS \leq C \hat{V} (k)\,.\]
Together with \cref{eq:fin1} and with \cref{eq:LHS-deco}, we obtain 
\[ \| O-1 \|_\HS \leq  C \hat{V} (k)  \;. \qedhere\]
\end{proof}

Using the bounds on the kernels $K(k)$ and $L(k)$, our next goal is to show that the unitary transformations $T$ and $Z$ defined 
in \cref{eq:3BT} act on the $c$-- and $c^*$--operators as bosonic Bogoliubov transformations, up to errors that are small on states with few excitations.
(This will allow us to show that conjugation of the r.\,h.\,s.\ of \cref{eq:quad-app} by $T$ and $Z$ produces approximately the r.\,h.\,s.\ of \cref{eq:diagonal-bos}.)
To reach this goal, we need to show first that conjugation with $T$ and $Z$ does not change the number operator $\cN$ and the gapped number operators $\cN_\delta$ substantially. 
We generalize the definition \cref{eq:3BT} for $\lambda \in \bR$ to 
\begin{equation}\label{eq:Zlambda} \begin{split}  
 T_\lambda & := \exp \left( \frac{\lambda}{2} \sum_{k \in \north} \sum_{\alpha,\beta \in \cI_k} K (k)_{\alpha,\beta} c_\alpha^* (k) c_\beta^* (k) - \text{h.c.} \right) , \\
Z_\lambda & := \exp \left( \lambda  \sum_{k \in \north} \sum_{\alpha,\beta \in \cI_k} L (k)_{\alpha,\beta} c_\alpha^* (k) c_\beta (k) \right)\;,
\end{split} \end{equation} 
so that $T = T_1$ and $Z = Z_1$. 
\begin{lemma}[Stability of number operators]
\label{lm:growthN} 
Assume $\| \hat{V} \|_1 < \infty$ and $M \gg N^{2\delta} R^2$. Then for every $m \in \bN$ there exists $C > 0$ such that for all $\lambda \in [-1,1]$ we have
\begin{equation}\label{eq:growth-T} \begin{split} T_\lambda^* \cN^m T_\lambda \leq C (\cN+1)^m \quad \text{and} \quad
T_\lambda^* \cN_\delta \cN^m T_\lambda &\leq C (\cN_\delta + 1)( \cN+1)^m \;. \end{split} \end{equation} 
Conjugation with $Z_\lambda$ leaves the total number of particles constant,  
\[ Z_\lambda^* \cN^m Z_\lambda = \cN^m\;. \]
Moreover, for every $m \in \bN$ there exists $C > 0$ such that, for all $\lambda \in [-1,1]$, we have
\begin{equation} \label{eq:ZNdelta}  
\begin{split} 
Z_\lambda^* \cN_\delta \cN^m Z_\lambda &\leq C  \cN_\delta \cN^m \;.
\end{split} \end{equation} 
\end{lemma} 
\begin{proof}
The proof of \cref{eq:growth-T} can be found in \cite[Lemma 7.2]{BNPSS} where it is stated under the additional assumption that $\hat{V}$ has a compact support; however, using \cref{lm:K} it easily extends to $\| \hat{V} \|_1 < \infty$.

The invariance of $\cN$ w.\,r.\,t.\ $Z_\lambda$ follows since the exponent commutes with $\Ncal$ (the $c^*$-operator creates two fermions while the $c$-operator annihilates two fermions). 

We still have to show \cref{eq:ZNdelta}. We consider the case $m=0$; the extension to $m > 0$ is straightforward. We compute 
\begin{equation}\label{eq:der-ZNZ}
\frac{\di}{\di\lambda}  \langle \psi , Z_\lambda^* \cN_\delta Z_\lambda \psi \rangle = \sum_{k \in \north}  \sum_{\alpha, \beta \in \cI_k} L_{\alpha,\beta} (k)  \langle \psi , Z_\lambda^* \left[ c_\alpha^* (k) c_\beta (k) , \cN_\delta \right] Z_\lambda \psi \rangle\;. \end{equation} 
Using the weighted  pairs operators introduced in \cref{lm:c-N} we have
\[ [ c_\alpha^* (k) , \cN_\delta ] = c_\alpha^{g*} (k)\; , \qquad [ c_\beta (k) , \cN_\delta ] = - c_\beta^g (k) \]
for a weight function $g$ with values in $\{ 0,1,2 \}$. Thus
\[
\frac{\di}{\di\lambda}  \langle \psi , Z_\lambda^* \cN_\delta Z_\lambda \psi \rangle  = 
\sum_{k \in \north}  \sum_{\alpha, \beta \in \cI_k} L_{\alpha,\beta}  (k) \left\langle \psi , Z_\lambda^* \left[ c_\alpha^{g*} (k) c_\beta (k) + c_\alpha^* (k) c_\beta^g (k) \right] Z_\lambda \psi \right\rangle
\]
and by Cauchy--Schwarz
\[
\left| \frac{\di}{\di\lambda}  \langle \psi , Z_\lambda^* \cN_\delta Z_\lambda \psi \rangle \right|
\leq \sum_{k \in \north}  \left( \sum_{\beta \in \cI_k}\Big\| \sum_{\alpha \in \cI_k} L_{\alpha, \beta} (k) \, c_\alpha^g (k) Z_\lambda \psi \Big\|^2 \right)^{\frac{1}{2}} \left( \sum_{\beta \in \cI_k}   \| c_\beta (k) Z_\lambda \psi \|^2 \right)^{\frac{1}{2}} \;.
\] 
Observe that
\[
\begin{split}
\sum_{\beta \in \cI_k}  \Big\| \sum_{\alpha \in \cI_k} L_{\alpha, \beta} (k) c_\alpha^g (k) Z_\lambda \psi \Big\|^2 &= \sum_{\beta, \alpha, \alpha' \in \cI_k} L_{\alpha,\beta} (k) \overline{L_{\alpha' , \beta}} (k)  \, \langle c_\alpha^g (k) Z_\lambda \psi , c_{\alpha'}^g (k) Z_\lambda \psi \rangle \\
&= \sum_{\alpha,\alpha' \in \cI_k} |L (k)|^2_{\alpha,\alpha'} \langle c_\alpha^g (k) Z_\lambda \psi , c_{\alpha'}^g (k) Z_\lambda \psi \rangle = \tr \, |L (k)|^2 C_g 
\end{split}
\]
with the $|\cI_k| \times |\cI_k|$ matrix $C_g$ having entries $(C_g)_{\alpha,\alpha'} = \langle c_\alpha^g (k) Z_\lambda \psi , c_{\alpha'}^g (k) Z_\lambda \psi \rangle$. Since $C_g$ is a positive matrix, we can use \cref{eq:L-bd} to estimate
\[
\begin{split}
\sum_{\beta \in \cI_k} \Big\| \sum_{\alpha \in \cI_k} L_{\alpha, \beta} (k) c_\alpha^g (k) Z_\lambda \psi \Big\|^2  &\leq C \hat{V} (k)^2  \, \tr C_g = C \hat{V} (k)^2   \sum_{\alpha \in \cI_k} \| c_\alpha^g (k) Z_\lambda \psi \|^2\;.
\end{split}
\]
Applying \cref{lm:c-N} and using $\| \hat{V} \|_1 < \infty$, we find   
\[ 
\left| \frac{\di}{\di\lambda}  \langle \psi , Z_\lambda^* \cN_\delta Z_\lambda \psi \rangle \right|  \leq C \langle \psi , Z_\lambda^* \cN_\delta Z_\lambda \psi \rangle \;.
\]
By Grönwall's lemma, we conclude that for all $\lambda \in [-1,1]$ we have
\[ \langle \psi , Z_\lambda^* \cN_\delta Z_\lambda \psi \rangle \leq C \langle \psi, \cN_\delta \psi \rangle\;. \qedhere\]
\end{proof} 

We can now show that the unitary operators $T$ and $Z$ approximately act on $c$-- and $c^*$--operators as bosonic Bogoliubov transformations, up to errors that are negligible on states with few excitations. The action of $T$ is described in the next lemma, whose proof can be found in \cite[Lemma~7.1]{BNPSS}.
\begin{lemma}[Approximate bosonic Bogoliubov transformation]\label{lm:T-action} 
For all $\lambda \in [-1,1]$, $k \in \north$, and $\gamma \in \cI_k$, we have 
\begin{equation}\label{eq:T-action}
T_\lambda^* c_\gamma (k) T_\lambda = \sum_{\alpha \in \cI_k} \cosh (\lambda K (k))_{\alpha,\gamma} c_\alpha (k) + \sum_{\alpha \in \cI_k} \sinh (\lambda K(k))_{\alpha, \gamma} c_\alpha^* (k) + \mathfrak{E}_\gamma (\lambda, k)
\end{equation} 
where for the error term $\mathfrak{E}_\gamma (\lambda, k)$ there exists a $C > 0$ such that for all $\psi \in \cF$ we have 
\[
\sum_{\gamma \in \cI_k} \| \mathfrak{E}_{\gamma} (\lambda, k) \psi \| \leq C M N^{-2/3+\delta} \| (\cN_\delta + M)^{1/2} (\cN+ 1) \psi \| \;.
\]
The same bound holds if we replace $\mathfrak{E}_\gamma (\lambda,k)$ with $\mathfrak{E}_\gamma^* (\lambda,k)$. 
 \end{lemma} 
In the next lemma, we control the action of $Z$ in an analogous fashion.
\begin{lemma}[Approximate bosonic one--particle unitary]\label{lm:Z1Z2} 
Assume $\| \hat{V} \|_1 < \infty$. Let $M \gg R^2 N^{2\delta}$. Then for every $\ell \in \north$, $\gamma \in \cI_\ell$, and $\lambda \in [-1,1]$ we have
\begin{equation} \begin{split} Z_{\lambda}^* c_\gamma (\ell) Z_{\lambda}  &= \sum_{\beta \in \cI_\ell} \exp (\lambda L (\ell))_{\gamma,\beta} c_\beta (\ell) + \mathfrak{F}_\gamma (\lambda , \ell )
\end{split} \end{equation} 
where there exists a $C > 0$ such that for all $\psi \in \fock$ we have
\begin{equation}\label{eq:frakFwtF} \begin{split} 
\sum_{\gamma \in \cI_\ell} \| \mathfrak{F}_\gamma (\lambda , \ell ) \psi \| &\leq  C  N^{-2/3+ \delta} M^{3/2}  \| \cN_\delta^{1/2} \cN  \psi \| \;.
\end{split} \end{equation} 
\end{lemma}

\begin{proof}
Recall that $L$ is antisymmetric; hence $Z_{\lambda}^*$ has the same form as $Z_{\lambda}$, but with $L$ replaced by $-L$. For $\lambda \in [-1,1]$ we compute 
\[ \begin{split} \frac{\di}{\di\lambda} 
\, Z_{\lambda}^* c_\gamma (\ell) Z_{\lambda} =  \sum_{\beta \in \cI_\ell} L  (\ell)_{\gamma,\beta}  Z_{\lambda}^* c_\beta (\ell) Z_{\lambda} + \sum_{k \in \north : \gamma \in \cI_k} \sum_{\beta \in \cI_k} L (k)_{\gamma,\beta} \, Z_{\lambda}^*  \cE_\gamma (\ell , k) c_\beta (k) Z_{\lambda} \end{split} \]
with the error operator $\cE_\gamma (\ell , k)$ introduced in \cref{eq:approximateCCR}. In integral form, we obtain 
\[
\begin{split}
Z_{\lambda}^* c_\gamma (\ell) Z_{\lambda} = \; &c_\gamma (\ell) + \sum_{\beta \in \cI_\ell} L (\ell)_{\gamma,\beta} \int_0^\lambda \di\tau \, Z_{\tau}^* c_\beta (\ell) Z_{\tau} \\ &+ \sum_{k \in \north : \gamma \in \cI_k} \sum_{\beta \in \cI_k} L (k)_{\gamma,\beta} \int_0^\lambda \di\tau \, Z_{\tau}^* \cE_\gamma (\ell , k) c_\beta (k) Z_{\tau}\;.
\end{split}
\]
Iterating $n_0$ times, we find (with $L (\ell)^n_{\gamma,\beta} = \left(L (\ell)^n\right)_{\gamma,\beta}$)
\[
\begin{split} 
Z_{\lambda}^* c_\gamma (\ell) Z_{\lambda} = \; &\sum_{n=0}^{n_0} \frac{\lambda^n}{n!}  \sum_{\beta \in \cI_\ell} L (\ell)^n_{\gamma,\beta} \, c_\beta (\ell) +  \sum_{\beta \in \cI_\ell} L (\ell)^{n_0+1}_{\gamma,\beta}  \int_0^\lambda \di\tau \frac{(\lambda - \tau)^{n_0}}{n_0!} \, Z_{\tau}^* c_\beta (\ell) Z_{\tau} 
 \\ &+ \sum_{n=0}^{n_0} \sum_{k \in \north} \sum_{\beta \in \cI_k \cap \cI_\ell} \sum_{\alpha \in \cI_k} L (\ell)^n_{\gamma,\beta} L (k)_{\beta, \alpha} \int_0^\lambda \di\tau \frac{(\lambda-\tau)^n}{n!} \, Z_{\tau}^* \cE_\beta (\ell ,k) c_\alpha (k) Z_{\tau} 
\end{split}
\]
where, in the last line, for $n=0$, we have $L(\ell)^0_{\gamma ,\beta} = \delta_{\gamma,\beta}$. Thus, completing the first sum to reconstruct the exponential, we have 
\[ Z_{\lambda}^* c_\gamma (\ell) Z_{\lambda} = \sum_{\beta \in \cI_\ell}  \exp (\lambda L (\ell))_{\gamma,\beta} \, c_\beta (\ell) + \mathfrak{F}_{\gamma} (\lambda, \ell) \]
with error term
\[
\begin{split}
\mathfrak{F}_\gamma (\lambda,\ell) = \; & - \sum_{n=n_0+1}^\infty \frac{\lambda^n}{n!}  \sum_{\beta \in \cI_\ell}  L(\ell)^n_{\gamma,\beta}  \, c_\beta (\ell) + \sum_{\beta \in \cI_\ell} L (\ell)^{n_0+1}_{\gamma,\beta}  \int_0^\lambda \di\tau \frac{(\lambda - \tau)^{n_0}}{n_0!}  \, Z_{\tau}^* c_\beta (\ell) Z_{\tau}  \\ &+ \sum_{n=0}^{n_0} \sum_{k \in \north} \sum_{\beta \in \cI_\ell \cap \cI_k} \sum_{\alpha \in \cI_k} L (\ell)^n_{\gamma,\beta} L (k)_{\beta, \alpha}  \int_0^\lambda \di\tau \frac{(\lambda-\tau)^n}{n!} \, Z_{\tau}^* \cE_\beta (\ell ,k) c_\alpha (k) Z_{\tau}
\end{split}
\] 
for an arbitrary $n_0 \in \bN$. This error term can be estimated by
\begin{align*}
\sum_{\gamma \in \cI_\ell} &\| \mathfrak{F}_\gamma (\lambda, \ell) \psi \| \\ \leq \; &\sum_{n > n_0} \frac{\lambda^n}{n!}  \sum_{\gamma,\beta \in \cI_\ell} |L(\ell)^n_{\gamma,\beta}| \| c_\beta (\ell) \psi \| + \sum_{\gamma,\beta \in \cI_\ell} | L(\ell)^{n_0+1}_{\gamma,\beta}|  \int_0^\lambda \di\tau \frac{(\lambda-\tau)^{n_0}}{n_0!}  \, \| c_\beta (\ell) Z_{\tau} \psi \| \\ &+ \sum_{n=0}^{n_0}  \sum_{k \in \north} \sum_{\gamma \in \cI_\ell, \beta \in \cI_k \cap \cI_\ell , \alpha \in \cI_k}  |L (\ell)^n_{\gamma,\beta} | |L(k)_{\beta, \alpha}|  \int_0^\lambda \di\tau \frac{(\lambda-\tau)^n}{n!} \, \| \cE_\beta (\ell ,k) c_\alpha (k) Z_{\tau} \psi \| \\
=: \; & \text{I} + \text{II} + \text{III} \;.  \tagg{eq:I-II-III}
\end{align*}

We estimate 
\[
\text{I} \leq M^{1/2}  \sum_{n > n_0} \frac{\lambda^n}{n!} \| L (\ell)^n \|_\HS  \bigg( \sum_{\beta \in \cI_\ell} \| c_\beta (\ell) \psi \|^2 \bigg)^{1/2}\;.
\]
With \cref{lm:LwtL}, we obtain $\| L(\ell)^n \|_\HS \leq C^n$, uniformly in $N$ and $\ell$. From \cref{lm:c-N} then
\begin{equation}\label{eq:I-Z1} 
\text{I} \leq M^{1/2}  \| \cN_\delta^{1/2} \psi \| \sum_{n > n_0} \frac{C^n}{n!} \;.
\end{equation} 

Similarly, using the invariance of $\cN$ w.\,r.\,t.\ conjugation with $Z_\tau$, we find
\begin{equation}
\label{eq:II-Z1}
\text{II} \leq \frac{C^{n_0}}{n_0!} M^{1/2}  \int_0^\lambda \di\tau \| \cN_\delta^{1/2} Z_\tau \psi \| \leq  \frac{C^{n_0}}{n_0!} M^{1/2} \| \cN^{1/2} \psi \|\;.
 \end{equation}
 
Let us finally consider the last term on the r.\,h.\,s.\ of \eqref{eq:I-II-III}. We have
\[
\text{III} \leq \sum_{n=0}^\infty \frac{\lambda^n}{n!} \sum_{k \in \north} \Bigg(\!\! \sum_{\substack{\gamma \in \cI_\ell,\\\alpha \in \cI_k,\\\beta \in \cI_k \cap \cI_\ell}} |L (\ell)^n_{\gamma,\beta}|^2 |L(k)_{\beta, \alpha}|^2 \Bigg)^{1/2} \!\int_0^\lambda \di\tau \Bigg(\!\! \sum_{\substack{\gamma \in \cI_\ell,\\\alpha \in \cI_k,\\\beta \in \cI_k \cap \cI_\ell}} \| \cE_\beta (k,\ell) c_\alpha (k) Z_{\tau} \psi \|^2 \Bigg)^{1/2}\!\!.
\]
Using 
\[\begin{split}
\Bigg(  \sum_{\substack{\gamma \in \cI_\ell,\\\alpha \in \cI_k,\\\beta \in \cI_k \cap \cI_\ell}} |L (\ell)^n_{\gamma,\beta}|^2 |L(k)_{\beta, \alpha}|^2 \Bigg)^{1/2} &\leq \| L (\ell)^n \|_\HS \| L(k) \|_\HS  \leq C^n \hat{V} (k) \;,
\end{split}
\] 
the bound \cref{eq:ccr-0}, the relation $\cN c_\alpha (k) = c_\alpha (k) (\cN-2)$, and \cref{lm:c-N}, we find
\begin{equation*}
\text{III} \leq C \sum_{k \in \north} \hat{V} (k)  N^{-2/3+ \delta} M^{3/2} \int_0^\lambda \di\tau \, \| \, \cN^{1/2}_\delta \cN Z_\tau \psi \|  \; .  \end{equation*}
With $\| \hat{V} \|_1 < \infty$ and \cref{lm:growthN}, we conclude that
\begin{equation}
\label{eq:III-Z1}
\text{III} \leq C N^{-2/3+ \delta} M^{3/2}  \| \cN_\delta^{1/2} \cN  \psi \|  \; .
\end{equation}

Since the r.\,h.\,s.\ of both \cref{eq:I-Z1} and \cref{eq:II-Z1} vanishes as $n_0 \to \infty$ (and since \cref{eq:III-Z1} does not depend on $n_0$), we arrive at \cref{eq:frakFwtF}.  
\end{proof} 

\section{Linearization of the Kinetic Energy} 

We will use \cref{lm:Z1Z2} to show that \cref{eq:diagonal-bos} and \cref{eq:diagonal-bos2} hold approximately true on states with few excitations. What is still missing to conclude the argument explained in \cref{sec:corrH} is the invariance of $\bH_0 - \bD_\textnormal{B}$ w.\,r.\,t.\ the action of the approximate Bogoliubov transformations \cref{eq:3BT}. The proof is based on the fact that the commutators of $\bH_0$ and $\bD_\textnormal{B}$ with the $c^*$--operators are approximately the same, as described by the following lemma.
\begin{lemma}[Kinetic commutators]\label{lm:HD-comm} 
Let $RM^{1/2} \leq N^{1/3}$. For all $k \in \north$ and all $\alpha \in \cI_k$, we have 
\begin{equation}\label{eq:H0D-comm}
\begin{split}
[\bH_0 , c_\alpha^* (k) ] &= 2\hbar \kappa \lvert k \cdot \hat{\omega}_\alpha\rvert c_\alpha^* (k) + \hbar \mathfrak{E}_\alpha^\textnormal{lin} (k)^* \\
[\bD_\textnormal{B} , c_\alpha^* (k) ] &= 2\hbar \kappa \lvert k \cdot \hat{\omega}_\alpha \rvert c_\alpha^* (k) + \hbar \mathfrak{E}^\textnormal{B}_\alpha (k)^*
\end{split}
\end{equation} 
where there exists a $C > 0$ such that for all $f \in \ell^2({\Ik})$ and all $\psi \in \fock$ we have
\begin{equation}\label{eq:lin-err}
\begin{split} 
\sum_{\alpha \in \cI_k} \Big\|  \mathfrak{E}_\alpha^\textnormal{lin} (k) \psi \Big\| &\leq C \lvert k\rvert \| \cN_\delta^{1/2} \psi \| \;,\\ 
\Big\| \sum_{\alpha \in \cI_k} f_\alpha \mathfrak{E}_\alpha^\textnormal{lin} (k) \psi \Big\| &\leq C \lvert k\rvert M^{-1/2} \| f \|_2 \| \cN_\delta^{1/2} \psi \| \;,\\ 
\sum_{\alpha \in \cI_k} \| \mathfrak{E}_\alpha^{\textnormal{B}} (k) \psi \| &\leq C R^3 M^{3/2} N^{-2/3+\delta} \| \cN^{1/2}_\delta  \cN \psi \| \;.
\end{split}
\end{equation} 
\end{lemma}
\begin{proof}
The bounds for $\mathfrak{E}_\alpha^\textnormal{lin}$ are shown as in \cite[Lemma 8.2]{BNPSS}, keeping track of the $k$--dependence. From \cref{eq:CAR} we get 
 \begin{align*}
 [\Hbb_0,c^*_{\alpha}(k)] & =  \frac{1}{n_{\alpha}(k)}   \sum_{\substack{p\colon p\in \BFc \cap B_\alpha\\p-k \in \BF \cap B_\alpha}} (e(p)+e(p-k))a^*_p a^*_{p-k} = 2\hbar \kappa \lvert k\cdot \hat{\omega}_\alpha\rvert c_\alpha^*(k) + \hbar \Efrak^{\textnormal{lin}}_\alpha(k)^* \;,
 \end{align*}
 where, using the weighted pair operators as in \cref{lm:c-N}, $\Efrak^{\textnormal{lin}}_\alpha(k) = c^{g}_{\alpha}(k)$  with
\begin{align*}
 g(p,k) = \hbar^{-1} \Big( e(p)+e(p-k) - 2\hbar  \kappa \lvert k\cdot \hat{\omega}_\alpha\rvert \Big)  = \hbar \Big( 2 k\cdot (p - \kF \hat{\omega}_\alpha) -\lvert k\rvert^2 \Big)\;.
\end{align*}
Since $B_\alpha$ has diameter of order $N^{1/3} M^{-1/2}$ on the Fermi surface and since $p$ can be at most at distance $|k|$ from the Fermi surface, we can bound (using the assumption $|k| M^{1/2} \leq R M^{1/2} \leq N^{1/3}$)
\[
 \lvert g(p,k) \rvert \leq C \hbar \lvert k\rvert \left( \lvert p - \kF \hat{\omega}_\alpha \rvert + \lvert k\rvert \right) \leq C |k| M^{-1/2} \;.
\]
The first two estimates in \cref{eq:lin-err} follow from \cref{eq:cg-Nsum,eq:cg-N}.

The last bound in \cref{eq:lin-err}  is shown exactly as in \cite[Eq.~(8.6)]{BNPSS}, using the bound $|\north| \leq C R^3$ to sum over $l \in \north$ there.
\end{proof}

The invariance w.\,r.\,t.\ $T$ is established in the next lemma. This lemma can be shown as  \cite[Lemma 8.1]{BNPSS}, replacing bounds for $\mathfrak{E}_\alpha^\textnormal{lin}$ and $\mathfrak{E}_\alpha^{\textnormal{B}}$ with those established in \cref{lm:HD-comm} (and using the assumption $\sum_{k} \hat{V} (k) |k| < \infty$). We skip further details. 
\begin{lemma}[Approximate $T$--invariance of $\Hbb_0 - \Dbb_\B$] \label{lm:T-inv}
Let $ \sum_{k \in \Zbb^3} \lvert \hat{V}(k)\rvert \left(1 + \lvert k\rvert \right) < \infty$. Then there exists a $C > 0$ such that for all $\psi \in \fock$ we have
\[
\begin{split}
\lvert\langle T \psi, &(\bH_0 - \bD_\textnormal{B}) T \psi \rangle - \langle \psi, (\bH_0 - \bD_\textnormal{B}) \psi \rangle \rvert \\ &\leq  C \hbar \left( M^{-1/2} \| (\cN_\delta +1)^{1/2} \psi \|^2 + R^3 M N^{-2/3+\delta} \| \cN_\delta^{1/2} (\cN+1) \psi \| \| (\cN_\delta + 1)^{1/2} \psi \| \right) \,.
\end{split}
\]
\end{lemma} 

In the next lemma, we use \cref{eq:lin-err} to show the approximate invariance of $\bH_0 - \bD_\textnormal{B}$ w.\,r.\,t.\ the action of the transformation $Z$ defined in \cref{eq:3BT}. 
\begin{lemma}[Approximate $Z$--invariance of $\Hbb_0 - \Dbb_\B$] \label{lm:Z-inv} 
Let $\sum_{k \in \Zbb^3} \lvert \hat{V}(k)\rvert \left(1 + \lvert k\rvert \right) < \infty$. Then there exists a $C >0$ such that for all $\psi \in \fock$ we have
\[
\begin{split} 
\lvert \langle Z \psi , (\bH_0 - &\bD_\textnormal{B}) Z \psi \rangle - \langle \psi, (\bH_0 - \bD_\textnormal{B}) \psi \rangle \rvert \\ &\leq  C \hbar \left( M^{-1/2} \| \cN_\delta^{1/2} \psi \|^2
+ R^3 M^{3/2}  N^{-2/3+ \delta} \| \cN_\delta^{1/2} \cN^{1/2} \psi \| \| \cN_\delta^{1/2} \psi \| 
\right) \,.
\end{split}
\]
\end{lemma}  

\begin{proof} 
Recalling the definition \cref{eq:Zlambda} of the operators $Z_\lambda$, we compute 
\[
\frac{\di}{\di\lambda} \langle Z_\lambda \psi , (\bH_0 - \bD_\textnormal{B}) Z_\lambda \psi \rangle = \sum_{k \in \north} \sum_{\alpha, \beta \in \cI_k} L_{\alpha,\beta} (k) \langle Z_\lambda \psi , \left[ c_\alpha^* (k) c_\beta (k) ,  (\bH_0 - \bD_\textnormal{B}) \right]  Z_\lambda \psi \rangle \;.
\]
With \cref{eq:H0D-comm} we obtain 
\[
\begin{split} 
\hbar^{-1} \frac{\di}{\di\lambda} \langle Z_\lambda \psi , (\bH_0 - \bD_\textnormal{B}) Z_\lambda \psi \rangle  = &- \sum_{k \in \north} \sum_{\alpha, \beta \in \cI_k} L_{\alpha,\beta} (k)  \langle Z_\lambda \psi , (\mathfrak{E}_\alpha^{\textnormal{lin}} (k) - \mathfrak{E}_\alpha^\textnormal{B} (k))^* c_\beta (k) Z_\lambda \psi \rangle \\
&- \sum_{k \in \north} \sum_{\alpha, \beta \in \cI_k} L_{\alpha,\beta} (k)  \langle Z_\lambda \psi , c^*_\alpha (k)  (\mathfrak{E}_\beta^\textnormal{lin} (k) - \mathfrak{E}_\beta^\textnormal{B} (k)) Z_\lambda \psi \rangle \;.
\end{split}
\]
Hence
\[
\begin{split}
\left| \hbar^{-1} \frac{\di}{\di\lambda} \langle Z_\lambda \psi , (\bH_0 - \bD_\textnormal{B}) Z_\lambda \psi \rangle \right| \leq \; &\sum_{k \in \north} \sum_{\beta \in \Ik}  \Big\| \sum_{\alpha \in \Ik} L_{\alpha,\beta}  (k) \mathfrak{E}_\alpha^\textnormal{lin} (k) Z_\lambda \psi \Big\|  \| c_\beta (k) Z_\lambda \psi \| \\
&+\sum_{k \in \north} \sum_{\alpha \in \Ik}   \| \mathfrak{E}_\alpha^\textnormal{B} (k) Z_\lambda \psi \|  \Big\| \sum_{\beta \in \Ik} L_{\alpha,\beta} (k) c_\beta (k) Z_\lambda \psi \Big\| \;.
\end{split}
\]
Using \cref{lm:HD-comm} (and $\| L_{\alpha,\cdot}(k) \|_2 \leq \| L(k) \|_\HS$ for all $\alpha \in \cI_k$) we conclude that 
\[
\begin{split}
\Big|\hbar^{-1} \frac{\di}{\di\lambda} \langle Z_\lambda \psi , &(\bH_0 - \bD_\textnormal{B}) Z_\lambda \psi \rangle \Big| \\ \leq \; &\sum_{k \in \north} C M^{-1/2} \lvert k\rvert \sum_{\beta \in \Ik} \| L_{\cdot , \beta} (k) \|_2 \| c_\beta (k) Z_\lambda \psi \| \| \cN_\delta^{1/2} Z_\lambda \psi \| \\
&+  \sum_{k \in \north} \sum_{\alpha \in \Ik}  \| L_{\alpha, \cdot} (k) \|_2  \| \mathfrak{E}_\alpha^\text{B} (k) Z_\lambda \psi \|  \| \cN_\delta^{1/2}  Z_\lambda \psi \|  \\
\leq \; &CM^{-1/2} \sum_{k \in \north} \lvert k\rvert \norm{L(k)}_\HS \| \cN_\delta^{1/2}  Z_\lambda \psi \|^2 \\ &+ C R^3 M^{3/2} N^{-2/3+\delta}\sum_{k \in \north}  \| L (k) \|_\HS 
 \| \cN^{1/2}_\delta  \cN Z_\lambda \psi \| \| \cN_\delta^{1/2} Z_\lambda \psi \| \;.
\end{split}
\]
With \cref{lm:LwtL,lm:growthN} we obtain (since $ \sum_{k \in \Zbb^3} \lvert \hat{V}(k)\rvert \left(1 + \lvert k\rvert \right) < \infty$)
\[
\begin{split}
\left| \hbar^{-1} \frac{\di}{\di\lambda} \langle Z_\lambda \psi , (\bH_0 - \bD_\textnormal{B}) Z_\lambda \psi \rangle \right| & \leq CM^{-1/2}   \| \cN_\delta^{1/2} \psi \|^2 \\
& \quad + C R^ 3M^{3/2} N^{-2/3+\delta} 
 \| \cN_\delta^{1/2}  \cN \psi \| \| \cN_\delta^{1/2}  \psi \|  \;.
\end{split} \]
Integrating over $\lambda \in [0,1]$ we arrive at the desired bound. 
\end{proof}

\section{Proof of \Cref{thm:main}} 
\label{sec:conclusion}
We use the following proposition for localization in particle number sectors of Fock space. It is taken from \cite[Prop.~6.1]{LNSS} (given there for bosonic Fock space, but inspection of the proof shows that the symmetry/antisymmetry of the wave function does not play any role).
\begin{prop}[Particle number localization] \label{prop:loc} 
Let $\cA$ be a non--negative operator on $\cF$ with $P_i D(\cA) \subset D(\cA)$ and $P_i \cA P_j = 0$ if $|i-j| > \ell$, where $P_i = \chi (\cN = i)$. Let $f,g : [0 , \infty )  \to [0,1]$ be smooth functions with $f^2 + g^2 = 1$, $f(x) = 1$ for $x \leq 1/2$, and $f(x) = 0$ for $x \geq 1$. For $L \geq 1$, let $f_L := f (\cN / L)$ and $g_L := g (\cN/L)$.

Then, there exists a $C > 0$ (one can take $C := 2 (\| f' \|_\infty^2 + \| g' \|_\infty^2)$) such that 
\[ - \frac{C \ell^3}{L^2} \cA_\text{diag} \leq \cA - f_L \cA f_L - g_L \cA g_L \leq \frac{C \ell^3}{L^2} \cA_\text{diag} \]
where $\cA_\text{diag} = \sum_{i=0}^\infty P_i \cA P_i$.   
\end{prop}
We turn to the proof of our main result.
\begin{proof}[Proof of \Cref{thm:main}] The main work is for the proof of the lower bound; the upper bound follows from the same operator estimates but using a specific trial state, for which the errors are easier to control.
\paragraph{Lower bound.}
Let $\psi_\text{gs}$ be a normalized ground state vector for the Hamilton operator $H_N$ in \cref{eq:HN0}. Since the Hartree--Fock energy arises from a restriction of the many--body variational problem to a smaller set, we have
\[ \langle \psi_\text{gs} , H_N \psi_\text{gs} \rangle \leq E_N^\textnormal{HF} \;.\]
Let $\xi_\text{gs} = R^* \psi_\text{gs}$ denote the excitation vector associated with $\psi_\text{gs}$, defined through the unitary particle--hole transformation \cref{eq:RaR}. From the definition \cref{eq:corr} of the correlation Hamiltonian we have $\langle \xi_\text{gs}, \cH_\textnormal{corr} \xi_\text{gs} \rangle \leq 0$. With \cref{lm:H0-apri} and \cref{cor:QB}, we find a $C > 0$ such that  
\begin{equation}\label{eq:aprixi}
\begin{split}
\langle \xi_\text{gs} , \bH_0 \xi_\text{gs} \rangle \leq C \hbar \;, \quad \langle \xi_\text{gs}, Q_\textnormal{B} \xi_\text{gs} \rangle \leq C \hbar \;,  \quad  \langle \xi_\text{gs} , \cE_1 \xi_\text{gs} \rangle \leq C \hbar \;.
\end{split}
\end{equation} 
The last bound follows because from \cref{lm:X} and \cref{cor:cE} we get $\cE_1 \leq C (\cH_\textnormal{corr} + \bH_0 + \hbar)$. Furthermore, from \cref{cor:apri-NN}, we have
\begin{equation} \label{eq:aprixi-N}
\langle \xi_\text{gs} , \cN \xi_\text{gs} \rangle \leq C N^{1/3}\; , \quad  \langle \xi_\text{gs} , \cN_\eps \xi_\text{gs} \rangle \leq C N^\eps \qquad \textnormal{for every $\eps > 0$.}
\end{equation}

Next we localize w.\,r.\,t.\ the number of particles. We choose smooth functions $f$ and $g$ as in \cref{prop:loc} and set $f_N := f (\cN / C_0 N^{1/3})$, $g_N := g (\cN /  C_0 N^{1/3})$ for a constant $C_0 > 0$ large enough, to be fixed below. We set $\cA = \cH_\textnormal{corr} + C \hbar$, with $C > 0$ large enough. From \cref{lm:H0-apri} we get $\cA \geq 0$. From the definition \cref{eq:corr} of $\cH_\textnormal{corr}$, combined with the bounds in \cref{cor:QB} for the operator $Q_\textnormal{B}$, in \cref{lm:X} for the exchange operator $\bX$ and in \cref{cor:cE} for the error term $\cE_2$, we conclude that 
\[ \cA \leq C (\bH_0 + \cE_1 + \hbar) \;. \]
Since $\bH_0$ and $\cE_1$ both commute with $\cN$, it also follows that $\cA_\text{diag} \leq C (\bH_0 + \cE_1 + \hbar)$. From \cref{prop:loc} (since, with the notation introduced in the proposition, $P_i \cA P_j = 0$ if $|i-j| > 4$), we find 
\[ - C N^{-2/3} (\bH_0 + \cE_1 + \hbar) \leq \cH_\textnormal{corr} - f_N \cH_\textnormal{corr} f_N - g_N \cH_\textnormal{corr} g_N \leq C N^{-2/3} (\bH_0 + \cE_1 + \hbar) \;. \]
We apply this bound to the ground state $\xi_\text{gs}$. From the a--priori bounds in \cref{eq:aprixi}, we obtain  
\begin{equation}
\label{eq:loc-xi} \langle \xi_\text{gs} , \cH_\textnormal{corr} \xi_\text{gs} \rangle \geq  \langle \xi_\text{gs} , f_N \cH_\textnormal{corr} f_N \xi_\text{gs} \rangle + \langle \xi_\text{gs} , g_N \cH_\textnormal{corr} g_N \xi_\text{gs} \rangle - C N^{-1} \;.
\end{equation} 
Since $\xi_\text{gs}$ is the ground state vector of $\cH_\textnormal{corr}$, we can estimate
\[
\langle \xi_\text{gs} , g_N \cH_\textnormal{corr} g_N \xi_\text{gs} \rangle \geq \| g_N \xi_\text{gs} \|^2 \,  \langle \xi_\text{gs}, \cH_\textnormal{corr} \xi_\text{gs} \rangle \;.
\]
With \cref{eq:loc-xi} (and since $f^2 + g^2 = 1$), we arrive at 
\begin{equation}\label{eq:loc-2}
\| f_N \xi_\text{gs} \|^2 \langle \xi_\text{gs} , \cH_\textnormal{corr} \xi_\text{gs} \rangle \geq \langle f_N \xi_\text{gs} , \cH_\textnormal{corr}  f_N \xi_\text{gs} \rangle - C N^{-1}\;.
\end{equation} 
From \cref{eq:aprixi-N}, we have, fixing $C_0$ large enough,  
\[
\| g_N \xi_\text{gs} \|^2 = \langle \xi_\text{gs}, g^2 (\cN / C_0 N^{1/3}) \xi_\text{gs} \rangle \leq \frac{1}{C_0 N^{1/3}} \langle \xi_\text{gs} , \cN \xi_\text{gs} \rangle \leq \frac{1}{2} \;.
\]
Hence $\| f_N \xi_\text{gs} \|^2 \geq 1/2$ and, from \cref{eq:loc-2}, 
\begin{equation}\label{eq:loc-fin}
\langle \xi_\text{gs} , \cH_\textnormal{corr} \xi_\text{gs} \rangle \geq \langle \xi , \cH_\textnormal{corr}   \xi \rangle - C N^{-1}
\end{equation} 
where we defined $\xi = f_N \xi_\text{gs} / \| f_N \xi_\text{gs} \| \in \chi (\cN_\textnormal{p} - \cN_\textnormal{h} = 0) \cF$ (particle number localization leaves the space invariant, since $\cN_\textnormal{p}$ and $\cN_\textnormal{h}$ commute with $\cN$). Like $\xi_\text{gs}$, the localized vector $\xi$ satisfies $\langle \xi , \cH_\textnormal{corr} \xi \rangle \leq C \hbar$ and therefore by \cref{lm:H0-apri} we get 
\begin{equation}\label{eq:apri-loc1}
\langle \xi, \bH_0 \xi \rangle \leq C \hbar \;.
\end{equation}
The advantage of working with $\xi$ is that it satisfies stronger bounds (compared with $\xi_\text{gs}$) on the number of particles. In fact, we find
\begin{equation}\label{eq:apri-loc2}
\langle \xi , \cN^m \xi \rangle \leq C^m N^{m/3} , \quad \langle \xi\;, \cN^m \cN_\eps \xi \rangle  \leq C^m N^{\eps + m/3}
\end{equation} 
for every $m \in \bN$ and $\eps > 0$ (to prove the second estimate, we used $[ \cN, \cN_\eps ] = 0$).

From \cref{eq:loc-fin}, to conclude the proof of the lower bound, it is therefore enough to show that $\langle \xi , \cH_\textnormal{corr} \xi \rangle \geq E_N^\textnormal{RPA}  - C N^{-1/3-\alpha}$, for sufficiently small $\alpha > 0$ and for all $\xi \in \chi (\cN_\textnormal{p} - \cN_\textnormal{h} = 0) \cF$ satisfying \cref{eq:apri-loc1,eq:apri-loc2}. For such vectors, it follows from \cref{lm:X}, \cref{cor:cE} and \cref{lm:QB} that, for any sufficiently small $\eps, \delta > 0$ and for $N^{2\delta} \ll M \ll N^{2/3-2\delta}$,
\begin{equation} \label{eq:Hcorrxi}
\begin{split} 
\langle \xi, \cH_\textnormal{corr} \xi \rangle \geq \; & \langle \xi, (\bH_0 + Q_\textnormal{B}^R) \xi \rangle - C \hbar \Big( N^{-1/3} + N^{-\eps/4} + N^{-(1-\gamma)/3 + 5\eps/4} + N^{-\delta /2} \\ &\hspace{5.5cm} + R^{1/2} M^{1/4} N^{-1/6+ \delta/2} + R^{-1/2}  \Big)
\end{split}
\end{equation}  
with the quadratic expression $Q_\textnormal{B}^R$ defined in \cref{eq:QBR} (notice that the definition of $Q_\textnormal{B}^R$ depends on $\delta$). Using the notation introduced in \cref{eq:H0-comm2} and in \cref{eq:heff-def}, we can write 
\begin{equation}\label{eq:H0+Q}
\langle \xi, (\bH_0 + Q_\textnormal{B}^R) \xi \rangle =  \langle \xi, (\bH_0 - \bD_\textnormal{B}) \xi \rangle + \sum_{k \in \north} 2\hbar \kappa |k| \langle \xi, h_\text{eff} (k) \xi \rangle \;.
\end{equation} 
Next, we diagonalize the quadratic Hamiltonian $h_\text{eff} (k)$ by means of the approximate Bogoliubov transformations defined in \cref{sec:BT}. Recalling \cref{eq:3BT}, we define $\eta = Z^* T^* \xi \in \chi (\cN_\textnormal{p} - \cN_\textnormal{h} = 0) \cF$. From \cref{eq:apri-loc2} and from \cref{lm:growthN}, we can control the number of particles in $\eta$ and $Z \eta = T^* \xi$: for every $m \in \bN$ we find a $C > 0$ such that
\begin{align}
\langle \eta, \cN^m \eta \rangle & \leq C N^{m/3}\;, & \langle Z  \eta, \cN^m Z \eta  \rangle &\leq C N^{m/3}\;, \label{eq:apri-eta} \\
\langle \eta, \cN^m \cN_\delta \eta \rangle & \leq C N^{\delta + m/3}\;, &   \langle Z \eta, \cN^m \cN_\delta Z \eta \rangle & \leq C N^{\delta+ m/3} \;. \label{eq:apri-eta-delta}
\end{align} 
Writing $\xi = T Z \eta$ and applying \cref{lm:T-inv} and \cref{lm:Z-inv}, we obtain 
\begin{align*}
\langle \xi, (\bH_0 - \bD_\textnormal{B}) \xi \rangle &=  \langle T Z \eta , (\bH_0 - \bD_\textnormal{B}) T Z \eta \rangle \\
&\geq \langle \eta, (\bH_0 - \bD_\textnormal{B}) \eta \rangle - C \hbar \Big( M^{-1/2}  \| (\cN_\delta + 1)^{1/2} \eta \|^2 \\
&\hspace{10em} + R^3 M^{3/2} N^{-2/3+ \delta} \| \cN_\delta^{1/2} (\cN+1) \eta \| \| (\cN_\delta + 1)^{1/2} \eta \| \Big)  \\
&\geq  \langle \eta, (\bH_0 - \bD_\textnormal{B}) \eta \rangle - C \hbar \left( M^{-1/2} N^\delta + R^3 M^{3/2} N^{-1/3+2\delta}  \right) \;. \tagg{eq:inva}
\end{align*}
We now focus on the second term on the r.\,h.\,s.\ of \cref{eq:H0+Q}. Writing $\xi = T Z \eta$, we compute first the action of $T$. We proceed here as in the proof of \cite[Lemma~10.1]{BNPSS}. Analogously to \cite[Eqs.~(10.13)]{BNPSS}
we find
\begin{align*}
& \sum_{k \in \north} \hbar \kappa |k| \langle \xi, h_\text{eff} (k) \xi \rangle \\
& = \sum_{k \in \north} 2\hbar \kappa |k| \langle T Z  \eta , h_\text{eff} (k) T Z  \eta \rangle \\
& \geq \sum_{k \in \north} \hbar \kappa \lvert k\rvert \tr \left( E(k) - D(k) - W(k) \right) + \sum_{k \in \north} \sum_{\alpha,\beta \in \cI_k} 2\hbar \kappa |k| \, \mathfrak{K} (k)_{\alpha, \beta} \langle Z \eta , c_\alpha^* (k) c_\beta (k) Z  \eta \rangle \\
& \quad - C\hbar \Big(  N^{-2/3+\delta} \| \cN^{1/2} Z  \eta \|^2   + M R^4 N^{-2/3+\delta} \| (\cN_\delta+1)^{1/2}  Z \eta \| \| (\cN_\delta + M)^{1/2} (\cN+1) Z  \eta \| \\
& \hspace{1.4cm} + M^2 R^4 N^{-4/3 + 2\delta} \| (\cN_\delta+ M)^{1/2} (\cN+1) Z  \eta \|^2 \Big) \tagg{eq:T-action0}
\end{align*}
where we introduced the $|\cI_k| \times |\cI_k|$ matrix $\mathfrak{K}$ by
\[
\begin{split} 
\left( \begin{array}{ll} \mathfrak{K} (k) & 0 \\ 0 & \mathfrak{K} (k) \end{array} \right) & := \left( \begin{array}{ll} \cosh (K(k)) & \sinh (K(k)) \\
\sinh (K(k)) & \cosh (K(k)) \end{array} \right) \\ &\hspace{.5cm} \times \left( \begin{array}{ll} D(k) + W(k) & \wt{W} (k) \\ \wt{W} (k)  & D(k) + W(k) \end{array} \right)
 \left( \begin{array}{ll} \cosh (K(k)) & \sinh (K(k)) \\
 \sinh (K(k)) & \cosh (K(k)) \end{array} \right)\,.
 \end{split}
 \]
Comparing with \cref{eq:deco-quad2}, we find $\mathfrak{K}(k) = O(k) E(k) O(k)^T$. The first error term in the square brackets on the r.\,h.\,s.\ of (\cref{eq:T-action0}) arises from \cite[Eq.~(10.10)]{BNPSS}, a bound which holds under the assumption $\| \hat{V} \|_1 < \infty$; this follows from the observation that \cite[Eq.~(10.9)]{BNPSS} can be improved to 
\[
\begin{split}
\Big| \big[ 2 \sinh (K(k)) (D(k) + W(k)) &\sinh (K(k)) + \cosh (K(k)) \wt{W} (k) \sinh (K(k)) \\ &+ \sinh (K(k)) \wt{W} (k) \cosh (K(k)) \big]_{\alpha, \alpha} \Big| \leq C \hat{V} (k) M^{-1} \; .
\end{split}
\]
The further two error terms in the square brackets arise from \cite[Eq.~(10.6)]{BNPSS}; this estimate holds for every fixed $k$. The sum over $k \in \north$ gives the additional factor $R^4$.
Using \cref{eq:apri-eta} and \cref{lem:rpaupdate} (and recalling $M \gg N^{2\delta}$) we find
\begin{align*}
 \sum_{k \in \north} 2\hbar \kappa |k| \langle \xi, h_\text{eff} (k) \xi \rangle
& \geq E_N^\textnormal{RPA} + \sum_{k \in \north} \sum_{\alpha,\beta \in \cI_k} 2\hbar \kappa |k| \, \mathfrak{K} (k)_{\alpha, \beta} \langle Z \eta , c_\alpha^* (k) c_\beta (k) Z  \eta \rangle \\
& \quad - C\hbar \Big( R^2 M^{1/4} N^{-1/6+\delta /2} + N^{-\delta/2} + M^{-1/4} N^{\delta/2} + N^{-1/3+\delta} \\
&\qquad\qquad+ M^{3/2} R^4 N^{-1/3+3\delta/2} + M^{3} R^4 N^{-2/3 + 2\delta}  \Big)\,. \tagg{eq:T-action2}
\end{align*}
Next, we compute the action of the approximate Bogoliubov transformation (approximate unitary transformation in the one--boson Hilbert space) $Z$ in the second term on the r.\,h.\,s.\ of \cref{eq:T-action2}. With \cref{lm:Z1Z2}, recalling that $\exp (L(k)) = O(k)\wt{O}(k)$, we find 
\begin{equation} \label{eq:conju-Z}
\begin{split}
\sum_{k \in \north} 2\hbar \kappa |k|  &\sum_{\alpha,\beta \in \cI_k}  \mathfrak{K}_{\alpha,\beta} (k) \langle  \eta , Z^* c_\alpha^* (k) c_\beta (k) Z  \eta\rangle \\
=  &\sum_{k \in \north}  2\hbar \kappa |k|  \sum_{\alpha,\beta \in \cI_k}  \left[ \wt{O}^T(k) O^T (k) \mathfrak{K} (k) O (k) \wt{O}(k) \right]_{\alpha,\beta} \langle  \eta , c_\alpha^* (k) c_\beta (k)  \eta \rangle \\
&+  \sum_{k \in \north}  2\hbar \kappa |k|  \sum_{\alpha,\beta \in \cI_k} \left[ \wt{O}^T(k) O^T (k) \mathfrak{K} (k) \right]_{\alpha,\beta} \langle  \eta , c_\alpha^* (k)  \mathfrak{F}_\beta (1,k)  \eta \rangle \\
 &+   \sum_{k \in \north}  2\hbar \kappa |k|  \sum_{\alpha,\beta \in \cI_k} \left[ \mathfrak{K} (k) O (k) \wt{O}(k)\right]_{\alpha,\beta} \langle  \eta , \mathfrak{F}^*_\alpha (1,k) c_\beta (k)  \eta \rangle\\
&+    \sum_{k \in \north}  2\hbar \kappa |k|   \sum_{\alpha,\beta \in \cI_k} \mathfrak{K} (k)_{\alpha,\beta} \langle  \eta , \mathfrak{F}^*_\alpha (1,k) \mathfrak{F}_\beta (k)   \eta \rangle \;.
\end{split}
\end{equation} 
By \cref{lm:Z1Z2} we can show that the contributions on the last three lines are negligible. For example, the second term can be bounded by
\[
\begin{split} 
& \Big|  \sum_{k \in \north}  2\hbar \kappa |k|  \sum_{\alpha,\beta \in \cI_k} \left[ \wt{O}^T(k) O^T (k) \mathfrak{K} (k) \right]_{\alpha,\beta} \langle  \eta , c_\alpha^* (k)  \mathfrak{F}_\beta (1,k)  \eta \rangle \Big| \\ & \leq \sum_{k \in \north}  2\hbar \kappa |k| \sum_{\beta \in \cI_k} \| \mathfrak{F}_\beta (1,k)  \eta \| \Big\| \sum_{\alpha \in \cI_k}  \left[\wt{O}^T(k) O^T (k) \mathfrak{K} (k) \right]_{\alpha,\beta} c_\alpha (k) \eta \Big\|\\
 & \leq \sum_{k \in \north} 2\hbar \kappa |k|  \sum_{\beta \in \cI_k} \| \mathfrak{F}_\beta (1,k) \eta \| \| [\wt{O}^T(k) O^T (k) \mathfrak{K} (k)]_{\beta, .}  \|_2 \| \cN_\delta^{1/2} \eta \| \\ 
 & \leq C N^{-1+ \delta} M^{3/2} \sum_{k \in \north} |k| \,  \| \wt{O}^T(k) O^T (k) \mathfrak{K} (k) \|_\HS  \| \cN_\delta^{1/2} \cN \eta \|  \| \cN_\delta^{1/2} \eta\| \;. 
 \end{split}
 \]
Recalling $\mathfrak{K} (k) = O (k) E (k)  O^T (k)$ and the expression \cref{eq:UEU} for the matrix $E (k)$, we find
\[
\| \wt{O}^T(k) O^T (k) \mathfrak{K} (k) \|_\HS  = \sqrt{2} \left( \tr \; d^2 + 2 \tr \; d^{1/2} b d^{1/2} \right)^{1/2}  \leq C M^{1/2} \;.
\]
Since $|k| < R$ for all $k \in \north$, we conclude, with the bounds  \cref{eq:apri-eta}, that
\[
\Big|  \sum_{k \in \north}   2\hbar \kappa |k| \sum_{\alpha,\beta \in \cI_k} \left[ \wt{O}^T(k) O^T (k) \mathfrak{K} (k) \right]_{\alpha,\beta} \langle \eta , c_\alpha^* (k)  \mathfrak{F}_\beta (1,k) \eta \rangle \Big|  \leq 
C N^{-2/3+ 2\delta} R^4 M^{2} \;.
\]
Proceeding similarly to bound the last two terms on the r.\,h.\,s.\ of \cref{eq:conju-Z}, we obtain 
\[ \begin{split}
& \sum_{k \in \north}  2\hbar \kappa |k| \sum_{\alpha,\beta \in \cI_k} \mathfrak{K}_{\alpha,\beta} (k) \langle  \eta , Z^* c_\alpha^* (k) c_\beta (k) Z  \eta\rangle \\
& \geq \sum_{k \in \north} 2\hbar \kappa |k|  \sum_{\alpha,\beta \in \cI_k}  \left[ \wt{O}^T(k) O^T (k) \mathfrak{K} (k) O (k) \wt{O}(k) \right]_{\alpha,\beta} \langle  \eta , c_\alpha^* (k) c_\beta (k)  \eta \rangle - C  N^{-2/3+ 2\delta} R^4 M^{2}\,. \end{split}
\]
According to \cref{eq:deco-quad3}, we have $\wt{O}^T (k) O^T (k) \mathfrak{K} (k) O (k) \wt{O} (k) = \wt{P} (k)$, with the matrix $\wt{P}$ defined as in \cref{eq:wtO-def}. From $P \geq D$ (and recalling from \cref{eq:quad-app} and \cref{eq:heff-def} the relation between $\bD_\textnormal{B}$ and $D$), we get the key lower bound 
\[
\sum_{k \in \north}  2\hbar \kappa |k| \sum_{\alpha,\beta \in \cI_k} \mathfrak{K}_{\alpha,\beta} (k) \langle  \eta , Z^* c_\alpha^* (k) c_\beta (k) Z  \eta\rangle  \geq \langle \eta , \bD_\textnormal{B} \eta\rangle - C \hbar N^{-1/3+2\delta} R^4 M^2\;.
\]
From \cref{eq:T-action2}, we obtain
\begin{align*}
\sum_{k \in \north} 2\hbar \kappa |k| \langle \xi, h_\text{eff} (k) \xi \rangle & \geq E_N^\textnormal{RPA} + \langle \eta, \bD_\textnormal{B} \eta \rangle \tagg{eq:heffD}\\
& \quad - C \hbar \Big(
R^2 M^{1/4} N^{-1/6+\delta /2} + N^{-\delta/2} + M^{-1/4} N^{\delta/2} \\ & \qquad\qquad + M^{2} R^4 N^{-1/3+2\delta} + M^{3} R^4 N^{-2/3 + 2\delta}  \Big) \;.
\end{align*}
Inserting the last equation and \cref{eq:inva} in \cref{eq:H0+Q}, we find
\[
\begin{split}
\langle \xi, (\bH_0 + Q_\textnormal{B}^R) \xi \rangle & \geq E^\textnormal{RPA}_N + \langle \eta , \bH_0 \eta \rangle \\
& \quad- C\hbar \Big( M^{-1/2} N^{\delta} + R^2 M^{1/4} N^{-1/6+\delta /2} + N^{-\delta/2} + M^{-1/4} N^{\delta/2} \\ &\qquad\qquad  + M^{2} R^4 N^{-1/3+2\delta} + M^{3} R^4 N^{-2/3 + 2\delta} \Big)\;.
\end{split}
\]
Since $\bH_0 \geq 0$, from \cref{eq:Hcorrxi} we obtain
\[
\begin{split}
\langle \xi, \cH_\textnormal{corr} \xi \rangle \geq \; &E^\textnormal{RPA}_N - C\hbar \Big( N^{-\eps /4} + N^{- (1-\gamma)/3 + 5\eps /4} + N^{-\delta/2} + R^{2} M^{1/4} N^{-1/6+\delta /2} + R^{-1 /2}  \\ &\hspace{5em}  +M^{-1/2} N^{\delta} + M^{-1/4} N^{\delta/2}  + M^{2} R^4 N^{-1/3+2\delta} + M^{3} R^4 N^{-2/3 + 2\delta} \Big) \,.
\end{split}
\]
Choosing $R= N^{\delta}$, $M = N^{C \delta}$ for a sufficiently large constant $C > 0$, $\gamma < 1$ and then both $\eps > 0$ and $\delta > 0$ small enough, we conclude that $\langle \xi, \cH_\textnormal{corr} \xi \rangle \geq E_N^\textnormal{RPA} - C N^{-1/3 - \alpha}$ for some $\alpha > 0$ and thus, from \cref{eq:loc-fin}, also that $\langle \xi_\text{gs}, \cH_\textnormal{corr} \xi_\text{gs} \rangle \geq E_N^\textnormal{RPA} - C N^{-1/3 - \alpha}$. This completes the proof of the lower bound for \cref{thm:main}.

\paragraph{Upper bound.}
Instead of working with the state $\xi = T Z \eta$ and establishing its properties through a--priori estimates, we directly use the trial state $\xi_\textnormal{trial} := T \Omega$, where the transformation $Z$ is not needed. We compute explicitly the expectation value
\[
\langle \xi_\textnormal{trial}, \Hcal_\textnormal{corr} \xi_\textnormal{trial} \rangle = \langle \xi_\textnormal{trial}, (\Hbb_0 + Q_\B + \Ecal_1 + \Ecal_2 + \Xbb) \xi_\textnormal{trial} \rangle\;.
\]
Note that by \cref{lm:growthN} we have
\begin{equation}\label{eq:err1}
\langle T \Omega, \Ncal^k T \Omega \rangle \leq C_k\;, \quad \textnormal{for }k \in \Nbb\;.
\end{equation}
Furthermore, for all $\delta > 0$, we have the simple bound for the gapped number operator 
\begin{equation}\label{eq:err2}
\Ncal_\delta \leq \Ncal\;,
\end{equation}
so that all expectations values of powers of $\Ncal$ and $\Ncal_\delta$ in $T\Omega$ are of order one w.\,r.\,t.\ $N$. By \cref{lm:T-inv} we get
\begin{align*}
\langle T\Omega, \Hbb_0 T\Omega \rangle 
& = \langle T\Omega, (\Hbb_0 - \Dbb_\B) T\Omega \rangle + \langle T\Omega, \Dbb_\B T\Omega \rangle \\
& \leq  \langle T\Omega, \Dbb_\B T\Omega \rangle + C \hbar \left( M^{-1/2} + R^3 M N^{-2/3+\delta} \right)\;.
\end{align*}
The expectation value $\langle T\Omega,\Dbb_\B T\Omega\rangle$ can be computed by applying the approximate Bogoliubov transform according to \cref{lm:T-action}. 
Expressions that are normal-ordered in terms of bosonic pairs operators vanish on $\Omega$; only the contribution of the form $c c^*$ is non-vanishing but easily seen to be of order $\hbar$. We conclude that 
\begin{equation}\label{eq:err3}\langle T\Omega, \Hbb_0 T\Omega\rangle \leq C \hbar \;.\end{equation}
The bounds \cref{eq:err1}, \cref{eq:err2}, and \cref{eq:err3} are sufficient to control all error terms in the following computation.
In fact, using \cref{lm:X} and \cref{cor:cE} the contributions of $\Ecal_1$, $\Ecal_2$, and $\Xbb$ are now found to be of order $N^{-1/3-\alpha}$ for some $\alpha > 0$. Furthermore, by \cref{lm:QB}, we can replace $Q_\B$ by the patch-decomposed $Q_\B^R$ at the cost of a only a further small error.  

It remains to compute explicitly the expectation value
\[
\begin{split}
\langle T \Omega, ( \Dbb_\B + Q_\B^R ) T\Omega \rangle = \; & \sum_{k \in \north} 2\hbar \kappa |k| \langle T\Omega, h_\text{eff} (k) T\Omega \rangle \leq  E_N^\textnormal{RPA} + C N^{-1/3-\alpha}
\end{split}
\]
for $\alpha > 0$ small enough. Here, we proceeded as in \eqref{eq:T-action0} (with $Z\eta$ replaced by $\Omega$) to implement the action of the approximate Bogoliubov transformation $T$ and used that all pair annihilation operators vanish on $\Omega$.  This completes the proof of the upper bound for \cref{thm:main}.
\end{proof}

We quickly discuss how to adapt the computation of \cite{BNPSS0} of the explicit RPA formula. The only new aspect here is the additional factor $R^2$ in the first error term.
\begin{lemma}[Explicit RPA formula]\label{lem:rpaupdate}
 Let $\| \hat{V} \|_1 < \infty$. Then
 \begin{align*}
 & \sum_{k \in \north}\!\!\! \hbar \kappa \lvert k\rvert \tr \left( E(k) - D(k) - W(k) \right) \\
 & = E^\textnormal{RPA}_N + \Ocal \left(\hbar \big( R^2 M^{1/4} N^{-1/6+\delta /2} + N^{-\delta/2} + M^{-1/4} N^{\delta/2}\big)\right)\,.
\end{align*}
\end{lemma}
\begin{proof}
 The proof was given in \cite[Eqs.~(5.13)--(5.18)]{BNPSS0} under the assumption that $\hat{V}$ has compact support. We only give the generalization of the main estimates in original notation. With a factor $\lvert k \rvert^2 < R^2$ (for $k\in \north$) originating from \cref{eq:nalphak} we find
 \begin{align*}
  \lvert \log f(\lambda) - \log \tilde{f}(\lambda) \rvert \leq C \left( R^2 \hat{V}(k) \sqrt{M} N^{-1/3 + \delta} + N^{-\delta} + \frac{N^\delta}{\sqrt{M}} \right) \,.
 \end{align*}
 Furthermore
 \begin{align*}
  \lvert \log f(\lambda)\rvert & \leq C \hat{V}(k) \lambda^{-2}\;, & \lvert \log \tilde{f}(\lambda)\rvert & \leq C \hat{V}(k) \lambda^{-2}\;.
 \end{align*}
Following \cite[Eq.~(5.18)]{BNPSS0} and using $\| \hat{V} \|_1 < \infty$ the proof is completed as before.
\end{proof}

\appendix

\section{Generalized Upper Bound}
\label{sec:appendix}
As an upper bound, the estimate \cref{eq:EN-asy} for the correlation energy holds under weaker assumptions on the interaction. 
\begin{theorem}[Generalized RPA upper bound]\label{thm:impr-up} 
Suppose $V : \mathbb{T}^{3} \to \bR$, $\hat{V} \geq 0$, and 
\begin{equation}\label{eq:quasiC}
\sum_{k \in \bZ^3} |k| \hat{V} (k)^2 < \infty \;.
\end{equation} 
For $k_\F > 0$ let $N := |B_\F| = | \{ k \in \bZ^3 : \lvert k\rvert \leq k_\F \}|$. Then, as $k_\F \to \infty$, we have
\begin{equation}\label{eq:up-app}
E_N \leq E_N^\textnormal{HF} + E_N^\textnormal{RPA} + o \, (\hbar)
\end{equation} 
with $E_N^\textnormal{RPA}$ as defined in \cref{eq:RPA}. 
\end{theorem} 

\begin{rem}Expanding the logarithm, it is easy to check that the assumption \cref{eq:quasiC} guarantees that the sum defining $E_N^\textnormal{RPA}$ in \cref{eq:RPA} is finite. 
\end{rem}

\begin{proof}[Proof of \cref{thm:impr-up}] 
 We now give the proof of \cref{thm:impr-up}, explaining how to generalize the argument presented in \cref{sec:conclusion} in the paragraph devoted to the upper bound. For given $0 < R \ll N^{1/3}$, we consider the set $\Gamma^\textnormal{nor}$, defined in 
\cref{eq:north}. Note that in particular $\north$ restricts our attention to momenta $\lvert k\rvert <R$.  Moreover, for $\delta > 0$ sufficiently small, we introduce the sets $\cI_k^{\pm}$ and $\cI_k = \cI_k^+ \cup \cI_k^-$ as in \cref{eq:Ik-def}. For $k \in \Gamma^\textnormal{nor}$, we define the $|\cI_k | \times |\cI_k|$ matrix $K (k)$ as in \cref{sec:BT}. As stated in \cref{lm:K}, we have pointwise in $k \in \Gamma^\textnormal{nor}$, without using the assumption on $\hat{V}$, the bound
 \begin{equation}
 \label{eq:K-point}
 |K_{\alpha, \beta} (k) | \leq C \frac{\hat{V} (k)}{M} \;.
 \end{equation}
 With the matrices $K(k)$ we define the unitary operators $T$ as in \cref{eq:3BT}. In fact, it will again be useful to consider, more generally, the family of operators $T_\lambda$, for $\lambda \in [0,1]$, as introduced in \cref{eq:Zlambda}, with $T_1 = T$ and $T_0 = 1$. 
 
 We define the trial state $\psi^\textnormal{trial} := R_\F T \Omega \in L^2_\textnormal{a} (\Tbb^{3N})$ and the corresponding excitation vector $\xi^\textnormal{trial} := R_\F^* \psi^\textnormal{trial} = T \Omega \in \chi (\cN_\textnormal{h} - \cN_\textnormal{p} = 0) \cF$. Since $R_\F$ and $T$ only create particles with momentum at distance smaller than $R$ from the Fermi surface, and since we assumed $R \ll N^{1/3}$, we have
 \[
 \langle \psi^\textnormal{trial} , \cH_N \psi^\textnormal{trial} \rangle = \langle \psi^\textnormal{trial} , \wt{\cH}_N \psi^\textnormal{trial} \rangle
 \]
 where $\wt{\cH}_N$ is the Hamilton operator \cref{eq:ham-fock}, with $\hat{V} (k)$ replaced by $\hat{V} (k) \chi (|k| \leq CN^{1/3})$. Proceeding as in \cref{sec:corrH}, we find 
 \begin{equation}\label{eq:Etrial}
 \langle \psi^\textnormal{trial} , \wt{\cH}_N \psi^\textnormal{trial} \rangle = E_N^\text{HF} + \langle \xi^\textnormal{trial} \wt{\cH}_\textnormal{corr} \xi^\textnormal{trial} \rangle
 \end{equation} 
 with the Hartree--Fock energy \cref{eq:HF-en} (replacing $\hat{V} (k)$ with $\hat{V} (k) \chi (|k| \leq C N^{1/3})$ does not change the r.\,h.\,s.\ of \cref{eq:HF-en} if $C > 0$ is large enough) and with 
\[
\wt{\cH}_\textnormal{corr} = \bH_0 + \wt{Q}_B + \wt{\cE}_1 + \wt{\cE}_2 + \wt{\bX}
\]
where $\wt{Q}_B$, $\wt{\cE}_1$, $\wt{\cE}_2$, $\wt{\bX}$ denote the operators $Q_B$, $\cE_1$, $\cE_2$, $\bX$, respectively, from \cref{eq:RHR-main} and \cref{eq:errors}, with $\hat{V} (k)$ replaced by $\wt{V} (k) \chi (|k| \leq C N^{1/3})$. 

To estimate the expectation of $\wt{\cH}_\textnormal{corr}$ in the state $\xi^\textnormal{trial}$, we first establish rough bounds on the number of particles and the energy of $\xi^\textnormal{trial}$. 
\begin{lemma}[Bounds for particle number and kinetic energy]\label{lm:N-up}
For every $R > 0$ and $m \in \bN$ there exists $C_{R,m} > 0$ such that 
\begin{equation}
\label{eq:trial-N} 
\langle T_\lambda \Omega , \cN^m T_\lambda \Omega \rangle \leq C_{R,m} \qquad \text{for all $\lambda \in [0,1]$\,.}
\end{equation} 
 Moreover, for every $R > 0$ there exists a constant $C_R < \infty$ such that
 \begin{equation}
 \label{eq:trial-en}
 \langle  T_\lambda \Omega , \bH_0  T_\lambda \Omega \rangle \leq C_R \hbar \qquad \text{for all $\lambda \in [0,1]$\,.}
 \end{equation} 
\end{lemma}
\begin{proof}[Proof of \cref{lm:N-up}]
For \cref{eq:trial-N} we can proceed as in the proof of \cite[Prop.~4.6]{BNPSS0}. The only new aspect is that
we use the assumption \cref{eq:quasiC} together with \cref{eq:K-point} to estimate 
\begin{equation}\label{eq:VR}
\sum_{k \in \Gamma^\textnormal{nor}} \| K (k) \|_\textnormal{HS} \leq C \sum_{|k| \leq R} \hat{V} (k) \leq C R \left( \sum_{|k| < R} \hat{V} (k)^2 |k| \right)^{1/2} \leq C R\;.
\end{equation}
This allows us to show that
\[
\left| \frac{\di}{\di\lambda} \langle T_\lambda \Omega, (\cN+5)^m T_\lambda \Omega \rangle \right| \leq C R  \langle T_\lambda \Omega, (\cN+5)^m T_\lambda \Omega \rangle \;.
\]
By Gr\"onwall's lemma, we conclude that 
\[
\langle T_\lambda \Omega, \cN^m T_\lambda \Omega \rangle \leq e^{C_mR \lambda} \;.
\]

To show \cref{eq:trial-en} we write 
\begin{equation}\label{eq:trial-en1}
\langle T_\lambda \Omega, \bH_0 T_\lambda \Omega \rangle = \langle T_\lambda \Omega, (\bH_0 - \bD_\textnormal{B}) T_\lambda \Omega \rangle + \langle T_\lambda \Omega, \bD_\textnormal{B} T_\lambda \Omega \rangle
\end{equation} 
with the operator $\bD_\textnormal{B}$ introduced in \cref{eq:H0-comm2}. From \cref{lm:T-inv} and \cref{eq:trial-N}, we find
\begin{equation}\label{eq:H0-D-up}
| \langle T_\lambda \Omega, (\bH_0 - \bD_\textnormal{B}) T_\lambda \Omega \rangle | \leq C_R \hbar \big( M^{-1/2} + M N^{-2/3+\delta} \big) \;.
\end{equation}
As in the proof of \cref{eq:trial-N} above, the condition $\sum_{k \in \bZ^3} \hat{V} (k) (1+ |k|) < \infty$ required in \cref{lm:T-inv} is now replaced (since $K (k) = 0$ for $|k| > R$) by 
\[
\sum_{|k| < R}  \hat{V} (k) (1+ |k|) \leq C R^2 \left( \sum_k \hat{V} (k)^2 |k| \right)^{1/2} \leq C R^2
\]
which leads (together with \cref{eq:trial-N}) to an $R$--dependent constant in \cref{eq:H0-D-up}. We also have 
\begin{equation}\label{eq:DBest}
\begin{split}
\langle T_\lambda \Omega , \bD_\textnormal{B} T_\lambda \Omega\rangle &\leq C_R \hbar \sum_{k \in \Gamma^{\text{nor}}} \sum_{\alpha = 1}^{M} \langle T_{\lambda} \Omega, c^{*}_{\alpha}(k) c_{\alpha}(k) T_{\lambda} \Omega \rangle \leq C_{R} \hbar \langle T_{\lambda} \Omega, \mathcal{N} T_{\lambda} \Omega \rangle \leq C_{R} \hbar\;,
\end{split}
\end{equation}
where we used \cref{eq:c-N0} in the second and \cref{eq:trial-N} in the third inequality. Inserting \cref{eq:H0-D-up} and \cref{eq:DBest} in \cref{eq:trial-en1}, we obtain \cref{eq:trial-en}. This concludes the proof of \cref{lm:N-up}.
\end{proof}

To estimate the potential energy we need the following lemma, which shows that, when computing expectation values in $\xi^\textnormal{trial}$, we can effectively cutoff the interaction $\hat{V}$ to momenta $|k| \leq R$, up to negligible errors. This observation relies on the fact that $T$ only creates particle--hole pairs with pair momentum $\lvert k\rvert \leq R$.
\begin{lemma}[Control of the high--momentum cutoff] \label{lm:bb-R}
Assume $\sum_{k \in \bZ^3} |k| \hat{V} (k)^2 < \infty$. Then for every $R > 0$ there exists $C_R > 0$ such that 
\begin{equation}
\label{eq:V>R}
\begin{split} 
\frac{1}{N} \sum_{k \in \bZ^3 : R < |k| \leq C N^{1/3}} \hat{V} (k) \langle T \Omega, b^* (k) b (k) T \Omega \rangle \leq C_R M^{3/2}  N^{-1/2 + \delta/2}\;, \\
\Big| \frac{1}{N} \sum_{k \in \bZ^3 : R < |k| \leq C N^{1/3}} \hat{V} (k) \langle T \Omega, b (k) b (-k) T \Omega \rangle \Big| \leq C_R M^{3/2}  N^{-1/2 + \delta/2}  \;.
\end{split} 
\end{equation}
\end{lemma}
\begin{proof}[Proof of \cref{lm:bb-R}]
Consider the second inequality in \cref{eq:V>R}. We write 
\begin{align*}
& \frac{1}{N} \sum_{R < |k| \leq C N^{1/3}} \hat{V} (k) \langle T\Omega, b (k) b (-k) T \Omega \rangle \tagg{eq:bbup1}\\
& = \frac{1}{N} \sum_{R < |k| \leq C N^{1/3}} \hat{V} (k)  \sum_{k' \in \Gamma^\textnormal{nor}} \sum_{\alpha, \beta \in \cI_{k'}} K_{\alpha,\beta} (k')   \int_0^1 \di\lambda \, \Big\langle T_\lambda \Omega, \Big[ c_\alpha^* (k') c^*_\beta (k') , b (k) b (-k) \Big]  T_\lambda \Omega \Big\rangle\;.
\end{align*}
We compute
\begin{equation}\label{eq:comm-up2}
\begin{split}
\Big[ c_\alpha^* (k') c^*_\beta (k') , b (k) b (-k) \Big]  &= c^{*}_{\alpha}(k') b(k) [ c^{*}_{\beta}(k'), b(-k) ] + c^{*}_{\alpha}(k') [ c^{*}_{\beta}(k'), b(k) ] b(-k) \\
& \quad + b(k) [ c^{*}_{\alpha}(k') , b(-k) ] c^{*}_{\beta}(k') + [ c^{*}_{\alpha}(k') , b(k) ] b(-k) c^{*}_{\beta}(k')\;.
\end{split}
\end{equation}
We consider the case $\alpha, \beta \in \mathcal{I}^{+}_{k'}$ (so that $c_\alpha^* (k') = b_\alpha^* (k')$ and $c_\beta^* (k') = b_\beta^* (k')$ by \cref{eq:calpha}); the other cases can be studied in the same way. We find 
\begin{equation}
\label{eq:cstarb}
[ c^{*}_{\alpha}(k'), b(k) ] = \frac{1}{n_{\alpha}(k')} \sum_{\substack{p \in B_F^c \cap B_\alpha : \\ p-k' \in B_F \cap B_\alpha}} \sum_{q \in B_F^c \cap B_F + k} \big( \delta_{p,q} \delta_{k,k'} -\delta_{p,q} a_{p-k'}^* a_{q-k} - \delta_{p-k', q-k} a_p^* a_q \big)\;.
\end{equation}
Thanks to the constraint $|k'| < R < |k|$, the otherwise dominant contribution due to $\delta_{p,q} \delta_{k,k'}$ vanishes. For such $k$ and $k'$ and for any $\psi, \varphi \in \mathcal{F}$ we obtain 
\begin{equation}
|\langle \varphi, [ c^{*}_{\alpha}(k'), b(k) ] \psi \rangle| \leq \frac{C}{n_{\alpha}(k')} \| \mathcal{N}^{1/2} \varphi \| \| \mathcal{N}^{1/2} \psi \|\;.
\end{equation}
We can use this estimate to bound all the contributions to \cref{eq:bbup1} arising from the various terms in the r.\,h.\,s.\ of \cref{eq:comm-up2}. For instance, consider the first. Using \cref{eq:cstarb} we have
\begin{align*}
| \langle T_{\lambda} \Omega , c^{*}_{\alpha}(k') b(k) [ c^{*}_{\beta}(k'), b(-k) ]   T_{\lambda} \Omega \rangle | &\leq  \frac{C}{n_{\beta}(k')} \| \mathcal{N}^{1/2} b^{*}(k) c_{\alpha}(k') T_{\lambda} \Omega  \|  \| \mathcal{N}^{1/2} T_{\lambda} \Omega \| \\
&= \frac{C}{n_{\beta}(k')}  \|  b^{*}(k) c_{\alpha}(k') \mathcal{N}^{1/2}T_{\lambda} \Omega  \| \| \mathcal{N}^{1/2} T_{\lambda} \Omega \|\;.
\end{align*}
\Cref{lem:counting}, together with the assumption $\alpha, \beta \in \cI_{k'}$, implies that $n_{\beta} (k') \geq C N^{1/3-\delta/2} M^{-1/2}$. Next, we will use the bounds
\begin{equation}
\| b^{\natural}(k) \varphi \| \leq C|k|^{1/2} N^{1/3} \| (\mathcal{N} + 1)^{1/2} \varphi \|\;,\qquad \| c^{\natural}_{\alpha}(k')  \varphi\| \leq C\| (\mathcal{N}+1)^{1/2} \varphi \|\;,
\end{equation}
where $b^\natural$ is either $b$ or $b^*$, and analogously for $c^\natural$. Here, the first estimate follows from \cite[Eqs.~(4.12) and (4.13)]{BNPSS2} (observing that $|B_F^c \cap B_F + k| \leq C |k| N^{2/3}$), the second from \cref{lm:c-N} (using the inequality $[c_\alpha (k'), c_\alpha^* (k')] \leq 1$; see \cite[Eq.~(5.10)]{BNPSS}). Thus
\begin{equation}
| \langle T_{\lambda} \Omega , c^{*}_{\alpha}(k') b(k) [ c^{*}_{\beta}(k'), b(-k) ]   T_{\lambda} \Omega \rangle | \leq C|k|^{1/2} N^{\delta/2} M^{1/2} \langle T_{\lambda} \Omega, (\mathcal{N} + 1)^{3} T_{\lambda} \Omega \rangle\;.
\end{equation}
All the other contributions in \cref{eq:comm-up2} can be estimated in a similar way. We get, using the bounds $|K_{\alpha,\beta}(k')| \leq \hat V(k')/M$ and \cref{eq:trial-N}, 
\begin{equation}
\begin{split}
& \Big| \frac{1}{N} \sum_{k \in \bZ^3 : R < |k| \leq C N^{1/3}} \hat{V} (k)  \langle T \Omega, b(k) b(-k) T \Omega \rangle \Big| \\ &\leq C_{R} M^{3/2} N^{-1 + \delta/2}  \sum_{R < |k| \leq C N^{1/3}} |k|^{1/2}\hat{V} (k)  \sum_{k' \in \Gamma^\textnormal{nor}} \hat V(k') \\  &\leq C_{R}
M^{3/2} N^{-1/2 + \delta/2} 
\end{split}
\end{equation}
where the sum over $k'$ has been absorbed in the constant $C_R$ (recall that $|k'| < R$ in 
$\Gamma^\textnormal{nor}$) and where we estimated  
\begin{equation}
\sum_{k: |k| \leq CN^{1/3}} |k|^{1/2}\hat{V} (k) \leq CN^{1/2} \left(\sum_{k} |k| \hat V(k)^{2}\right)^{1/2} \leq CN^{1/2}\; .
\end{equation}
This concludes the proof of the second inequality in \cref{eq:V>R}. The first can be shown similarly; we omit the details. This concludes the proof of \cref{lm:bb-R}.
\end{proof} 
 
With \cref{lm:N-up} and \cref{lm:bb-R}, we can go back to the computation of the expectation value on the r.\,h.\,s.\  of \cref{eq:Etrial}. We control the expectation  of the error term 
$\wt{\cE}_1$ with the bound  
\[
\| \wt{\cE}_1 \xi \| \leq \frac{C \| \hat{V} \|_1}{N} \| \cN^2 \xi \|
\]
established in \cite[Eq.~(4.10)]{BNPSS2}. With \cref{eq:trial-N} and estimating 
\[
\sum_{|k| \leq C N^{1/3}} \hat{V} (k) \leq C N^{1/3} \Big(\sum_{k \in \bZ^3} \hat{V} (k)^2 |k| \Big)^{1/2}  \leq C N^{1/3}
\]
we find, for a constant $C_R$ depending on the cutoff $R > 0$,
\[ 
\langle \xi^\textnormal{trial}, \wt{\cE}_1 \xi^\textnormal{trial} \rangle \leq C_R N^{-2/3}\;.
\]
The expectation value of $\wt{\cE}_2$ in our trial state vanishes for parity reasons exactly as in \cite[Lemma~5.2]{BNPSS0}. 

Applying \cref{lm:bb-R} and \cref{eq:V>R} and using the fact that $\bX \leq 0$, from \cref{eq:Etrial} we get
\[
\langle \psi^\textnormal{trial} , \cH_N \psi^\textnormal{trial} \rangle \leq E_N^\text{HF} + \langle \xi^\textnormal{trial}, (\bH_0 + \wt{Q}_B^R) \xi^\textnormal{trial} \rangle + C_R N^{-2/3} + C_R M^{3/2} N^{-1/2 + \delta/2} 
\] 
where we defined 
\[
\wt{Q}_B^R := \frac{1}{N} \sum_{k \in \bR^3 : |k| \leq R} \hat{V} (k) \Big( b^* (k) b(k) + \frac{1}{2} \big(b^* (k) b^* (-k) + b(k) b(-k) \big) \Big) \;.
\]
In order to obtain an upper bound for the expectation of the operator $\bH_0 + \wt{Q}_B^R$, we proceed as in the proof of \cref{thm:main}, now with $\hat{V} (k)$ replaced everywhere by $\hat{V} (k) \chi (|k| \leq R)$. We conclude that 
\[
\begin{split}
& \langle \psi^\textnormal{trial} , \cH_N \psi^\textnormal{trial} \rangle \\ \leq &\, E_N^\text{HF} + \hbar \kappa_0 \sum_{|k| \leq R} |k| \left( \frac{1}{\pi} \int_0^\infty \log \left( 1 + 2\pi \kappa_0 \hat{V} (k) \Big(1-\lambda \arctan \big(\frac{1}{\lambda}\big) \Big) \right) \di\lambda - \frac{\pi}{2} \kappa_0 \hat{V} (k) \right) \\
&+ C_R N^{-2/3} + C_R N^{-1/3} M^{-1/2} + C_R M^{3/2} N^{-1/2+\delta/2} \\
\leq &\, E_N^\text{HF} + E^\text{RPA}_N + C \sum_{|k| > R} \hat{V} (k)^2 |k| + C_R \Big( N^{-2/3} + N^{-1/3} M^{-1/2} + M^{3/2} N^{-1/2 + \delta/2} \Big)\;.
\end{split}
\]
Fixing $M = N^\alpha$, choosing $\alpha > 0$ small enough and then $R = R(N)$ so that $R(N) \to \infty$ as $N \to\infty$ at a sufficiently slow pace, we obtain \cref{eq:up-app}. This concludes the proof of the generalized RPA upper bound, \cref{thm:impr-up}. \end{proof}

\section*{Acknowledgements} 
RS was supported by the European Research Council under the European Union’s Horizon 2020 research and innovation programme (grant agreement No.~694227). MP acknowledges financial support from the European Research Council under the European Union’s Horizon 2020 research and innovation programme (ERC StG MaMBoQ, grant agreement No.~802901). BS acknowledges financial support from the NCCR SwissMAP, from the Swiss National Science Foundation through the Grant ``Dynamical and energetic properties of Bose-Einstein condensates'' and from the European Research Council through the ERC AdG CLaQS (grant agreement No.~834782). NB was supported by Gruppo Nazionale per la Fisica Matematica (GNFM) of Italy and the European Research Council's Starting Grant \textsc{FermiMath} (grant agreement No.~101040991).
 
\section*{Competing Interests}
The authors have no competing interests to declare that are relevant to the content of this article.

\section*{Data Availability}
Data sharing not applicable to this article as no datasets were generated or analysed during the current study.


\end{document}